\documentclass{article}

\usepackage{geometry}
\usepackage{amsmath}
\usepackage{amssymb}
\usepackage{amsthm}
\usepackage{graphicx}
\usepackage{subcaption}
\usepackage[round]{natbib}
\usepackage{hyperref}
\hypersetup{
	colorlinks=true,
	linkcolor=blue,
	citecolor=blue,
}
\usepackage{xcolor}
\usepackage{enumerate}
\usepackage{accents}
\usepackage{booktabs}

\theoremstyle{plain}
\newtheorem{definition}{Definition}
\newtheorem{example}{Example}

\newtheorem{proposition}{Proposition}
\newtheorem{corollary}{Corollary}
\newtheorem{lemma}{Lemma}
\newtheorem{lemmaappendix}{Lemma}[section]
\newtheorem{remark}{Remark}

\title{
	Causal Non-causal State Space Models and the Modelling of Financial Bubbles
}
\author{
	\normalsize{Frederik Bjerg Krabbe}\footnote{I am grateful to Jean-Michel Zakoïan, Leopoldo Catania, Mathias Barding, Tom Engsted, Martin Møller Andreasen, and Christian Gourieroux, seminar participants at Aarhus University, and conference participants at the 2025 Virtual Workshop for Junior Researchers in Time Series and the 2025 Workshop on Non-causal Econometrics for comments and discussions. This research is supported by the Danish National Research Foundation (DNRF Chair grant number DNRF154).} \\
	\small{Aarhus University} \\
}
\date{}

\allowdisplaybreaks

\begin{document}
	
\maketitle

\begin{abstract}
	
	In this paper, we study causal non-causal state space models to model time series characterised by a local explosive increase followed by a sharp decrease such as stock prices. To motivate the use of causal non-causal state space models, we show that the causal non-causal convolution autoregressive model introduced by \cite{GourierouxZakoian2017} can be consistent with the rational expectations stock price model. As in a causal state space model, a central question is how to perform state and parameter inference in the causal non-causal state space model, which we discuss in the paper. We also study the causal non-causal convolution autoregressive model in more detail, providing some new results for the model. To illustrate the usefulness of causal non-causal state space models, we use the causal non-causal convolution autoregressive model to estimate the size of the dot-com bubble in both real time and a posteriori with the stable non-causal autoregressive model considered also by \cite{GourierouxZakoian2017} as a benchmark.
	
\end{abstract}

\section{Introduction} \label{sec:introduction}

Many economic and financial time series are characterised by a local explosive increase followed by a sharp decrease. One example is stock prices where the aforementioned phenomenon can be the result of a bubble that bursts.

In recent years, there has been an increasing interest in mixed causal non-causal (vector) autoregressive models to model such time series. Indeed, mixed causal non-causal autoregressive models have by now successfully been applied to model stock prices, commodity prices, inflation, and cryptocurrencies, see, for instance, \cite{LanneSaikkonen2011}, \cite{LanneLuotoSaikkonen2012}, \cite{LanneLuoto2013}, \cite{GourierouxJasiak2016}, \cite{LanneLuoto2017}, \cite{GourierouxZakoian2017}, \cite{FriesZakoian2019}, \cite{HecqIsslerTelg2020}, \cite{CavaliereNielsenRahbek2020}, and \cite{Fries2022}. In the multivariate case, mixed causal non-causal vector autoregressive models have been applied to model (i) gross domestic product (GDP), total government expenditure, and total government revenue, (ii) commodity prices, and (iii) interest rates, see, for instance, \cite{LanneSaikkonen2013}, \cite{GourierouxJasiak2016}, \cite{GourierouxJasiak2017}, and \cite{DavisSong2020}. 

Causal non-causal state space models, however, have attracted much less attention. To the best of our knowledge, the only example of a causal non-causal state space model is the causal non-causal convolution autoregressive model introduced by \cite{GourierouxZakoian2017} and further studied by \cite{GourierouxJasiakTong2021} to model stock prices. In the causal non-causal convolution autoregressive model, the observed component is equal to the sum of two unobserved components, a causal Gaussian autoregressive model and a non-causal stable autoregressive model, the idea being that the stock price is equal to the fundamental value of the stock plus a potential bubble. \cite{GourierouxJasiakTong2021} argued that "so far, the causal non-causal convolution (autoregressive) model has not been widely used in empirical research, likely due to the lack of a filtering and forecasting algorithm" \cite[p.~1230]{GourierouxJasiakTong2021}, so they proposed a partial filter. The partial filter they proposed is, however, not optimal as they also acknowledge. Indeed, the partial filter at time $t$ uses information only at time $t-1$, $t$, and $t+1$, and not from the beginning of the sample to the end of it. And although mathematically, it is possible to extend the partial filter at time $t$ such that it uses information from time $t-h$ to $t+h$ and even from the beginning of the sample to the end of it, it is not possible to compute it numerically, leaving a gap in the literature.

In this paper, we thus study causal non-causal state space models. To motivate the use of causal non-causal state space models, we show that the causal non-causal convolution autoregressive model can be consistent with the arguably most fundamental model of stock prices, namely, the rational expectations stock price model, thereby bridging the literature on mixed causal non-causal (vector) autoregressive models with the one on rational expectations stock price models. This result generalises the one by \cite{GourierouxJasiakMonfort2020} by showing that the non-causal stable autoregressive model considered by \cite{GourierouxZakoian2017} can also be consistent with the rational expectations stock price model and not only the rational expectations model. As in a causal state space model, a central question is how to perform state and parameter inference in the causal non-causal state space model; here, state inference refers to filtering, prediction, and smoothing. If the transition kernels of the non-causal Markov process have densities, we show that the model can be transformed into a causal state space model, so one can perform inference by applying standard techniques. Although the model can still be transformed into a causal state space model in theory if the transition kernels of the non-causal Markov process do not have densities, it is not useful to do so in practice. However, if the transition kernels of the causal Markov process have densities, we show that one can still perform inference by transforming the model into a non-causal state space model. We also study the causal non-causal convolution autoregressive model in more detail, providing some new results for the model.

Many authors have found evidence of a bubble in the NASDAQ index in the 1990s, which is commonly referred to as the dot-com bubble, see, for instance, \cite{PhillipsWuYu2011} and the references therein. To illustrate the usefulness of causal non-causal state space models, we use the causal non-causal convolution autoregressive model to estimate the size of the dot-com bubble in both real time and a posteriori with the non-causal autoregressive model as a benchmark. We find that the causal non-causal convolution autoregressive model is able to detect the dot-com bubble in both real time and a posteriori, but the size of the dot-com bubble estimated in real time is larger than the one estimated a posteriori. The non-causal autoregressive model is also able to detect the dot-com bubble. However, it also detects a bubble before the dot-com bubble, in contrast to the causal non-causal convolution autoregressive model, which contradicts, for instance, the evidence in \cite{PhillipsWuYu2011} and thus shows that the causal non-causal convolution autoregressive model is the more realistic of the two.

The rest of the paper is organised as follows. Section \ref{sec:motivation} shows that the causal non-causal convolution autoregressive model can be consistent with the rational expectations stock price model. Section \ref{sec:cncssm} introduces the causal non-causal state space model and discusses how to perform inference in it. Section \ref{sec:cnccar} discusses the causal non-causal convolution autoregressive model in more detail, and Section \ref{sec:motivation_revisited} demonstrates how the general theory can be applied to the causal non-causal convolution autoregressive model. Section \ref{sec:application} presents the application. Finally, Section \ref{sec:conclusion} concludes. All proofs are collected in the appendix.

\section{A Motivating Example} \label{sec:motivation}

In the following, $X \sim \mathcal{S} (\alpha,\mu,\sigma,\beta)$ denotes a stable distributed random variable where $\alpha \in (0,2]$ is the index of stability, which determines the tails of the distribution in the sense that if $\alpha < 2$, then
\begin{equation*}
	\mathbb{E} \left[ |X|^p \right] < \infty \quad \textup{if} \quad 0 < p < \alpha \quad \textup{and} \quad \mathbb{E} \left[ |X|^p \right] = \infty \quad \textup{if} \quad p \geq \alpha,
\end{equation*}
$\mu \in \mathbb{R}$ is a location parameter, $\sigma \in (0,\infty)$ is a scale parameter, and $\beta \in [-1,1]$ is an asymmetry parameter. In general, the probability density function of $X$ is not available in closed form, but the characteristic function of $X$, $\varphi_X (t) = \mathbb{E} [ \exp(\textup{i}tX) ], t \in \mathbb{R}$, is given by
\begin{equation*}
	\varphi_X (t) = \exp \left( \textup{i} \mu t - \sigma^{\alpha} \left| t \right|^{\alpha} \left( 1 - \textup{i} \beta \textup{sign} (t) \Phi \right) \right),
\end{equation*}
with $\Phi = \tan \left( \frac{\pi \alpha}{2} \right)$ if $\alpha \neq 1$ and $\Phi = - \frac{2}{\pi} \log \left( \left| t \right| \right)$ if $\alpha = 1$. Exemptions are when $\alpha = 2$ and $\beta = 0$, $\alpha = 1$ and $\beta = 0$, and $\alpha = \frac{1}{2}$ and $\beta = 1$ where the stable distribution reduces to the normal, Cauchy, and Lévy distribution, respectively. See, for instance, \cite{SamorodnitskyTaqqu1994} for more details on stable distributed random variables.

To model bubbles from a probabilistic point of view, \cite{GourierouxZakoian2017} considered the non-causal autoregressive model
\begin{equation}
	Y_t = \rho Y_{t+1} + \varepsilon_t, \quad \varepsilon_t \overset{i.i.d.}{\sim} \mathcal{S} (\alpha,0,\sigma,\beta), \label{eq:ncar} 
\end{equation}
where $| \rho | < 1$. The non-causal autoregressive model is indeed able to generate bubbles: Assume that a large value of $\varepsilon_t$ is realised at time $t = \tau$. Then, in reverse time, $Y_t$ will jump at time $t = \tau$ and decrease exponentially such that, in direct time, $Y_t$ will increase exponentially and crash at time $t = \tau$; see Section 2 in \cite{GourierouxZakoian2017} for more details. 

Figure \ref{fig:ncar} shows a simulation from the non-causal autoregressive model in Equation \eqref{eq:ncar} with $\rho = 0.9$, $\alpha = 1$, $\sigma = 0.5$, and $\beta = 0$.

\begin{figure}[htbp]
	
	\centering
	
	\includegraphics[width=0.75\textwidth,height=0.25\textheight]{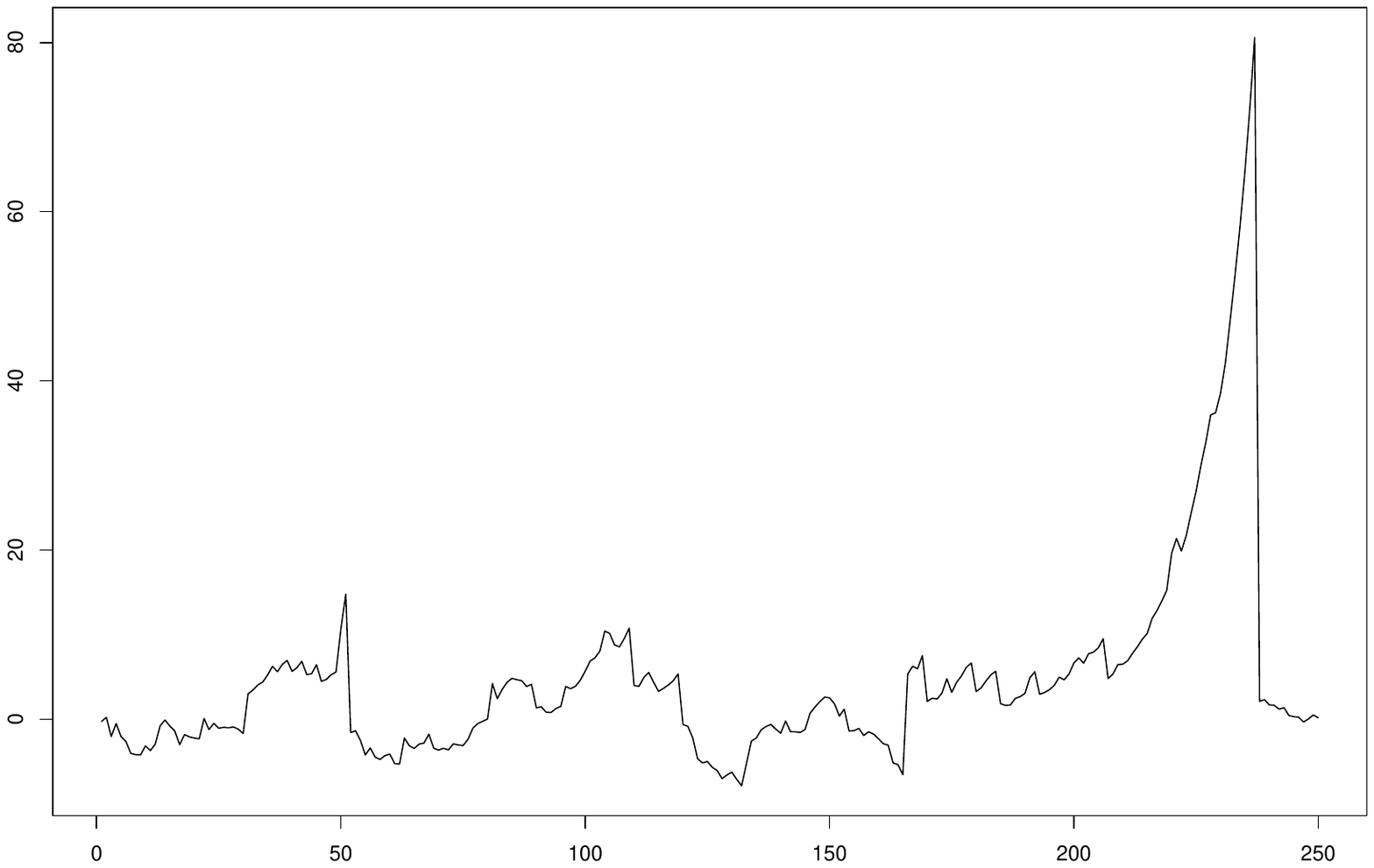}
	
	\caption{A simulation from the non-causal autoregressive model in Equation \eqref{eq:ncar} with $\rho = 0.9$, $\alpha = 1$, $\sigma = 0.5$, and $\beta = 0$.}
	
	\label{fig:ncar}
	
\end{figure}

The arguably most fundamental model of stock prices is the rational expectations stock price model; a textbook treatment of it can be found in Chapter 7 in \cite{CampbellLoMacKinlay1997}. In the rational expectations stock price model, the price $P_t$ of a stock is given by
\begin{equation}
	P_t = \frac{\mathbb{E}_t [P_{t+1}+D_{t+1}]}{1+R}, \label{eq:price_equation}
\end{equation}
where $R > 0$ is the expected return on the stock, which we assume is constant, $\mathbb{E}_t [\cdot]$ is the expectation conditional on the information at time $t$, and $D_t$ is the dividend from the stock.

For simplicity, we assume that the expected return is constant since, in this case, the rational expectations stock price model is linear (Equation \eqref{eq:price_equation}) and has therefore a closed-form solution (Equations \eqref{eq:price_solution1}-\eqref{eq:price_solution3}). Many authors have, however, argued that the expected return is both time-varying and persistent, see, for instance, \cite{FamaFrench1988}, \cite{CampbellCochrane1999}, \cite{FersonSarkissianSimin2003}, and \cite{PastorStambaugh2009}, in which case the rational expectations stock price model is non-linear and does therefore not have a closed-form solution. A solution to this problem is to log-linearise the model as in \cite{CampbellShiller1988}, see also Chapter 7 in \cite{CampbellLoMacKinlay1997}, such that it can be 'solved'. We, however, continue to assume that the expected return is constant in line with the majority of the rational expectations stock price models with bubbles in the literature.

The solution to Equation \eqref{eq:price_equation} is
\begin{equation}
	P_t = F_t + B_t, \label{eq:price_solution1}
\end{equation}
where $F_t$ is the fundamental value of the stock given by
\begin{equation}
	F_t = \sum_{i=1}^{\infty} \left( \frac{1}{1+R} \right)^i \mathbb{E}_t [D_{t+i}] \label{eq:price_solution2}
\end{equation}
and $B_t$ is a rational bubble given by
\begin{equation}
	B_t = \frac{\mathbb{E}_t [B_{t+1}]}{1+R} \label{eq:price_solution3}
\end{equation}
which cannot be negative, cannot burst completely and start again so if it exists in the stock today, then it must have existed since trading in the stock began, and cannot be positive if the price is bounded from above, see Chapter 7 in \cite{CampbellLoMacKinlay1997} and the references therein for more details.

In the case with no bubble, it is common to assume that the price
\begin{equation*}
	P_t = \sum_{i=1}^{\infty} \left( \frac{1}{1+R} \right)^i \mathbb{E}_t [D_{t+i}]
\end{equation*}
is non-stationary but that the spread
\begin{equation*}
	P_t - \frac{D_t}{R} = \frac{1}{R} \sum_{i=0}^{\infty} \left( \frac{1}{1+R} \right)^i \mathbb{E}_t [\Delta D_{t+1+i}]
\end{equation*}
is stationary making the spread more convenient to model than the price. Like the expected return, many authors have also argued that the expected dividend growth is both time-varying and persistent, see, for instance, \cite{BansalYaron2004}, \cite{MenzlySantosVeronesi2004}, and \cite{LettauLudvigson2005}. Assume therefore, in the same vein as \cite{VanBinsbergenKoijen2010} and \cite{AndreasenBro2024}, that the expected dividend growth $G_t := \mathbb{E}_t [\Delta D_{t+1}]$ follows a first-order Gaussian autoregressive model
\begin{equation*}
	G_t - \phi_0^{G} = \phi_1^{G} (G_{t-1} - \phi_0^{G}) + \varepsilon_t^{G}, \quad \varepsilon_t^{G} \overset{i.i.d.}{\sim} \mathcal{N} (0,\sigma_{G}^2),
\end{equation*}
with $|\phi_1^{G}| < 1$, see also \cite{Cochrane2008}. Returning to the case with a bubble, the spread $P_t^S := P_t - \frac{D_t}{R}$ is then given by
\begin{equation}
	P_t^S = F_t^S + B_t, \label{eq:rem1}
\end{equation}
where $F_t^S := F_t - \frac{D_t}{R}$ is given by
\begin{equation}
	F_t^S = \phi_0^{FS} + \phi_1^{FS} F_{t-1}^S + \varepsilon_t^{FS}, \quad \varepsilon_t^{FS} \overset{i.i.d.}{\sim} \mathcal{N} (0,\sigma_{FS}^2), \label{eq:rem2}
\end{equation}
with 
\begin{equation*}
	\phi_0^{FS} := \frac{1}{R} \frac{1+R}{R} \phi_0^{G} (1-\phi_1^{G}), \quad \quad \phi_1^{FS} := \phi_1^{G}, \quad \quad \sigma_{FS}^2 := \left( \frac{1}{R} \frac{1+R}{1+R-\phi_1^{G}} \right)^2 \sigma_{G}^2
\end{equation*}
and $B_t$ is given by
\begin{equation}
	\mathbb{E}_t [B_{t+1}] = (1+R) B_t \label{eq:rem3}
\end{equation}
as above, see Appendix \ref{appendix:motivatingexample} for details.

\cite{GourierouxZakoian2017} also briefly considered a causal non-causal state space model namely the causal non-causal convolution autoregressive model
\begin{equation}
	Y_t = X_t^{F} + X_{t}^{B}, \label{eq:ncssm1}
\end{equation}
where
\begin{equation}
	X_t^{F} = \mu_F + \rho_F X_{t-1}^{F} + \varepsilon_t^{F}, \quad \varepsilon_t^F \overset{i.i.d.}{\sim} \mathcal{N} (0,\sigma_F^2), \label{eq:ncssm2}
\end{equation}
with $| \rho_F | < 1$ and 
\begin{equation}
	X_t^{B} = \rho_B X_{t+1}^{B} + \varepsilon_t^{B}, \quad \varepsilon_t^B \overset{i.i.d.}{\sim} \mathcal{S} (\alpha,0,\sigma_B,\beta), \label{eq:ncssm3}
\end{equation}
with $| \rho_B | < 1$ where $\varepsilon_s^F$ and $\varepsilon_t^B$ are independent for all $s$ and $t$, see also \cite{GourierouxJasiakTong2021}. Although \cite{GourierouxZakoian2017} modelled bubbles from a probabilistic point of view as mentioned above, the causal non-causal convolution autoregressive model in Equations \eqref{eq:ncssm1}-\eqref{eq:ncssm3} is in fact consistent with the rational expectations stock price model in Equations \eqref{eq:rem1}-\eqref{eq:rem3}. Indeed, if $\alpha < 1$ and $\beta = 1$, then $X_t^{B} \geq 0$ and, by Proposition \ref{prop:Proposition3}, 
\begin{equation*}
	\mathbb{E} \left[ X_t^B \mid X_{t-1}^B \right] = (1+R) X_{t-1}^B
\end{equation*}
if
\begin{equation}
	\rho_B^{\alpha-1} = 1+R. \label{eq:eqeq}
\end{equation}
Note that the non-causal autoregressive model in Equation \eqref{eq:ncar} with a constant is strictly speaking also consistent with the rational expectations stock price model in Equations \eqref{eq:rem1}-\eqref{eq:rem3} if $\phi_1^{FS} = 0$ and $\sigma_{FS}^2 = 0$ which, however, is unrealistic.

Figure \ref{fig:cncssm} shows a simulation from the causal non-causal convolution autoregressive model in Equations \eqref{eq:ncssm1}-\eqref{eq:ncssm3} with $\mu_F = 5$, $\rho_F = \rho_B = 0.9$, $\sigma_F^2 = 10$, $\alpha = 1$, $\sigma_B = 0.5$, and $\beta = 0$. The top figure shows $Y_t$ (black) and the bottom one shows $Y_t$ together with $X_t^{F}$ (green) and $X_t^{B}$ (red).

\begin{figure}[htbp]
	
	\centering
	
	\begin{subfigure}{1\textwidth}
		\centering
		\includegraphics[width=0.75\textwidth,height=0.25\textheight]{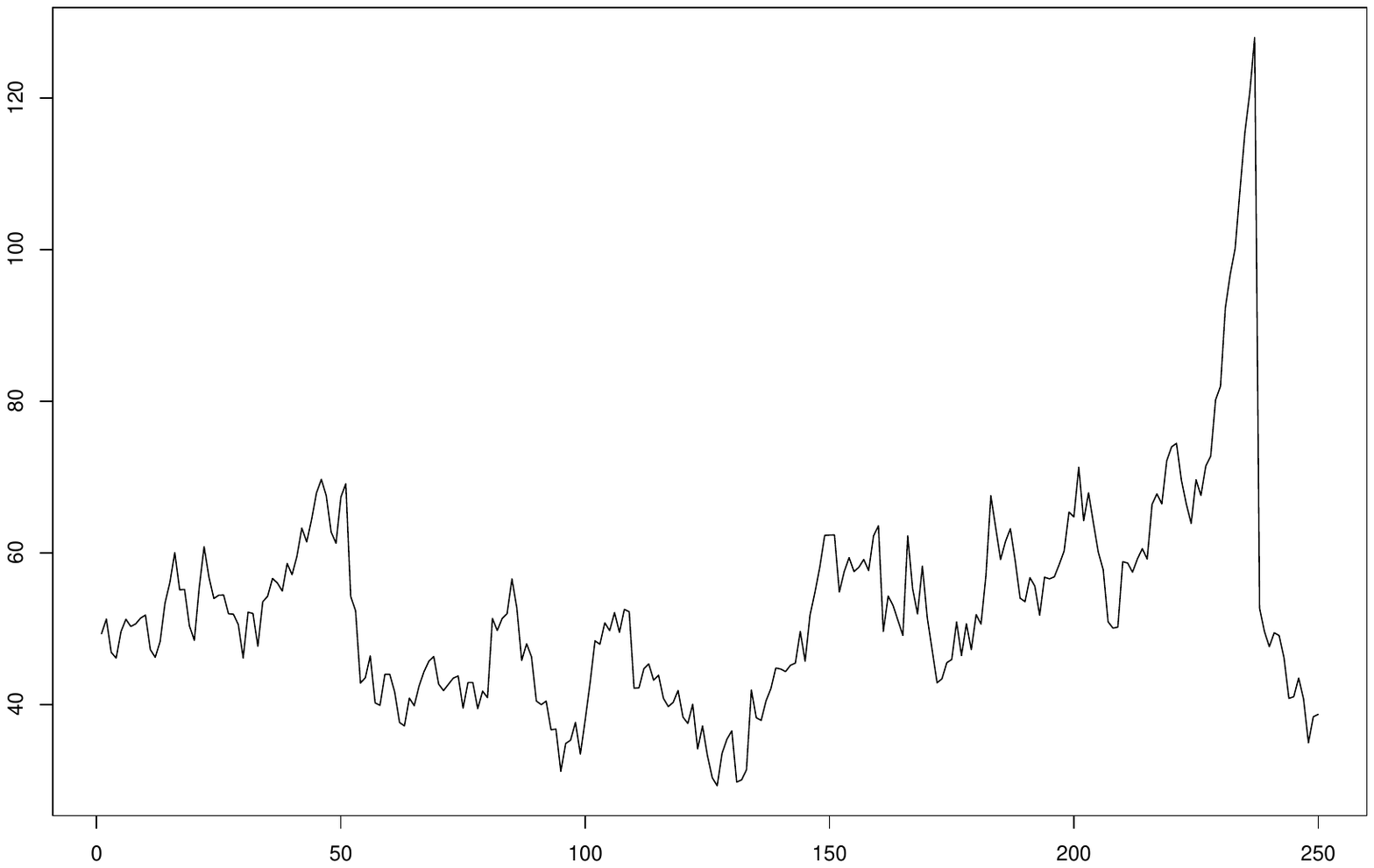}
	\end{subfigure}
	
	\begin{subfigure}{1\textwidth}
		\centering
		\includegraphics[width=0.75\textwidth,height=0.25\textheight]{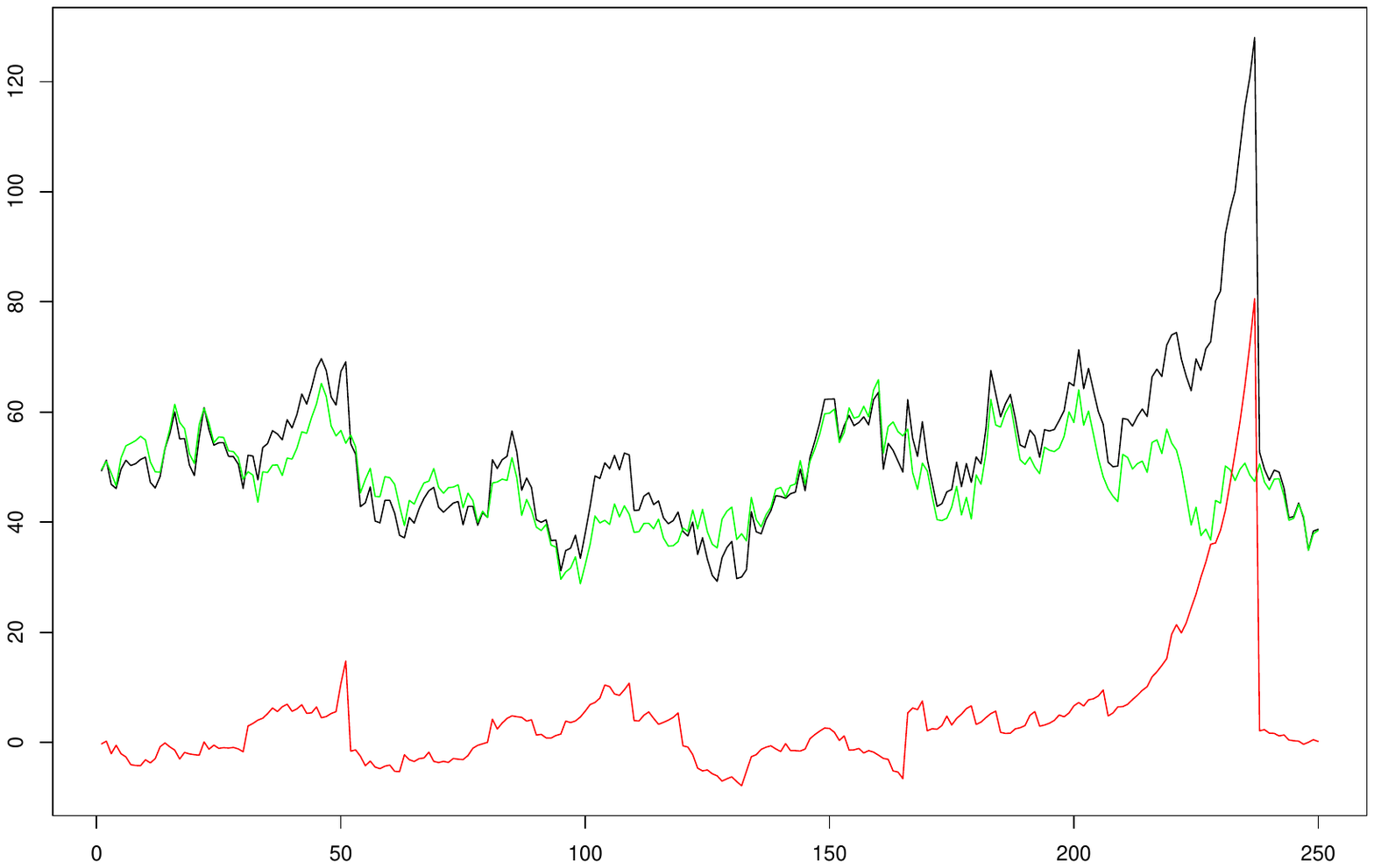}
	\end{subfigure}
	
	\caption{A simulation from the causal non-causal convolution autoregressive model in Equations \eqref{eq:ncssm1}-\eqref{eq:ncssm3} with $\mu_F = 5$, $\rho_F = \rho_B = 0.9$, $\sigma_F^2 = 10$, $\alpha = 1$, $\sigma_B = 0.5$, and $\beta = 0$ where $Y_t$ is black, $X_t^{F}$ is green, and $X_t^{B}$ is red.}
	
	\label{fig:cncssm}
	
\end{figure}

We consider the causal non-causal convolution autoregressive model in more detail in Section \ref{sec:cnccar}.

\section{The Causal Non-causal State Space Model} \label{sec:cncssm}

We now introduce the causal non-causal state space model (Section \ref{sec:model}) and discuss how to perform state and parameter inference in it; first, if the transition kernels of the non-causal Markov process have densities, by transforming it into a causal state space model (Section \ref{sec:causalmodel}), then, if the transition kernels of the non-causal Markov process do not have densities but the ones of the causal Markov process do, by transforming it into a non-causal state space model (Section \ref{sec:noncausalmodel}).

\subsection{The Model} \label{sec:model}

In the following, all random variables are defined on a probability space $(\Omega,\mathcal{F},\mathbb{P})$.

To define the causal non-causal state space model, we need the notion of a transition kernel. A kernel from a measurable space $(\textup{X},\mathcal{X})$ to a measurable space $(\textup{Y},\mathcal{Y})$ is a function $P : \textup{X} \times \mathcal{Y} \rightarrow [0,\infty]$ such that (i) for all $x \in \textup{X}$, $A \mapsto P(x,A)$ is a measure and (ii) for all $A \in \mathcal{Y}$, $x \mapsto P(x,A)$ is measurable. If, in addition, $P(x,\textup{Y}) = 1$ for all $x \in \textup{X}$, then the kernel is called a transition kernel.

We define a causal non-causal state space model as follows.\footnote{We call the model a causal non-causal state space model and \textit{not} a mixed causal non-causal state space model to distinguish it from the mixed causal non-causal autoregressive model mentioned in the introduction.}
\begin{definition} \label{def:cncssm}
	A stochastic process $((\accentset{\rightarrow}{X}_t,\accentset{\leftarrow}{X}_t,Y_t))_{t=1}^{T}$ on $(\accentset{\rightarrow}{\textup{X}} \times \accentset{\leftarrow}{\textup{X}} \times \textup{Y},\accentset{\rightarrow}{\mathcal{X}} \otimes \accentset{\leftarrow}{\mathcal{X}} \otimes \mathcal{Y})$ is called a causal non-causal state space model if it satisfies the following conditions. (i) $(\accentset{\rightarrow}{X}_t)_{t=1}^{T}$ is a causal Markov process on $(\accentset{\rightarrow}{\textup{X}},\accentset{\rightarrow}{\mathcal{X}})$ with transition kernels
	\begin{equation*}
		\mathbb{P} ( \accentset{\rightarrow}{X}_t \in \cdot \mid \accentset{\rightarrow}{X}_{t-1} ) = \accentset{\rightarrow}{P}_t (\accentset{\rightarrow}{X}_{t-1},\cdot)
	\end{equation*}
	and initial distribution
	\begin{equation*}
		\mathbb{P} ( \accentset{\rightarrow}{X}_1 \in \cdot) = \accentset{\rightarrow}{\Pi}_1 (\cdot).
	\end{equation*}
	(ii) $(\accentset{\leftarrow}{X}_t)_{t=1}^{T}$ is a non-causal Markov process on $(\accentset{\leftarrow}{\textup{X}},\accentset{\leftarrow}{\mathcal{X}})$ with transition kernels
	\begin{equation*}
		\mathbb{P} ( \accentset{\leftarrow}{X}_t \in \cdot \mid \accentset{\leftarrow}{X}_{t+1} ) = \accentset{\leftarrow}{P}_t (\accentset{\leftarrow}{X}_{t+1},\cdot)
	\end{equation*}
	and initial distribution
	\begin{equation*}
		\mathbb{P} ( \accentset{\leftarrow}{X}_T \in \cdot) = \accentset{\leftarrow}{\Pi}_T (\cdot).
	\end{equation*}
	(iii) $(\accentset{\rightarrow}{X}_t)_{t=1}^{T}$ and $(\accentset{\leftarrow}{X}_t)_{t=1}^{T}$ are independent. (iv) Conditional on $((\accentset{\rightarrow}{X}_t,\accentset{\leftarrow}{X}_t))_{t=1}^{T}$, $(Y_t)_{t=1}^{T}$ is a sequence of independent random variables on $(\textup{Y},\mathcal{Y})$ with transition kernels
	\begin{equation*}
		\mathbb{P} ( Y_t \in \cdot \mid \accentset{\rightarrow}{X}_1,\accentset{\leftarrow}{X}_1,...,\accentset{\rightarrow}{X}_T,\accentset{\leftarrow}{X}_T) = \Phi_t ((\accentset{\rightarrow}{X}_t,\accentset{\leftarrow}{X}_t),\cdot).
	\end{equation*}
	The causal non-causal state space model is called a causal state space model if $((\accentset{\rightarrow}{X}_t,\accentset{\leftarrow}{X}_t,Y_t))_{t=1}^{T} = ((\accentset{\rightarrow}{X}_t,Y_t))_{t=1}^{T}$ and a non-causal state space model if $((\accentset{\rightarrow}{X}_t,\accentset{\leftarrow}{X}_t,Y_t))_{t=1}^{T} = ((\accentset{\leftarrow}{X}_t,Y_t))_{t=1}^{T}$.
\end{definition}

The causal non-causal convolution autoregressive model in Equations \eqref{eq:ncssm1}-\eqref{eq:ncssm3} is indeed a causal non-causal state space model as the following example shows where, for simplicity, $\alpha = 1$ and $\beta = 0$.
\begin{example} \label{ex:un}
	Let $((X_t^F,X_t^B,Y_t))_{t=1}^{T}$ be a stochastic process on $(\mathbb{R}^3,\mathcal{B}(\mathbb{R}^3))$ given by
	\begin{equation*}
		Y_t = X_t^F + X_t^B,
	\end{equation*}
	where
	\begin{equation*}
		X_t^F = \mu_F + \rho_F X_{t-1}^F + \varepsilon_t^{F}, \quad \varepsilon_t^F \overset{i.i.d.}{\sim} \mathcal{N} (0,\sigma_F^2),
	\end{equation*}
	with $| \rho_F | < 1$ and $X_1^F \sim \mathcal{N} \left( \frac{\mu_F}{1-\rho_F},\frac{\sigma_F^2}{1-\rho_F^2} \right)$ and 
	\begin{equation*}
		X_t^B = \rho_B X_{t+1}^B + \varepsilon_t^{B}, \quad \varepsilon_t^B \overset{i.i.d.}{\sim} \mathcal{C} (0,\sigma_B),
	\end{equation*}
	with $| \rho_B | < 1$ and $X_T^B \sim \mathcal{C} \left( 0,\frac{\sigma_B}{1-|\rho_B|} \right)$ where $\varepsilon_s^F$, $X_1^F$, $\varepsilon_t^B$, and $X_T^B$ are independent for all $s \in \{2,...,T\}$ and $t \in \{1,...,T-1\}$.
	
	This is a causal non-causal state space model. Indeed, (i) $(X_t^F)_{t=1}^{T}$ is a causal Markov process on $(\mathbb{R},\mathcal{B}(\mathbb{R}))$ with transition kernel
	\begin{equation*}
		P^F (x^F,A^F) = \int_{A^F} \frac{1}{\sigma_F} f_{\mathcal{N}} \left( \frac{z - \mu_F - \rho_F x^F}{\sigma_F} \right) \lambda(\textup{d}z), \quad x^F \in \mathbb{R}, A^F \in \mathcal{B}(\mathbb{R}),
	\end{equation*}
	and initial distribution
	\begin{equation*}
		\Pi^F (A^F) = \int_{A^F} \frac{1}{\frac{\sigma_F}{\sqrt{1-\rho_F^2}}} f_{\mathcal{N}} \left( \frac{z-\frac{\mu_F}{1-\rho_F}}{\frac{\sigma_F}{\sqrt{1-\rho_F^2}}} \right) \lambda(\textup{d}z), \quad A^F \in \mathcal{B}(\mathbb{R}),
	\end{equation*}
	where $f_{\mathcal{N}} (z) = \frac{1}{\sqrt{2 \pi}} \exp \left( - \frac{z^2}{2} \right), z \in \mathbb{R}$. (ii) $(X_t^B)_{t=1}^{T}$ is a non-causal Markov process on $(\mathbb{R},\mathcal{B}(\mathbb{R}))$ with transition kernel
	\begin{equation*}
		P^B (x^B,A^B) = \int_{A^B} \frac{1}{\sigma_B} f_{\mathcal{C}} \left( \frac{z - \rho_{B} x^B}{\sigma_B} \right) \lambda(\textup{d}z), \quad x^B \in \mathbb{R}, A^B \in \mathcal{B}(\mathbb{R}),
	\end{equation*}
	and initial distribution
	\begin{equation*}
		\Pi^B (A^B) = \int_{A^B} \frac{1}{\frac{\sigma_B}{1-|\rho_B|}} f_{\mathcal{C}} \left( \frac{z}{\frac{\sigma_B}{1-|\rho_B|}} \right) \lambda(\textup{d}z), \quad A^B \in \mathcal{B}(\mathbb{R}),
	\end{equation*}
	where $f_{\mathcal{C}} (z) = \frac{1}{\pi} \frac{1}{1 + z^2}, z \in \mathbb{R}$. (iii) $(X_t^F)_{t=1}^{T}$ and $(X_t^B)_{t=1}^{T}$ are independent. (iv) Conditional on $((X_t^F,X_t^B))_{t=1}^{T}$, $(Y_t)_{t=1}^{T}$ is a sequence of independent random variables on $(\mathbb{R},\mathcal{B}(\mathbb{R}))$ with transition kernel
	\begin{equation*}
		\Phi ((x^F,x^B),B) = \delta_{x^F+x^B}(B), \quad (x^F,x^B) \in \mathbb{R}^2, B \in \mathcal{B}(\mathbb{R}).  
	\end{equation*}
	Here, $\lambda$ is the Lebesgue measure on $(\mathbb{R},\mathcal{B}(\mathbb{R}))$ and $\delta_{x}$ is the Dirac measure on $(\mathbb{R},\mathcal{B}(\mathbb{R}))$. See Appendix \ref{appendix:otherexamples} for details.
\end{example}

The next lemma, from which the distribution of the causal non-causal state space model follows, will be useful in the following. Here, $Z_{1:T} := (Z_T,Z_{T-1},...,Z_1)$ and $z_{1:T} := (z_T,z_{T-1},...,z_1)$ for $Z = \accentset{\rightarrow}{X},\accentset{\leftarrow}{X},Y$ and $z = \accentset{\rightarrow}{x},\accentset{\leftarrow}{x},y$.
\begin{lemma} \label{lem:cncssm}
	Let $((\accentset{\rightarrow}{X}_t,\accentset{\leftarrow}{X}_t,Y_t))_{t=1}^{T}$ be a causal non-causal state space model. Then,
	\begin{align*}
		\mathbb{E} [ f(\accentset{\rightarrow}{X}_{1:T},\accentset{\leftarrow}{X}_{1:T},Y_{1:T}) ] &= \int_{\accentset{\rightarrow}{\textup{X}}^T} \int_{\accentset{\leftarrow}{\textup{X}}^T} \int_{\textup{Y}^T} f(\accentset{\rightarrow}{x}_{1:T},\accentset{\leftarrow}{x}_{1:T},y_{1:T}) \prod_{t=1}^{T} \Phi_t ((\accentset{\rightarrow}{x}_t,\accentset{\leftarrow}{x}_t),\textup{d}y_t) \\
		& \quad \cdot \prod_{t=T}^{2} \accentset{\rightarrow}{P}_t (\accentset{\rightarrow}{x}_{t-1},\textup{d}\accentset{\rightarrow}{x}_{t}) \accentset{\rightarrow}{\Pi}_1 (\textup{d}\accentset{\rightarrow}{x}_{1}) \prod_{t=1}^{T-1} \accentset{\leftarrow}{P}_t (\accentset{\leftarrow}{x}_{t+1},\textup{d}\accentset{\leftarrow}{x}_{t}) \accentset{\leftarrow}{\Pi}_T (\textup{d}\accentset{\leftarrow}{x}_{T})
	\end{align*}
	for all $f \in \mathcal{M}_{b}((\accentset{\rightarrow}{\textup{X}} \times \accentset{\leftarrow}{\textup{X}} \times \textup{Y})^T)$ where $\mathcal{M}_{b}((\accentset{\rightarrow}{\textup{X}} \times \accentset{\leftarrow}{\textup{X}} \times \textup{Y})^T)$ is the set of bounded, measurable functions from $(\accentset{\rightarrow}{\textup{X}} \times \accentset{\leftarrow}{\textup{X}} \times \textup{Y})^T$ to $\mathbb{R}$.
\end{lemma}

\subsection{The Causal State Space Model} \label{sec:causalmodel}

In theory, any causal non-causal state space model can be transformed into a causal state space model, which follows from the well-known fact that a Markov process in direct time is also a Markov process in reverse time and vice versa. The next proposition formalises this.
\begin{proposition} \label{prop:cncssm}
	Let $((\accentset{\rightarrow}{X}_t,\accentset{\leftarrow}{X}_t,Y_t))_{t=1}^{T}$ be a causal non-causal state space model on $(\accentset{\rightarrow}{\textup{X}} \times \accentset{\leftarrow}{\textup{X}} \times \textup{Y},\accentset{\rightarrow}{\mathcal{X}} \otimes \accentset{\leftarrow}{\mathcal{X}} \otimes \mathcal{Y})$ with transition kernels $\accentset{\rightarrow}{P}_t$, $\accentset{\leftarrow}{P}_t$, and $\Phi_t$ and initial distributions $\accentset{\rightarrow}{\Pi}_1$ and $\accentset{\leftarrow}{\Pi}_T$. Let $(\tilde{P}_t)_{t=2}^{T}$ be a sequence of transition kernels from $(\accentset{\leftarrow}{\textup{X}},\accentset{\leftarrow}{\mathcal{X}})$ to $(\accentset{\leftarrow}{\textup{X}},\accentset{\leftarrow}{\mathcal{X}})$ given by
	\begin{equation}
		\int_{\accentset{\leftarrow}{\textup{X}}} \int_{\accentset{\leftarrow}{\textup{X}}} 1_{C_{t-1}} (z_{t-1}) 1_{C_{t}} (z_{t}) \tilde{P}_{t} (z_{t-1},\textup{d}z_{t}) \accentset{\leftarrow}{\Pi}_{t-1} (\textup{d}z_{t-1}) = \int_{\accentset{\leftarrow}{\textup{X}}} \int_{\accentset{\leftarrow}{\textup{X}}} 1_{C_{t-1}} (z_{t-1}) 1_{C_{t}} (z_{t}) \accentset{\leftarrow}{P}_{t-1} (z_{t},\textup{d}z_{t-1}) \accentset{\leftarrow}{\Pi}_{t} (\textup{d}z_{t}) \label{eq:RK1}
	\end{equation}
	for all $C_{t-1} \in \accentset{\leftarrow}{\mathcal{X}}$ and $C_{t} \in \accentset{\leftarrow}{\mathcal{X}}$ where
	\begin{equation*}
		\accentset{\leftarrow}{\Pi}_{t} (C_{t}) = \int_{\accentset{\leftarrow}{\textup{X}}} \cdots \int_{\accentset{\leftarrow}{\textup{X}}} 1_{C_{t}} (z_{t}) \accentset{\leftarrow}{P}_{t} (z_{t+1},\textup{d}z_{t}) \cdots \accentset{\leftarrow}{P}_{T-1} (z_{T},\textup{d}z_{T-1}) \accentset{\leftarrow}{\Pi}_{T} (\textup{d}z_{T})
	\end{equation*}
	for all $C_{t} \in \accentset{\leftarrow}{\mathcal{X}}$. Let now $((\bar{X}_t,\bar{Y}_t))_{t=1}^{T}$ be a causal state space model on $((\accentset{\rightarrow}{\textup{X}} \times \accentset{\leftarrow}{\textup{X}}) \times \textup{Y},(\accentset{\rightarrow}{\mathcal{X}} \otimes \accentset{\leftarrow}{\mathcal{X}}) \otimes \mathcal{Y})$ with transition kernels $\accentset{\rightarrow}{P}_t \tilde{P}_t$ and $\Phi_t$ and initial distribution $\accentset{\rightarrow}{\Pi}_1 \accentset{\leftarrow}{\Pi}_1$. Then, $((\accentset{\rightarrow}{X}_t,\accentset{\leftarrow}{X}_t,Y_t))_{t=1}^{T}$ and $((\bar{X}_t,\bar{Y}_t))_{t=1}^{T}$ are equal in distribution.
\end{proposition}
\begin{remark}
	If
	\begin{equation}
		\accentset{\leftarrow}{P}_{t-1} (z_{t},C_{t-1}) = \int_{C_{t-1}} \accentset{\leftarrow}{p}_{t-1} (z_{t},z_{t-1}) \accentset{\leftarrow}{\mu}_X (\textup{d}z_{t-1}), \quad z_{t} \in \accentset{\leftarrow}{\textup{X}}, C_{t-1} \in \accentset{\leftarrow}{\mathcal{X}}, \label{eq:RK2}
	\end{equation}
	where $\accentset{\leftarrow}{p}_{t-1} : \accentset{\leftarrow}{\textup{X}} \times \accentset{\leftarrow}{\textup{X}} \rightarrow (0,\infty)$ are measurable functions and $\accentset{\leftarrow}{\mu}_X : \accentset{\leftarrow}{\mathcal{X}} \rightarrow [0,\infty]$ is a $\sigma$-finite measure, then, by Lemma \ref{lem:app1},
	\begin{equation}
		\tilde{P}_{t} (z_{t-1},C_{t}) = \int_{C_{t}} \frac{\accentset{\leftarrow}{p}_{t-1} (z_{t},z_{t-1})}{\int_{\accentset{\leftarrow}{\textup{X}}} \accentset{\leftarrow}{p}_{t-1} (z_{t},z_{t-1}) \accentset{\leftarrow}{\Pi}_{t} (\textup{d}z_{t})} \accentset{\leftarrow}{\Pi}_{t} (\textup{d}z_{t}), \quad z_{t-1} \in \accentset{\leftarrow}{\textup{X}}, C_{t} \in \accentset{\leftarrow}{\mathcal{X}}. \label{eq:RK3}
	\end{equation}
\end{remark}

Thus, to perform inference in a causal non-causal state space model, one can transform it into a causal state space model and apply standard results from, for instance, \cite{CappeMoulinesRyden2005}. In practice, however, this trick is only useful if for all $\accentset{\leftarrow}{x} \in \accentset{\leftarrow}{\textup{X}}$, the transition kernel $\accentset{\leftarrow}{P}_t(\accentset{\leftarrow}{x},\cdot)$ has a density $\accentset{\leftarrow}{p}_t(\accentset{\leftarrow}{x},\cdot)$ (Equation \eqref{eq:RK2}) such that the transition kernel $\tilde{P}_{t}$ can be obtained explicitly (Equation \eqref{eq:RK3}) and not only implicitly (Equation \eqref{eq:RK1}) and, as the following example shows, there are non-contrived situations where this is not the case.
\begin{example} \label{ex:deux}
	Let $((X_t^F,X_t^B,Y_t))_{t=1}^{T}$ be a stochastic process on $(\mathbb{R}^3,\mathcal{B}(\mathbb{R}^3))$ given by
	\begin{equation*}
		Y_t = X_t^F + X_t^B + \varepsilon_t^Y, \quad \varepsilon_t^Y \overset{i.i.d.}{\sim} \mathcal{N} (0,\sigma_Y^2).
	\end{equation*}
	Here,
	\begin{equation*}
		X_t^F = \mu_F + \rho_{1F} X_{t-1}^F + \rho_{2F} X_{t-2}^F + \varepsilon_t^{F}, \quad \varepsilon_t^F \overset{i.i.d.}{\sim} \mathcal{N} (0,\sigma_F^2),
	\end{equation*}
	with the spectral radius of $\bar{\rho}_F := (\rho_{1F},\rho_{2F};1,0)$ less than one, $(X_1^F,X_0^F)^{\prime} \sim \mathcal{N} \left( M , V \right)$, $M := \bar{\mu}_F (I_2 - \bar{\rho}_F)^{-1}$, $\bar{\mu}_F := (\mu_F,0)^{\prime}$, $V := \sum_{i=0}^{\infty} \bar{\rho}_{F}^{i} \bar{\sigma}_F^2 (\bar{\rho}_{F}^{\prime})^{i}$, and $\bar{\sigma}_F^2 := (\sigma_F^2,0;0,0)$. Moreover, 
	\begin{equation*}
		X_t^B = \rho_{1B} X_{t+1}^B + \rho_{2B} X_{t+2}^B + \varepsilon_t^{B}, \quad \varepsilon_t^B \overset{i.i.d.}{\sim} \mathcal{C} (0,\sigma_B),
	\end{equation*}
	with $(X_T^B,X_{T+1}^B)^{\prime} = (x_T^B,x_{T+1}^B)^{\prime}$ where $\varepsilon_s^Y$, $\varepsilon_t^F$, $(X_1^F,X_0^F)^{\prime}$, and $\varepsilon_u^B$ are independent for all $s \in \{1,...,T\}$, $t \in \{2,...,T\}$, and $u \in \{1,...,T-1\}$.
	
	This is not a causal non-causal state space model, but it can be transformed into one as follows. Set $\bar{X}_t^F := (X_t^F,X_{t-1}^F)^{\prime}$, $\bar{\varepsilon}_t^{F} := (\varepsilon_t^{F},0)^{\prime}$, $\bar{X}_t^B := (X_t^B,X_{t+1}^B)^{\prime}$, and $\bar{\varepsilon}_t^{B} := (\varepsilon_t^{B},0)^{\prime}$. Then, the stochastic process $((\bar{X}_t^F,\bar{X}_t^B,Y_t))_{t=1}^{T}$ on $(\mathbb{R}^5,\mathcal{B}(\mathbb{R}^5))$ given by
	\begin{equation*}
		Y_t = \bar{X}_{1,t}^F + \bar{X}_{1,t}^B + \varepsilon_t^Y,
	\end{equation*}
	where
	\begin{equation*}
		\bar{X}_t^F = \bar{\mu}_F + \bar{\rho}_{F} \bar{X}_{t-1}^F + \bar{\varepsilon}_t^{F}
	\end{equation*}
	and
	\begin{equation*}
		\bar{X}_t^B = \bar{\rho}_B \bar{X}_{t+1}^B + \bar{\varepsilon}_t^{B},
	\end{equation*}
	with $\bar{\rho}_B := (\rho_{1B},\rho_{2B};1,0)$ is a causal non-causal state space model. Indeed, (i) $(\bar{X}_t^F)_{t=1}^{T}$ is a causal Markov process on $(\mathbb{R}^2,\mathcal{B}(\mathbb{R}^2))$ with transition kernel
	\begin{equation*}
		P^F (\bar{x}^F,A^F) = \int_{A_1^F} \frac{1}{\sigma_F} f_{\mathcal{N}} \left( \frac{z-\mu_F-\rho_{1F}\bar{x}_1^F-\rho_{2F}\bar{x}_2^F}{\sigma_F} \right) \lambda(\textup{d}z) \delta_{\bar{x}_1^F} (A_2^F), \quad \bar{x}^F \in \mathbb{R}^2, A^F \in \mathcal{B}(\mathbb{R}^2),
	\end{equation*}
	and initial distribution
	\begin{equation*}
		\Pi^F (A^F) = \int_{A^F} \det(V)^{-1/2} f_{\mathcal{N}_2} \left( V^{-1/2} (z-M) \right) \lambda_2(\textup{d}z), \quad A^F \in \mathcal{B}(\mathbb{R}^2),
	\end{equation*}
	where $f_{\mathcal{N}_k} (z) = (2 \pi)^{-k/2} \exp \left( - \frac{z^{\prime} z}{2} \right), z \in \mathbb{R}^k$ and $\lambda_k$ is the Lebesgue measure on $(\mathbb{R}^k,\mathcal{B}(\mathbb{R}^k))$. (ii) $(\bar{X}_t^B)_{t=1}^{T}$ is a non-causal Markov process on $(\mathbb{R}^2,\mathcal{B}(\mathbb{R}^2))$ with transition kernel
	\begin{equation*}
		P^B (\bar{x}^B,A^B) = \int_{A_1^B} \frac{1}{\sigma_B} f_{\mathcal{C}} \left( \frac{z-\rho_{1B}\bar{x}_1^B-\rho_{2B}\bar{x}_2^B}{\sigma_B} \right) \lambda(\textup{d}z) \delta_{\bar{x}_1^B} (A_2^B), \quad \bar{x}^B \in \mathbb{R}^2, A^B \in \mathcal{B}(\mathbb{R}^2),
	\end{equation*}
	and initial distribution
	\begin{equation*}
		\Pi^B (A^B) = \delta_{\bar{x}_{1,T}^B} (A_1^B) \delta_{\bar{x}_{2,T}^B} (A_2^B), \quad A^B \in \mathcal{B}(\mathbb{R}^2).
	\end{equation*}
	(iii) $(\bar{X}_t^F)_{t=1}^{T}$ and $(\bar{X}_t^B)_{t=1}^{T}$ are independent. (iv) Conditional on $((\bar{X}_t^F,\bar{X}_t^B))_{t=1}^{T}$, $(Y_t)_{t=1}^{T}$ is a sequence of independent random variables on $(\mathbb{R},\mathcal{B}(\mathbb{R}))$ with transition kernel
	\begin{equation*}
		\Phi ((\bar{x}^F,\bar{x}^B),B) = \int_{B} \frac{1}{\sigma_Y} f_{\mathcal{N}} \left( \frac{z-\bar{x}_1^F-\bar{x}_{1}^B}{\sigma_Y} \right) \lambda(\textup{d}z), \quad (\bar{x}^F,\bar{x}^B) \in \mathbb{R}^4, B \in \mathcal{B}(\mathbb{R}).
	\end{equation*}
	See Appendix \ref{appendix:otherexamples} for details once again.
	
	It is, however, not useful to transform it into a causal state space model since the transition kernel of the non-causal Markov process does not have a density. On the other hand, it is straightforward to transform it into a non-causal state space model, although the transition kernel of the causal Markov process does not have a density either, since, by Proposition 4.4.2 in \cite{BrockwellDavis1991}, a causal Gaussian autoregressive model and a non-causal one are distributionally equivalent.
\end{example}

\subsection{The Non-causal State Space Model} \label{sec:noncausalmodel}

We now discuss how to perform inference in a non-causal state space model $((X_t,Y_t))_{t=1}^{T}$ on $(\textup{X} \times \textup{Y},\mathcal{X} \otimes \mathcal{Y})$ where $(X_t)_{t=1}^{T}$ is a non-causal Markov process on $(\textup{X},\mathcal{X})$ with transition kernels
\begin{equation}
	\mathbb{P} ( X_t \in \cdot \mid X_{t+1} ) = P_t (X_{t+1},\cdot), \label{eq:M1}
\end{equation}
which we simply call the transition kernels, and initial distribution
\begin{equation}
	\mathbb{P} ( X_T \in \cdot) = \Pi_T (\cdot). \label{eq:M2}
\end{equation}
Moreover, conditional on $(X_t)_{t=1}^{T}$, $(Y_t)_{t=1}^{T}$ is a sequence of independent random variables on $(\textup{Y},\mathcal{Y})$ with transition kernels
\begin{equation}
	\mathbb{P} ( Y_t \in \cdot \mid X_1,...,X_T) = \int_{\cdot} \phi_t (X_t,y) \mu_Y (\textup{d}y), \label{eq:M3}
\end{equation}
which we call the observation kernels, where $\phi_t : \textup{X} \times \textup{Y} \rightarrow (0,\infty)$ are measurable functions, which we call the observation densities, and $\mu_Y : \mathcal{Y} \rightarrow [0,\infty]$ is a $\sigma$-finite measure; we have simplified the notation for convenience.
\begin{remark} \label{remark:transformation}
	Note that the observation kernel in the causal non-causal convolution autoregressive model in Equations \eqref{eq:ncssm1}-\eqref{eq:ncssm3} does not have a density. To the best of our knowledge, it is, however, not possible to perform inference in neither a causal state space model nor a non-causal state space model where the observation kernels do not have densities unless it is linear and Gaussian, see, for instance, the discussion in Section 2.4.6 in \cite{ChopinPapaspiliopoulos2020}. 
	
	One 'solution', which we will use later, is to use the theory for a causal (non-causal) state space model where the observation kernels have densities to perform inference, at least approximately, in a causal (non-causal) state space model where the observation kernels do not have densities. 
	
	In some cases, it is possible to transform a causal (non-causal) state space model where the observation kernels do not have densities into a causal (non-causal) autoregressive state space model where the observation kernels do have densities. If this is the case, then another solution is to transform it and use theory for a causal (non-causal) autoregressive state space model where the observation kernels have densities to perform inference in it. Although this is the case for the causal non-causal convolution autoregressive model in Equations \eqref{eq:ncssm1}-\eqref{eq:ncssm3} as shown in Appendix \ref{appendix:motivatingexample}, we do not use this solution at the moment.
\end{remark}

\subsubsection*{State Inference}

First, we discuss the estimation of the unobserved component $(X_t)_{t=1}^{T}$ given the observed component $(Y_t)_{t=1}^{T}$ in the non-causal state space model in Equations \eqref{eq:M1}-\eqref{eq:M3} via the conditional expectation
\begin{equation*}
	\mathbb{E} [f(X_t) \mid Y_1,...,Y_s]
\end{equation*}
for all $f : \textup{X} \rightarrow \mathbb{R}$ such that it exists; first, in theory, then, in practice.

To compute the conditional expectation, we need the conditional distribution
\begin{equation*}
	\pi_{t \mid s} (A) := \mathbb{P} (X_t \in A \mid Y_1,...,Y_s), \quad A \in \mathcal{X},
\end{equation*}
which is called the filtering distribution if $t = s$, the prediction distribution if $t > s$, and the smoothing distribution if $t < s$. We, however, restrict our attention to the filtering and smoothing distributions since, to compute the prediction distribution, the transition kernels must have densities in which case the non-causal state space model can be transformed into a causal state space model as discussed above.

To compute the filtering and smoothing distributions, we need the backward kernel and the forward function. The backward kernel is defined as follows.
\begin{definition} \label{def:bk}
	The backward kernel is the kernel $\alpha_k : \textup{Y}^{s-k+1} \times \mathcal{X} \rightarrow [0,\infty)$ given by
	\begin{equation*}
		\alpha_k (y_{k:s},A) := \int_{\textup{X}} \cdots \int_{\textup{X}} 1_{A} (x_k) \phi_k(x_k,y_k) \cdots \phi_s(x_s,y_s) P_k(x_{k+1},\textup{d}x_k) \cdots P_{s-1}(x_s,\textup{d}x_{s-1}) \Pi_{s} (\textup{d}x_s), \quad y_{k:s} \in \textup{Y}^{s-k+1}, A \in \mathcal{X},
	\end{equation*}
	where
	\begin{equation*}
		\Pi_{s} (A) := \int_{\textup{X}} \cdots \int_{\textup{X}} 1_{A}(x_s) P_s(x_{s+1},\textup{d}x_s) \cdots P_{T-1}(x_T,\textup{d}x_{T-1}) \Pi_T(\textup{d}x_T), \quad A \in \mathcal{X},
	\end{equation*}
	which can be computed recursively as
	\begin{equation*}
		\alpha_k (y_{k:s},A) = \int_{\textup{X}} \int_{\textup{X}} 1_A (x_k) \phi_k (x_k,y_k) P_k(x_{k+1},\textup{d}x_k) \alpha_{k+1} (y_{k+1:s},\textup{d}x_{k+1})
	\end{equation*}
	with terminal condition $\alpha_s (y_s,A) = \int_{\textup{X}} 1_A (x_s) \phi_s (x_s,y_s) \Pi_{s} (\textup{d}x_s)$, see Lemma \ref{lem:app2} for details.
\end{definition}
\noindent
Moreover, the forward function is defined as follows.
\begin{definition} \label{def:ff}
	The forward function is the function $\beta_k : \textup{Y}^{k-1} \times \textup{X} \rightarrow (0,\infty)$ given by
	\begin{equation*}
		\beta_k (y_{1:k-1},x_k) := \int_{\textup{X}} \cdots \int_{\textup{X}} \phi_1(x_1,y_1) \cdots \phi_{k-1}(x_{k-1},y_{k-1}) P_1(x_2,\textup{d}x_1) \cdots P_{k-1}(x_k,\textup{d}x_{k-1}), \quad y_{1:k-1} \in \textup{Y}^{k-1}, x_k \in \textup{X},
	\end{equation*}
	which can be computed recursively as
	\begin{equation*}
		\beta_k (y_{1:k-1},x_k) = \int_{\textup{X}} \beta_{k-1} (y_{1:k-2},x_{k-1}) \phi_{k-1} (x_{k-1},y_{k-1}) P_{k-1} (x_{k},\textup{d}x_{k-1})
	\end{equation*}
	with initial condition $\beta_1 (y_{1:0},x_1) = 1$.
\end{definition}
\noindent
It will soon become clear that the recursion for the backward kernel and the one for the forward function is the key to compute the filtering and smoothing distributions in practice.

The filtering distribution is given in the following proposition.
\begin{proposition} \label{prop:filtering}
	The filtering distribution $\pi_{t \mid t} := \pi_{t}$ can be computed as
	\begin{equation*}
		\pi_t (A) = \frac{\int_{\textup{X}} 1_{A} (x_t) \beta_t (Y_{1:t-1},x_t) \alpha_t (Y_t,\textup{d}x_t) }{\int_{\textup{X}} \beta_t (Y_{1:t-1},x_t) \alpha_t (Y_t,\textup{d}x_t)}, \quad A \in \mathcal{X}.
	\end{equation*}
\end{proposition}
\noindent
The filtering distribution for a non-causal state space model is thus computed via the forward function in the same vein as the filtering distribution for a causal state space model which is computed via the so-called forward kernel. Note, however, that the former is computed via the forward \textit{function}, whereas the latter is computed via the forward \textit{kernel}.

The smoothing distribution is given in the following proposition.
\begin{proposition} \label{prop:smoothing}
	The smoothing distribution $\pi_{t \mid s}, t < s$ can be computed as
	\begin{equation*}
		\pi_{t \mid s} (A) = \frac{\int_{\textup{X}} 1_{A} (x_t) \beta_t (Y_{1:t-1},x_t) \alpha_t (Y_{t:s},\textup{d}x_t) }{\int_{\textup{X}} \beta_t (Y_{1:t-1},x_t) \alpha_t (Y_{t:s},\textup{d}x_t)}, \quad A \in \mathcal{X}.
	\end{equation*}
\end{proposition}
\noindent
The smoothing distribution for a non-causal state space model is thus computed via both the backward kernel and the forward function analogous to the smoothing distribution for a causal state space model which is computed via the forward kernel and the so-called backward function.

In practice, the filtering and smoothing distributions must be approximated by, for instance, sequential Monte Carlo methods; see, for instance, \cite{ChopinPapaspiliopoulos2020} for an introduction to sequential Monte Carlo methods. Indeed, the recursion for the (normalised) backward kernel is similar to the forward recursion in Section 5.2.1 in \cite{ChopinPapaspiliopoulos2020}, so the (normalised) backward kernel can be approximated by the methods in Chapter 10 in \cite{ChopinPapaspiliopoulos2020}. Moreover, the one for the forward function is similar to the backward recursion in Section 5.3.1 in \cite{ChopinPapaspiliopoulos2020}, so the forward function can be approximated by the methods in Chapter 12 in \cite{ChopinPapaspiliopoulos2020}. In the end, the filtering and smoothing distributions can then be approximated by combining them.


\subsubsection*{Parameter Inference}

We now briefly discuss the estimation of the non-causal state space model in Equations \eqref{eq:M1}-\eqref{eq:M3} via maximum likelihood estimation.

To this end, let $\{ ((X_t^{\theta},Y_t^{\theta}))_{t=1}^{T} : \theta \in \Theta \}$ be a family of non-causal state space models with transition kernels $P_t^{\theta}$, initial distribution $\Pi_T^{\theta}$, and observation densities $\phi_t^{\theta}$ where $(\Theta,\mathcal{H})$ is a measurable space. Assume that a sample $(y_t)_{t=1}^{T}$ from the non-causal state space model $((X_t^{\theta_0},Y_t^{\theta_0}))_{t=1}^{T}$ is observed. Then, the maximum likelihood estimator $\hat{\theta}_T (y_{1:T})$ of $\theta_0$ is given by
\begin{equation*}
	\hat{\theta}_T (y_{1:T}) = \underset{\theta \in \Theta}{\text{argmax}} \, L_T^{\theta} (y_{1:T}),
\end{equation*}
where $L_T^{\theta} (y_{1:T})$ is the likelihood function given by
\begin{equation*}
	L_T^{\theta} (y_{1:T}) = \alpha_1^{\theta} (y_{1:T},\textup{X})
\end{equation*}
since, by Lemma \ref{lem:cncssm} together with Fubini's theorem,
\begin{equation*}
	\mathbb{E}^{\theta} [ f(Y_1^{\theta},...,Y_T^{\theta}) ] = \int_{\textup{Y}} \cdots \int_{\textup{Y}} \int_{\textup{X}} \cdots \int_{\textup{X}} f(y_1,...,y_T) \prod_{t=1}^{T} \phi_t^{\theta}(x_t,y_t) \prod_{t=1}^{T-1} P_t^{\theta}(x_{t+1},\textup{d}x_t) \Pi_{T}^{\theta} (\textup{d}x_T) \prod_{t=1}^{T} \mu_Y (\textup{d}y_{t})
\end{equation*}
for all $f \in \mathcal{M}_{b}(\textup{Y}^T)$.


\section{The Causal Non-causal Convolution Autoregressive Model} \label{sec:cnccar}

In this section, we consider the causal non-causal convolution autoregressive model by \cite{GourierouxZakoian2017}, see also \cite{GourierouxJasiakTong2021}, in more detail.

\subsection{Structure}

Recall that the causal non-causal convolution autoregressive model $(Y_t)_{t \in \mathbb{Z}}$ is given by
\begin{equation}
	Y_t = X_t^F + X_t^B, \label{eq:cnccar1}
\end{equation}
where $(X_t^F)_{t \in \mathbb{Z}}$ is given by
\begin{equation}
	X_t^F = \mu_F + \rho_F X_{t-1}^F + \varepsilon_t^F, \quad \varepsilon_t^F \overset{i.i.d.}{\sim} \mathcal{N} (0,\sigma_F^2), \label{eq:cnccar2}
\end{equation}
with $| \rho_F | < 1$ and $(X_t^B)_{t \in \mathbb{Z}}$ is given by
\begin{equation}
	X_t^B = \rho_B X_{t+1}^B + \varepsilon_t^B, \quad \varepsilon_t^B \overset{i.i.d.}{\sim} \mathcal{S} (\alpha,0,\sigma_B,\beta), \label{eq:cnccar3}
\end{equation}
with $| \rho_B | < 1$ where $\varepsilon_s^F$ and $\varepsilon_t^B$ are independent for all $s \in \mathbb{Z}$ and $t \in \mathbb{Z}$.

First, the stationary and ergodic solution to Equation \eqref{eq:cnccar2} is given by
\begin{equation*}
	X_t^F = \frac{\mu_F}{1-\rho_F} + \sum_{i=0}^{\infty} \rho_F^i \varepsilon_{t-i}^F, 
\end{equation*}
and the distribution of $X_t^F$ is given by $X_t^F \sim \mathcal{N} \left( \mu_{X^F},\sigma_{X^F}^{2} \right)$ where $\mu_{X^F} = \frac{\mu_F}{1-\rho_F}$ and $\sigma_{X^F}^{2} = \frac{\sigma_F^2}{1-\rho_F^2}$.

Moreover, by Proposition 13.3.1 in \cite{BrockwellDavis1991}, the stationary and ergodic solution to Equation \eqref{eq:cnccar3} is given by
\begin{equation*}
	X_t^B = \sum_{i=0}^{\infty} \rho_B^i \varepsilon_{t+i}^B,
\end{equation*}
and, by Proposition 1 in \cite{GourierouxZakoian2017}, the distribution of $X_t^B$ is given by $X_t^B \sim \mathcal{S} \left( \alpha,\mu_{X^B},\sigma_{X^B},\beta_{X^B} \right)$ where $\mu_{X^B} = 0$ if $\alpha \neq 1$ and $\mu_{X^B} = - \beta \sigma_B \frac{2}{\pi} \frac{\rho_B \log |\rho_B|}{(1-\rho_{B})^2}$ if $\alpha = 1$, $\sigma_{X^B} = \frac{\sigma_B}{(1-|\rho_{B}|^{\alpha})^{1/\alpha}}$, and $\beta_{X^B}  = \frac{1-|\rho_{B}|^{\alpha}}{1-\textup{sign}(\rho_{B}) |\rho_{B}|^{\alpha}} \beta$.

The conditional expectation $\mathbb{E} [X_t^B \mid X_{t+1}^B]$ does not exist if $\alpha \leq 1$ and neither does the unconditional expectation $\mathbb{E} [X_t^B]$. However, by Proposition 2 in \cite{GourierouxZakoian2017}, the conditional expectation $\mathbb{E} [X_t^B \mid X_{t-1}^B]$ exists and is given in the following proposition, which is a generalisation of Proposition 3 in \cite{GourierouxZakoian2017}.\footnote{An analogous result for the causal autoregressive model $$Y_t = \phi Y_{t-1} + \varepsilon_t, \quad \varepsilon_t \overset{i.i.d.}{\sim} \mathcal{S} (\alpha,0,\sigma,\beta),$$ where $| \phi | < 1$ has been proved by \cite{CambanisFakhre-Zakeri1995}.}
\begin{proposition} \label{prop:Proposition3}
	If $\rho_B \neq 0$ and $\beta = 0$ or if $\rho_{B} \in (0,1)$ and $\alpha \neq 1$, then
	\begin{equation*}
		\mathbb{E} \left[ X_t^B \mid X_{t-1}^B \right] = \textup{sign}(\rho_B) |\rho_B|^{\alpha-1} X_{t-1}^B.
	\end{equation*}
\end{proposition}
\noindent
It follows from Proposition \ref{prop:Proposition3} and the law of iterated expectations that
\begin{equation*}
	\mathbb{E} \left[ X_{t+h}^B \mid X_{t-1}^B \right] = \textup{sign}(\rho_B)^{(h+1)} |\rho_B|^{(\alpha-1)(h+1)} X_{t-1}^B
\end{equation*}
for all $h \in \mathbb{N}$.\footnote{There is a typo in Proposition 3 in \cite{GourierouxZakoian2017}.}

It is well-known that if $\alpha = 2$, then the conditional expectation
\begin{equation*}
	\mathbb{E} \left[ X_t^B \mid X_{t-1}^B \right] = \rho_B X_{t-1}^B
\end{equation*}
is the best predictor in mean square and, since it is linear, it is also the best linear predictor in mean square. In the case where $\alpha < 2$, however, the mean square error criterion is no longer applicable since the unconditional expectation $\mathbb{E} [(X_t^B)^2]$ does not exist if $\alpha < 2$. In this case, a natural criterion to use is the scale since if $X \sim \mathcal{S} (\alpha,0,\sigma(X),\beta)$ with $\alpha < 2$ and $\beta = 0$ if $\alpha = 1$, then for all $0 < p < \alpha$, there exists a $c(\alpha,\beta,p) > 0$ such that
\begin{equation*}
	\mathbb{E} \left[ |X|^p \right] = c(\alpha,\beta,p) \sigma^p(X),
\end{equation*}
see, for instance, \cite{SamorodnitskyTaqqu1994}. In the following, a linear predictor
\begin{equation*}
	L(X_t^B \mid X_{t-1}^B) = a X_{t-1}^B
\end{equation*}
is thus called a best linear predictor if it minimises
\begin{equation*}
	\sigma(X_t^B-L(X_t^B \mid X_{t-1}^B)).
\end{equation*}
The best linear predictor(s) $L^{*}(X_t^B \mid X_{t-1}^B)$ is (are) given in the following proposition.\footnote{An analogous result for the causal autoregressive model $$Y_t = \phi Y_{t-1} + \varepsilon_t, \quad \varepsilon_t \overset{i.i.d.}{\sim} \mathcal{S} (\alpha,0,\sigma,0),$$ where $| \phi | < 1$ has been proved by \cite{CambanisFakhre-Zakeri1995}.}
\begin{proposition} \label{prop:BLP}
	If $\rho_B \neq 0$ and $\beta = 0$ or if $\rho_{B} \in (0,1)$ and $\alpha \neq 1$, then the best linear predictor(s) $L^{*}(X_t^B \mid X_{t-1}^B)$ is (are) given by
	\begin{equation*}
		L^{*}(X_t^B \mid X_{t-1}^B) = a^{*} X_{t-1}^B,
	\end{equation*}
	where
	\begin{equation*}
		a^{*} = \frac{\textup{sign}(\rho_B)|\rho_B|^{1/(\alpha-1)}}{|\rho_B|^{\alpha/(\alpha-1)}+(1-|\rho_B|^{\alpha})^{1/(\alpha-1)}}, \quad \textup{if} \quad 1 < \alpha \leq 2,
	\end{equation*}
	and
	\begin{equation*}
		a^{*}
		\begin{cases}
			= 0, & \textup{if} \quad 0 < |\rho_B|^{\alpha} < 1/2 \ \textup{and} \ 0 < \alpha \leq 1, \\
			= 1/\rho_B, & \textup{if} \quad 1/2 < |\rho_B|^{\alpha} < 1 \ \textup{and} \ 0 < \alpha \leq 1, \\
			\in \{0,1/\rho_B\}, & \textup{if} \quad |\rho_B|^{\alpha} = 1/2 \ \textup{and} \ 0 < \alpha < 1, \\
			\in [0,2], & \textup{if} \quad \rho_B = 1/2 \ \textup{and} \ \alpha = 1,  \\
			\in [-2,0], & \textup{if} \quad \rho_B = -1/2 \ \textup{and} \ \alpha = 1.
		\end{cases}
	\end{equation*}
\end{proposition}
\noindent
Note that if $\alpha < 2$, then the conditional expectation $\mathbb{E} \left[ X_t^B \mid X_{t-1}^B \right]$ in Proposition \ref{prop:Proposition3} is different from the best linear predictor(s) 
$L^{*}(X_t^B \mid X_{t-1}^B)$ in Proposition \ref{prop:BLP}.

The causal non-causal convolution autoregressive model $(Y_t)_{t \in \mathbb{Z}}$ is thus stationary and ergodic since $(X_t^F)_{t \in \mathbb{Z}}$ and $(X_t^B)_{t \in \mathbb{Z}}$ are, but, in general, the distribution of $Y_t$ is not known although those of $X_t^F$ and $X_t^B$ are. The characteristic function of $Y_{t-p:t} := (Y_t,Y_{t-1},...,Y_{t-p})^{\prime}, p \in \mathbb{N}_0$ is, however, given in the following proposition.
\begin{proposition} \label{prop:cf}
	The characteristic function of $Y_{t-p:t}$ is given by
	\begin{equation*}
		\varphi_Y^p (u) = \varphi_{X^F}^p (u) \varphi_{X^B}^p (u), \quad u \in \mathbb{R}^{p+1},
	\end{equation*}
	where
	\begin{equation*}
		\varphi_{X^F}^p (u) = \exp \left( \textup{i} \mu_{X^F} \sum_{i=0}^{p} u_i - \frac{\sigma_F^2 \sum_{j=0}^{p-1} \left( \sum_{i=0}^{j} u_i \rho_F^{j-i} \right)^2}{2} - \frac{\sigma_{X^F}^{2} \left( \sum_{i=0}^{p} u_i \rho_F^{p-i} \right)^2}{2} \right)
	\end{equation*}
	and
	\begin{align*}
		\varphi_{X^B}^p (u) &= \exp \left\{ \textup{i} \mu_{X^B} \sum_{i=0}^{p} u_{p-i} \rho_B^{p-i} - \sigma_B^{\alpha} \sum_{j=0}^{p-1} \left| \sum_{i=0}^{j} u_{p-i} \rho_B^{j-i} \right|^{\alpha} \left( 1 - \textup{i} \beta \textup{sign} \left( \sum_{i=0}^{j} u_{p-i} \rho_B^{j-i} \right) \Phi^j \right) \right. \\
		& \quad \left. - \sigma_{X^B}^{\alpha} \left| \sum_{i=0}^{p} u_{p-i} \rho_B^{p-i} \right|^{\alpha} \left( 1 - \textup{i} \beta_{X^B} \textup{sign} \left( \sum_{i=0}^{p} u_{p-i} \rho_B^{p-i} \right) \Phi^p \right) \right\}
	\end{align*}
	with $\Phi^j = \tan \left( \frac{\pi \alpha}{2} \right)$ if $\alpha \neq 1$ and $\Phi^j = - \frac{2}{\pi} \log \left( \left| \sum_{i=0}^{j} u_{p-i} \rho_B^{j-i} \right| \right)$ if $\alpha = 1$ where, by convention, $\sum_{j=0}^{-1} \cdot = 0$.
\end{proposition}
\noindent
Proposition \ref{prop:cf} shows that the characteristic function of $Y_{t-p:t}$ is equal to the product of the characteristic functions of $X_{t-p:t}^F$ and $X_{t-p:t}^B$. The former is equal to the characteristic function of $X_{t-p}^F$ evaluated at $\sum_{i=0}^{p} u_i \rho_F^{p-i}$ multiplied by a term that describes the dependence between $X_{t}^F,...,X_{t-p}^F$.\footnote{To see this, note that $$\varphi_{X^F}^p (u) = \exp \left( \textup{i} \mu_{X^F} \sum_{i=0}^{p} u_i \rho_F^{p-i} + \textup{i} \mu_{X^F} \sum_{i=0}^{p} u_i (1-\rho_F^{p-i}) - \frac{\sigma_F^2 \sum_{j=0}^{p-1} \left( \sum_{i=0}^{j} u_i \rho_F^{j-i} \right)^2}{2} - \frac{\sigma_{X^F}^{2} \left( \sum_{i=0}^{p} u_i \rho_F^{p-i} \right)^2}{2} \right).$$} Similarly, the latter is equal to the characteristic function of $X_{t}^B$ evaluated at $\sum_{i=0}^{p} u_{p-i} \rho_B^{p-i}$ multiplied by a term that describes the dependence between $X_{t}^B,...,X_{t-p}^B$. 

The characteristic function of $Y_t$ follows directly from Proposition \ref{prop:cf}.
\begin{corollary} \label{coro:cf}
	The characteristic function of $Y_t$ is given by
	\begin{equation*}
		\varphi_Y (u) = \varphi_{X^F} (u) \varphi_{X^B} (u), \quad u \in \mathbb{R},
	\end{equation*}
	where
	\begin{equation*}
		\varphi_{X^F} (u) = \exp \left( \textup{i} \mu_{X^F} u - \frac{\sigma_{X^F}^{2} u^2}{2} \right)
	\end{equation*}
	and
	\begin{equation*}
		\varphi_{X^B} (u) = \exp \left( \textup{i} \mu_{X^B} u - \sigma_{X^B}^{\alpha} \left| u \right|^{\alpha} \left( 1 - \textup{i} \beta_{X^B} \textup{sign} (u) \Phi \right) \right).
	\end{equation*}
\end{corollary}

An exemption, besides the trivial one when $\alpha = 2$ and $\beta = 0$ where the stable distribution reduces to the normal distribution, is when $\alpha = 1$ and $\beta = 0$ where the stable distribution reduces to the Cauchy distribution.
\begin{proposition} \label{prop:Voigt}
	If $\alpha = 1$ and $\beta = 0$, then $Y_t$ is Voigt distributed,
	\begin{equation}
		\varphi_{Y} (u) = \exp \left( \textup{i} \mu_{X^F} u - \frac{\sigma_{X^F}^{2} u^2}{2} - \sigma_{X^B} \left| u \right| \right), \quad u \in \mathbb{R}, \label{eq:Voigt1}
	\end{equation}
	and
	\begin{equation}
		f_{Y} (y) = \frac{1}{\sqrt{2\pi} \sigma_{X_F}} \textup{Re} \left( w \left( \frac{y-\mu_{X^F}+\textup{i}\sigma_{X_B}}{\sqrt{2}\sigma_{X_F}} \right) \right), \quad y \in \mathbb{R}, \label{eq:Voigt2}
	\end{equation}
	where $w (z) = \exp(-z^2) \textup{erfc} (-\textup{i}z), z \in \mathbb{C}$ is the Faddeeva function and $\textup{erfc} (z) = \frac{2}{\sqrt{\pi}} \int_{z}^{\infty} \exp(-t^2) \textup{d}t, z \in \mathbb{C}$ is the complementary error function.
\end{proposition}
\noindent
Note that, in contrast to the characteristic function in Equation \eqref{eq:Voigt1}, the probability density function in Equation \eqref{eq:Voigt2} must be evaluated numerically, but this can be done very fast and accurately.

\subsection{Statistical Inference}

The causal non-causal convolution autoregressive model cannot be estimated by maximum likelihood estimation since the observation kernel does not have a density. However, since the characteristic function is available in closed form, it can straightforward be estimated by empirical characteristic function estimation.

The basic idea of empirical characteristic function estimation, which was initiated by \cite{Parzen1962}, is to minimise some distance measure between the empirical characteristic function and the characteristic function motivated by the fact that two cumulative distribution functions are equal if and only if their characteristic functions are equal. There is by now an abundant literature on empirical characteristic function estimation. For independent and identically distributed random variables, see, for instance, \cite{Press1972}, \cite{PaulsonHolcombEdwardLeitch1975}, \cite{FeuervergerMureika1977}, \cite{Heathcote1977}, and \cite{KogonWilliams1998} and, for dependent random variables, see, for instance, \cite{Feuerverger1990}, \cite{KnightYu2002}, \cite{MeintanisTaufer2012}, and \cite{FrancqMeintanis2016}. See \cite{Yu2004} for a review of the literature on empirical characteristic function estimation.

Assume that a sample $(y_t)_{t=1}^{T}$ from the causal non-causal convolution autoregressive model $(Y_t)_{t \in \mathbb{Z}}$ with $\theta = \theta_0$ is observed. Here,
\begin{equation*}
	\theta = (\mu_F,\rho_F,\sigma_F^2,\rho_B,\alpha,\sigma_B,\beta)^{\prime}
\end{equation*}
is the parameter vector and
\begin{equation*}
	\Theta = \mathbb{R} \times (-1,1) \times (0,\infty) \times (-1,1) \times (0,2] \times (0,\infty) \times [-1,1]
\end{equation*}
is the parameter space. Divide the observations $(y_t)_{t=1}^{T}$ into $T-p$ blocks $y_{t-p:t} := (y_{t},y_{t-1},...,y_{t-p})^{\prime}, t \in \{p+1,p+2...,T\}$ of size $p+1$ where $p \in \mathbb{N}_0$. Then, the empirical characteristic function estimator $\hat{\theta}_{T}$ of $\theta_0$ by \cite{KnightYu2002} is given by
\begin{equation*}
	\hat{\theta}_{T} = \underset{\theta \in \Theta}{\text{argmin}} \, \int_{\mathbb{R}^{p+1}} | \varphi_T^p (u) - \varphi_Y^p (u;\theta) |^2 w(u) \textup{d}u,
\end{equation*}
where
\begin{equation*}
	\varphi_T^p (u) = \frac{1}{T-p} \sum_{t=p+1}^{T} \exp \left( \textup{i} u^{\prime} y_{t-p:t} \right), \quad u \in \mathbb{R}^{p+1},
\end{equation*}
is the empirical characteristic function of $Y_{t-p:t}$,
\begin{equation*}
	\varphi_Y^p (u;\theta) = \mathbb{E} \left[ \exp \left( \textup{i} u^{\prime} Y_{t-p:t} \right) \right], \quad u \in \mathbb{R}^{p+1},
\end{equation*}
is the characteristic function of $Y_{t-p:t}$ given in Proposition \ref{prop:cf}, and $w : \mathbb{R}^{p+1} \rightarrow [0,\infty)$ is a weight function.\footnote{The empirical characteristic function estimator by \cite{KnightYu2002} reduces to the one by \cite{Feuerverger1990} if the weight function takes values in a finite set.} In contrast to the characteristic function, the empirical characteristic function estimator is not available in closed form, so it must be obtained via numerical methods.

\cite{KnightYu2002} also gave general conditions under which the empirical characteristic function estimator is consistent and asymptotically normal. It is, however, not immediately clear whether the general conditions in \cite{KnightYu2002} hold for the causal non-causal convolution autoregressive model. We therefore leave the asymptotic properties of the empirical characteristic function estimator for the causal non-causal convolution autoregressive model for future research, but study its finite-sample properties in the next section.

Finally, a word on the choice of $p$. In general, the choice of $p$ does not affect the consistency or asymptotic normality of the empirical characteristic function estimator; indeed, the empirical characteristic function estimator is consistent and asymptotically normal for all $p$ provided that the conditions in \cite{KnightYu2002} hold. It can, however, affect the efficiency of it, but there is a trade-off. On the one hand, the larger the $p$, the more efficient the estimator. On the other hand, however, a larger $p$ also implies a larger computational burden required to obtain the estimator. See Section 2.4 in \cite{KnightYu2002} for more details.

\section{The Motivating Example Revisited} \label{sec:motivation_revisited}

The purpose of this section is twofold. First, we study the finite-sample properties of the empirical characteristic function estimator for the causal non-causal convolution autoregressive model, $(Y_t)_{t \in \mathbb{Z}}$, fixing $\alpha = 1$ and $\beta = 0$ in a Monte Carlo simulation study to see whether we can recover the true parameters. Using the true parameters, we then estimate the states, $(x_t^F)_{t=1}^{250}$ and $(x_t^B)_{t=1}^{250}$, from the observations, $(y_t)_{t =1}^{250}$, in Figure \ref{fig:cncssm} to see whether we can recover the true states. 

Table \ref{tab:parinf1} reports the estimated means and standard deviations (in parenthesis) of the estimated parameters together with the true parameters. The means and standard deviations of the estimated parameters are estimated from ten thousand replications where a replication consists of first simulating $T$ observations from the model and then estimating the model from the $T$ observations using the empirical characteristic function estimator with $p = 1$ and $w(u) = \exp(-u^{\prime}u)$ as in \cite{KnightYu2002}. The finite-sample properties of the estimated parameters are satisfactory. Indeed, the estimated means and standard deviations of the estimated parameters converge to the true parameters and zero, respectively.

\begin{table}[htbp]
	
	\centering
	
	\begin{tabular}{cccccc}
		\toprule
		& $\mu_{F,0} = 5$ & $\rho_{F,0} = 0.9$ & $\sigma_{F,0}^2 = 10$ & $\rho_{B,0} = 0.9$ & $\sigma_{B,0} = 0.5$ \\
		\midrule
		$T = 125$ & $\underset{(1.057)}{5.306}$ & $\underset{(0.024)}{0.893}$  & $\underset{(6.278)}{10.951}$ & $\underset{(0.097)}{0.891}$ & $\underset{(0.476)}{0.546}$ \vspace{0.10cm} \\
		$T = 250$ & $\underset{(0.929)}{5.263}$ & $\underset{(0.020)}{0.894}$  & $\underset{(5.302)}{10.923}$ & $\underset{(0.087)}{0.894}$ & $\underset{(0.387)}{0.538}$ \vspace{0.10cm} \\
		$T = 500$ & $\underset{(0.687)}{5.184}$ & $\underset{(0.015)}{0.896}$  & $\underset{(4.455)}{10.942}$ & $\underset{(0.068)}{0.899}$ & $\underset{(0.318)}{0.522}$ \vspace{0.10cm} \\
		\bottomrule
	\end{tabular}
	
	\caption{The estimated means and standard deviations (in parenthesis) of the estimated parameters of the causal non-causal convolution autoregressive model together with the true parameters.}
	
	\label{tab:parinf1}
	
\end{table}

Moving on to the second purpose of this section, Figure \ref{fig:stainf1} shows the estimated states, $(\hat{x}_t^F)_{t=1}^{250}$ (green) and $(\hat{x}_t^B)_{t=1}^{250}$ (red). The estimated states are obtained using the filtered (top) and smoothed (bottom) expectation with $f(x) = 1_{|x| \leq 10^6} x + 1_{|x| > 10^6} 10^6$ for a causal non-causal convolution autoregressive model where $\Phi((x^F,x^B),B) = \int_{B} f_{\mathcal{N}} (z-x^F-x^B) \lambda(\textup{d}z)$ implemented via the particle filter and smoother by \cite{GordonSalmondSmith1993} and \cite{KlaasBriersDeFreitasDoucetMaskellLang2006}, respectively, with a thousand particles. Figure \ref{fig:stainf1} also shows the true states, $(x_t^F)_{t=1}^{250}$ and $(x_t^B)_{t=1}^{250}$ (both black). The estimated states are reasonably close to the true states overall; there are periods in which they are close to the true ones, and there are periods in which they are far away from the true ones.

\begin{figure}[htbp]
	
	\centering
	
	\begin{subfigure}{0.5\textwidth}
		\centering
		\includegraphics[width=0.9\textwidth,height=0.25\textheight]{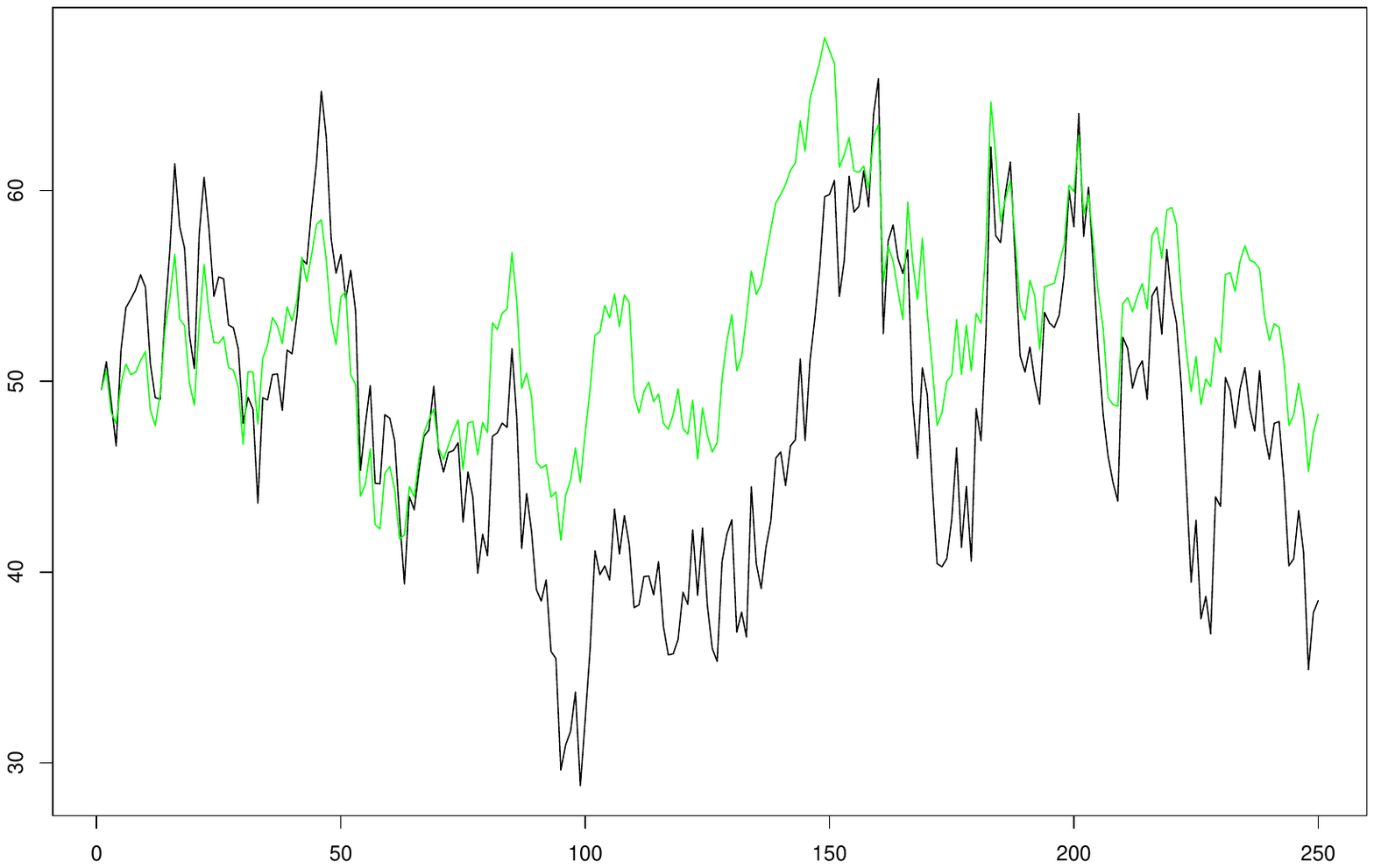}
	\end{subfigure} \kern-2.5em
	\begin{subfigure}{0.5\textwidth}
		\centering
		\includegraphics[width=0.9\textwidth,height=0.25\textheight]{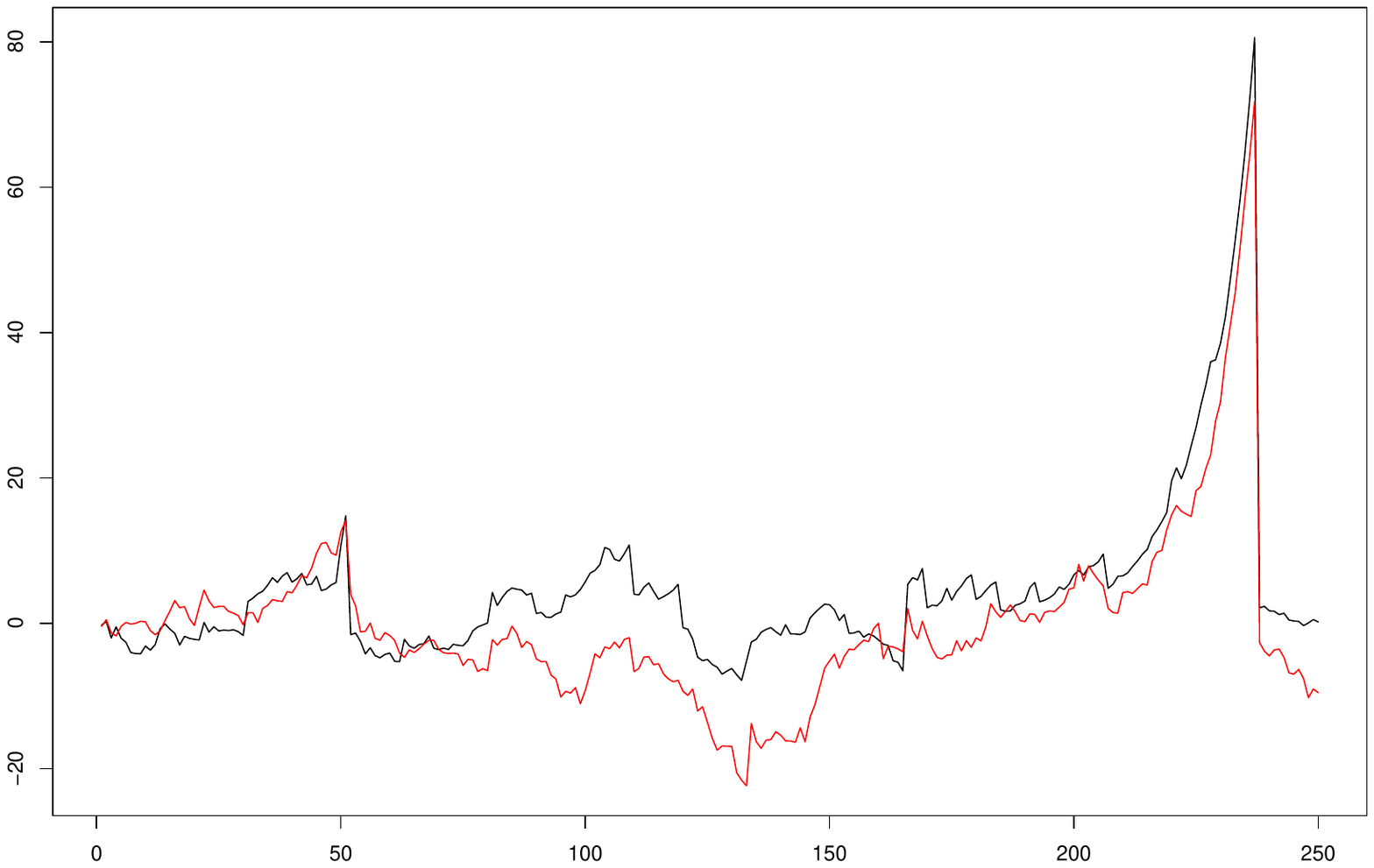}
	\end{subfigure}
	
	\begin{subfigure}{0.5\textwidth}
		\centering
		\includegraphics[width=0.9\textwidth,height=0.25\textheight]{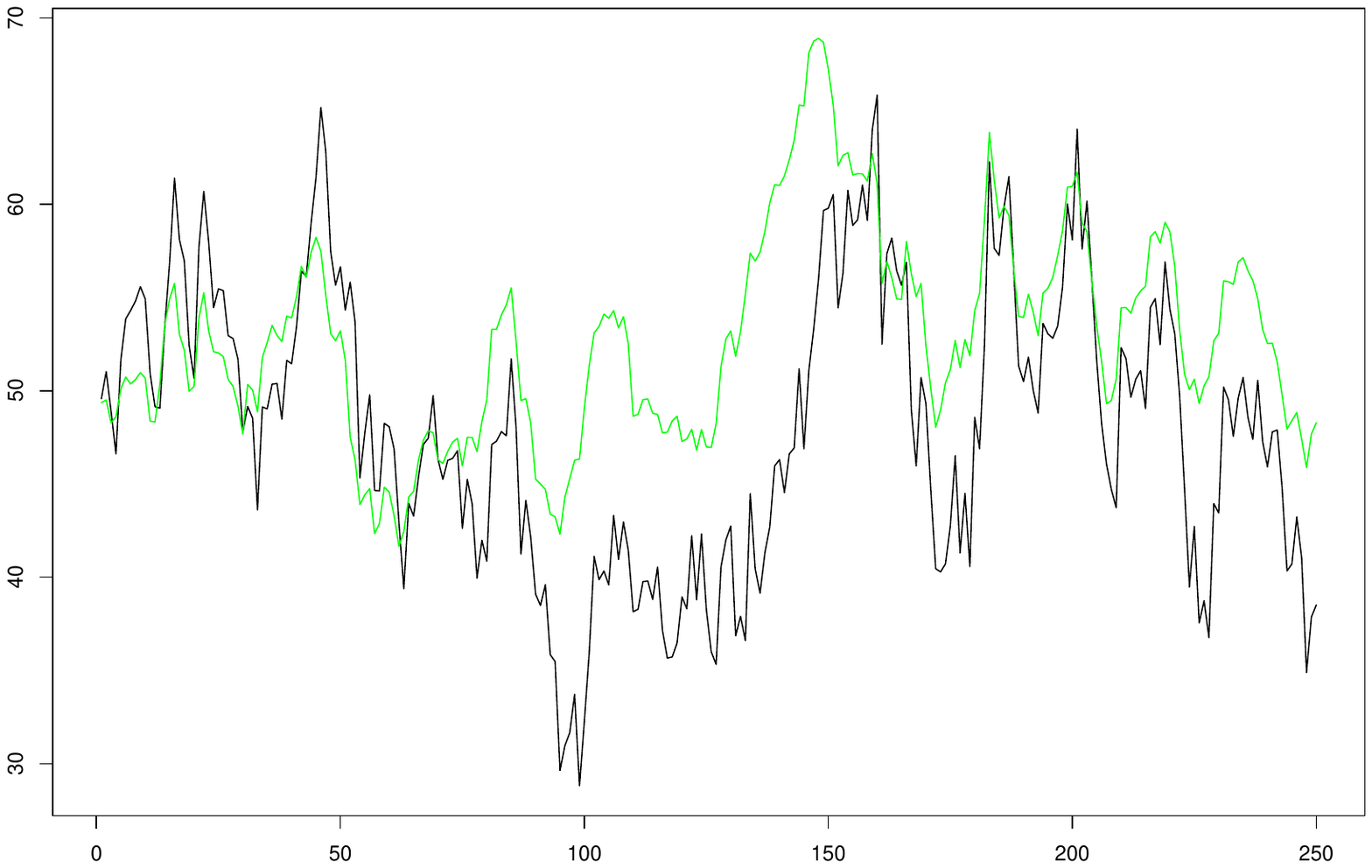}
	\end{subfigure} \kern-2.5em
	\begin{subfigure}{0.5\textwidth}
		\centering
		\includegraphics[width=0.9\textwidth,height=0.25\textheight]{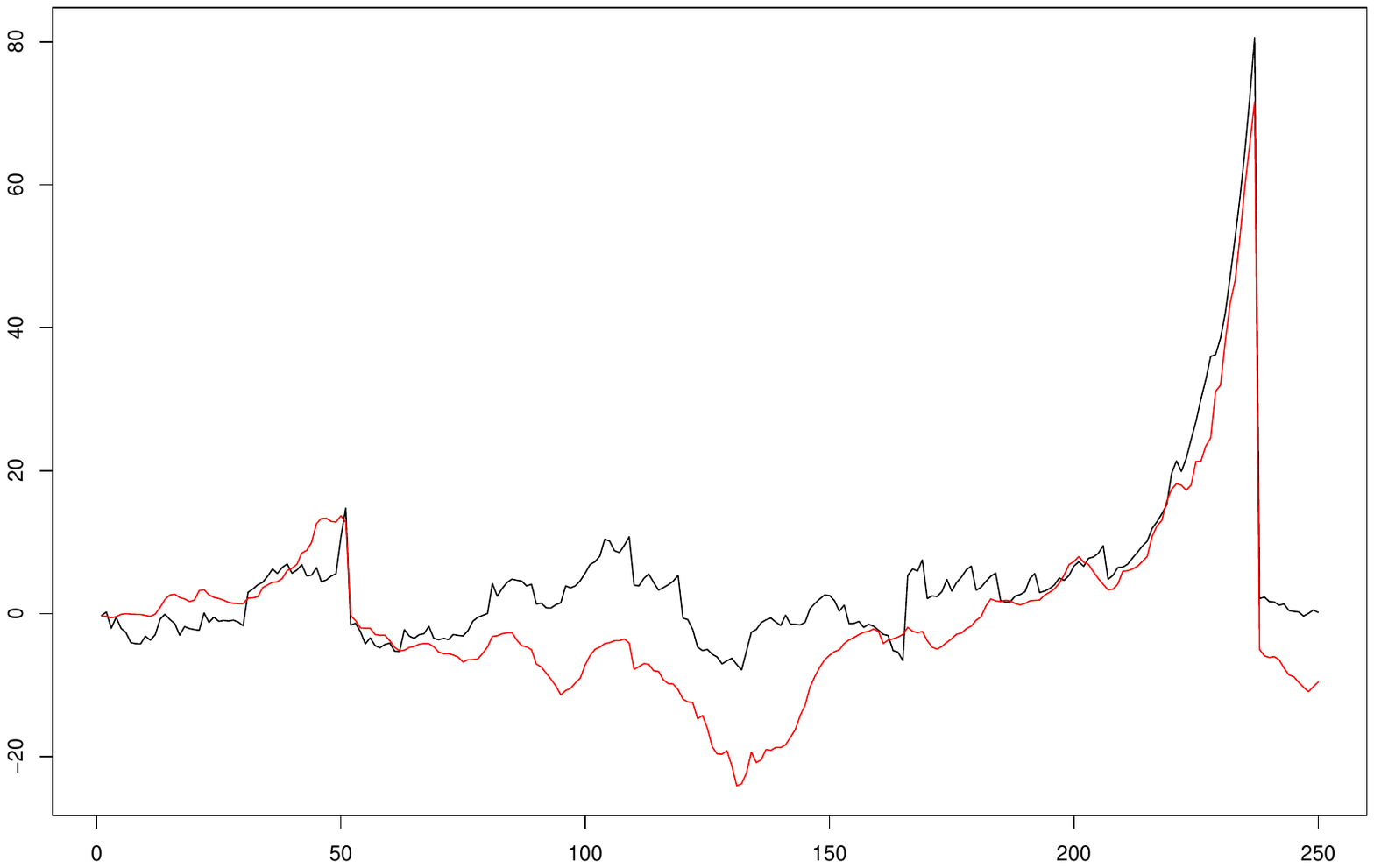}
	\end{subfigure}
	
	\caption{The estimated states of the causal non-causal convolution autoregressive model where $\hat{x}_t^F$ is green, $\hat{x}_t^B$ is red, filter is top, smoother is bottom, and $x_t^F$ and $x_t^B$ are both black.}
	
	\label{fig:stainf1}
	
\end{figure}

It is, however, unsatisfactory that the estimated states are only close to the true states in some periods. A possible explanation is that the causal non-causal convolution autoregressive model has two unobserved components but only one observed component, making it difficult to identify both unobserved components from the observed one. To investigate whether this is a valid explanation, we redo the experiment under the additional assumption that a sample $(\bar{x}_t^F)_{t = 1}^{T}$ from $(\bar{X}_t^F)_{t \in \mathbb{Z}}$, where $(\bar{X}_t^F)_{t \in \mathbb{Z}}$ is given by
\begin{equation}
	\bar{X}_t^F = X_t^F + \bar{\varepsilon}_t^F, \quad \bar{\varepsilon}_t^F \overset{i.i.d.}{\sim} \mathcal{N} (0,\bar{\sigma}_{F,0}^2), \label{eq:encssm}
\end{equation}
is observed. That is, we first study the finite-sample properties of the empirical characteristic function estimator for the extended causal non-causal convolution autoregressive model, $(Y_t)_{t \in \mathbb{Z}}$ and $(\bar{X}_t^F)_{t \in \mathbb{Z}}$, fixing $\alpha = 1$ and $\beta = 0$ in a Monte Carlo simulation study. Then, we estimate the states, $(x_t^F)_{t=1}^{250}$ and $(x_t^B)_{t=1}^{250}$, from the observations, $(y_t)_{t=1}^{250}$ and $(\bar{x}_t^F)_{t = 1}^{250}$, in Figures \ref{fig:cncssm} and \ref{fig:cncssm_extended} using the true parameters.

\begin{figure}[htbp]
	
	\centering
	
	\begin{subfigure}{1\textwidth}
		\centering
		\includegraphics[width=0.75\textwidth,height=0.25\textheight]{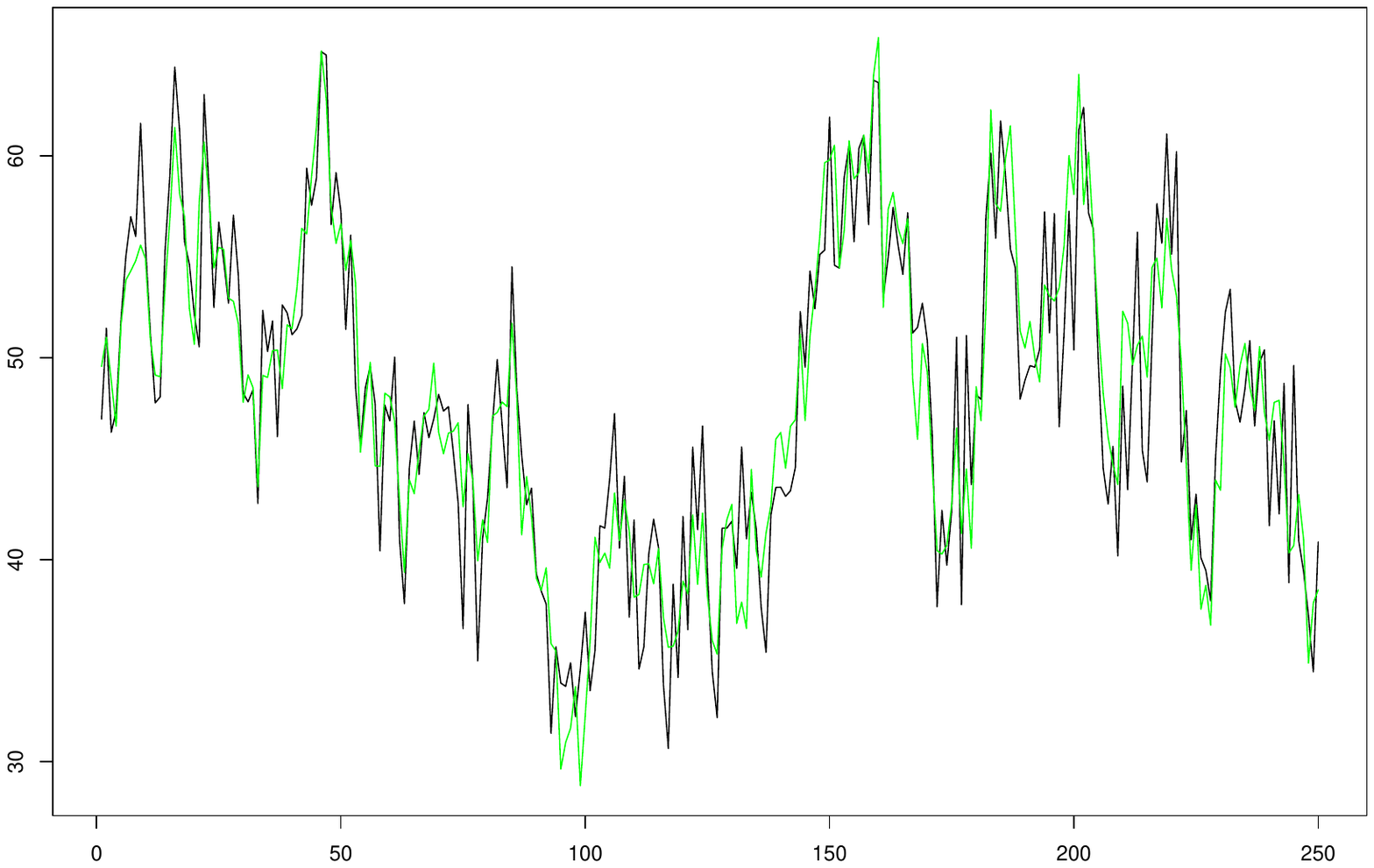}
	\end{subfigure}
	
	\caption{A simulation from the extended causal non-causal convolution autoregressive model in Equation \eqref{eq:encssm} with $\bar{\sigma}_{F}^2 = 10$ where $\bar{X}_t^F$ is black and $X_t^F$ is green.}
	
	\label{fig:cncssm_extended}
	
\end{figure}

Table \ref{tab:parinf2} reports the estimated means and standard deviations (in parenthesis) of the estimated parameters of the extended model obtained in the same way as above with one exemption. Instead of estimating the model from the $T$ observations in one step, it is estimated in two steps as follows. In a first step, the parameters of $(\bar{X}_t^F)_{t \in \mathbb{Z}}$ are estimated from $(\bar{x}_t^F)_{t = 1}^{T}$ using the empirical characteristic function estimator with $p = 1$ and $w(u) = \exp(-u^{\prime}u)$ as above. In a second step, the parameters of $(X_t^B)_{t \in \mathbb{Z}}$ are then estimated from $(y_t)_{t = 1}^{T}$ in the same way where the parameters of $(X_t^F)_{t \in \mathbb{Z}}$ are set equal to the estimated parameters from the first step. The two-step procedure is used in order to avoid minimising a four-dimensional integral, which must be evaluated numerically, numerically. Table \ref{tab:parinf2} also reports the true parameters. The finite-sample properties of the estimated parameters of the extended model are also satisfactory. In fact, the finite-sample properties of the estimated parameters of $(X_t^F)_{t \in \mathbb{Z}}$ in the extended model are better than the ones in the original model in the sense that the estimated means and standard deviations of the estimated parameters of $(X_t^F)_{t \in \mathbb{Z}}$ in the extended model are closer to the true parameters and zero, respectively, than the ones in the original model.

\begin{table}[htbp]
	
	\centering
	
	\begin{tabular}{ccccccc}
		\toprule
		& $\mu_{F,0} = 5$ & $\rho_{F,0} = 0.9$ & $\sigma_{F,0}^2 = 10$ & $\bar{\sigma}^2_{F,0} = 10$ & $\rho_{B,0} = 0.9$ & $\sigma_{B,0} = 0.5$ \\
		\midrule
		$T = 125$ & $\underset{(0.538)}{5.163}$ & $\underset{(0.012)}{0.896}$  & $\underset{(3.797)}{10.318}$ & $\underset{(4.519)}{10.545}$ & $\underset{(0.116)}{0.892}$ & $\underset{(0.661)}{0.570}$ \vspace{0.10cm} \\
		$T = 250$ & $\underset{(0.460)}{5.105}$ & $\underset{(0.010)}{0.898}$  & $\underset{(3.138)}{10.570}$ & $\underset{(3.832)}{10.310}$ & $\underset{(0.103)}{0.895}$ & $\underset{(0.511)}{0.542}$ \vspace{0.10cm} \\
		$T = 500$ & $\underset{(0.344)}{5.093}$ & $\underset{(0.007)}{0.898}$  & $\underset{(2.687)}{10.645}$ & $\underset{(3.444)}{10.409}$ & $\underset{(0.084)}{0.901}$ & $\underset{(0.426)}{0.513}$ \vspace{0.10cm} \\
		\bottomrule
	\end{tabular}
	
	\caption{The estimated means and standard deviations (in parenthesis) of the estimated parameters of the extended causal non-causal convolution autoregressive model together with the true parameters.}
	
	\label{tab:parinf2}
	
\end{table}

Figure \ref{fig:stainf2} shows the estimated states, $(\hat{x}_t^F)_{t=1}^{250}$ (green) and $(\hat{x}_t^B)_{t=1}^{250}$ (red), of the extended model. The estimated states are obtained using the filtered (top) and smoothed (bottom) expectation with $f(x) = 1_{|x| \leq 10^6} x + 1_{|x| > 10^6} 10^6$ for an extended causal non-causal convolution autoregressive model where $\Phi((x^F,x^B),B) = \int_{B} f_{\mathcal{N}} (z_1-x^F-x^B) \frac{1}{\bar{\sigma}_F} f_{\mathcal{N}} \left( \frac{z_2-x^F}{\bar{\sigma}_F} \right) \lambda_2(\textup{d}z)$ implemented via the particle filter and smoother by \cite{GordonSalmondSmith1993} and \cite{KlaasBriersDeFreitasDoucetMaskellLang2006}, respectively, with a thousand particles. Figure \ref{fig:stainf2} also shows the true states, $(x_t^F)_{t =1}^{250}$ and $(x_t^B)_{t =1}^{250}$ (both black). In contrast to before, the estimated states are now close to the true states at all times. We thus deduce that in the extended causal non-causal convolution autoregressive model, which has two unobserved components and two observed ones, both unobserved components can be identified.

\begin{figure}[htbp]
	
	\centering
	
	\begin{subfigure}{0.5\textwidth}
		\centering
		\includegraphics[width=0.9\textwidth,height=0.25\textheight]{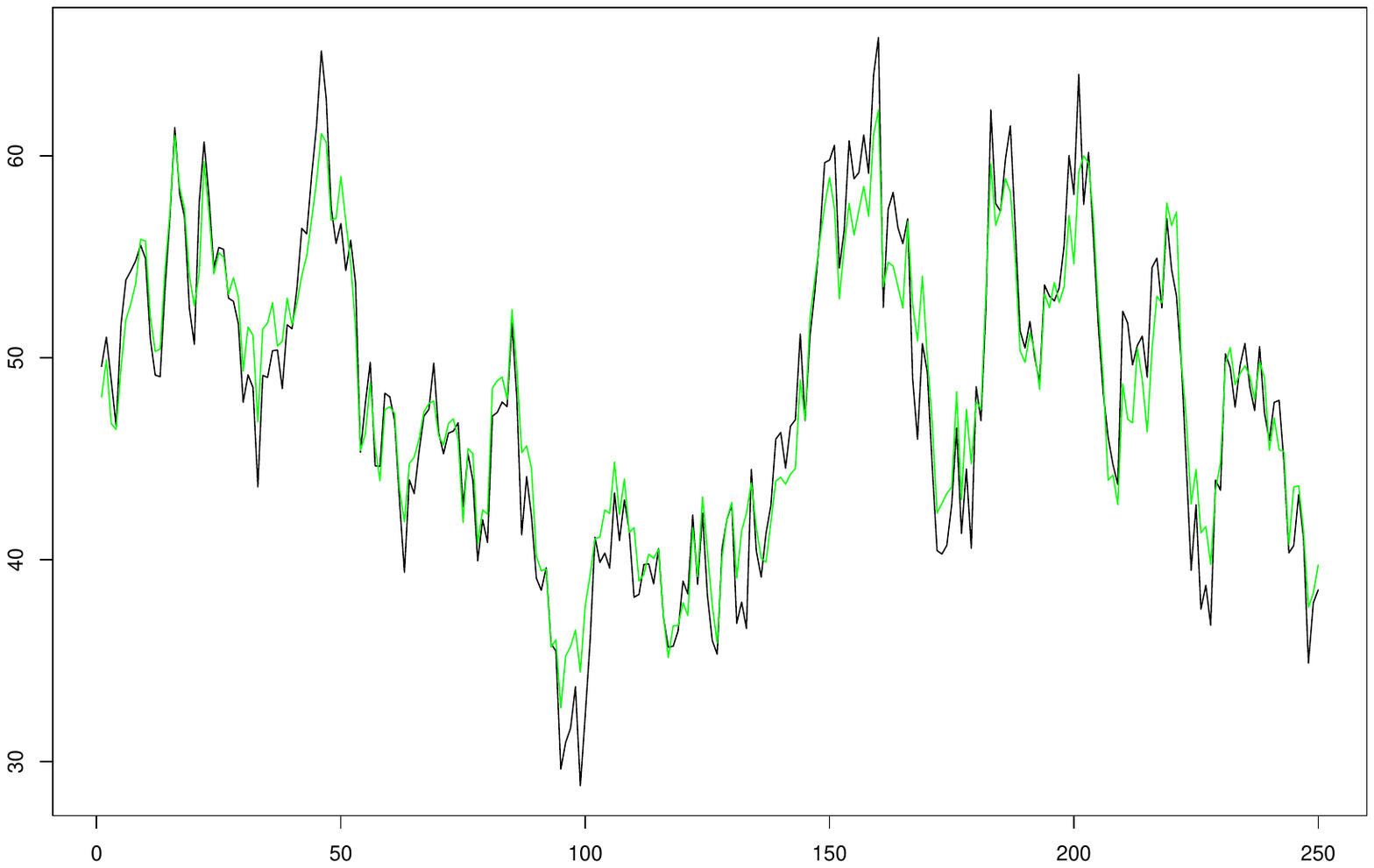}
	\end{subfigure} \kern-2.5em
	\begin{subfigure}{0.5\textwidth}
		\centering
		\includegraphics[width=0.9\textwidth,height=0.25\textheight]{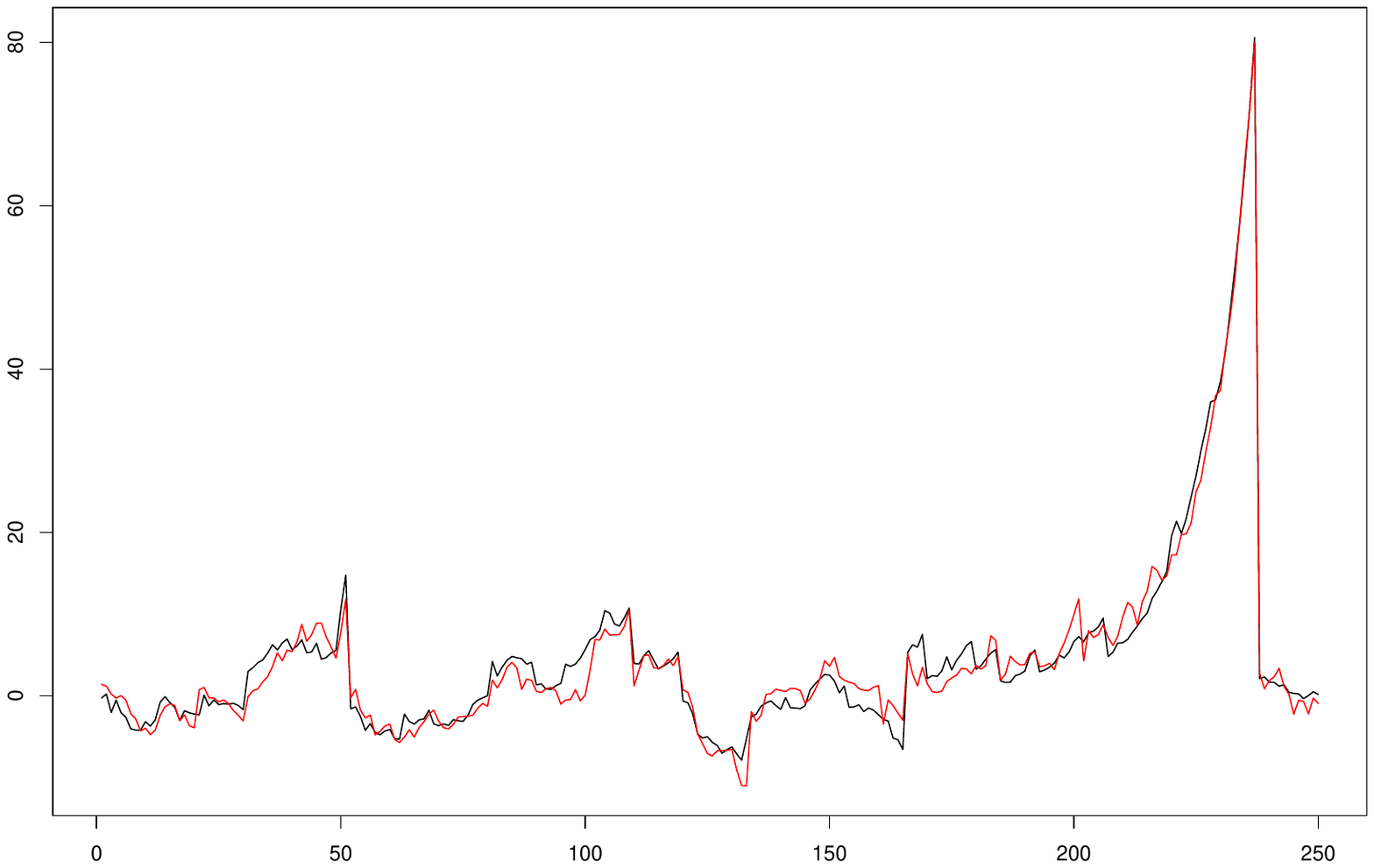}
	\end{subfigure}
	
	\begin{subfigure}{0.5\textwidth}
		\centering
		\includegraphics[width=0.9\textwidth,height=0.25\textheight]{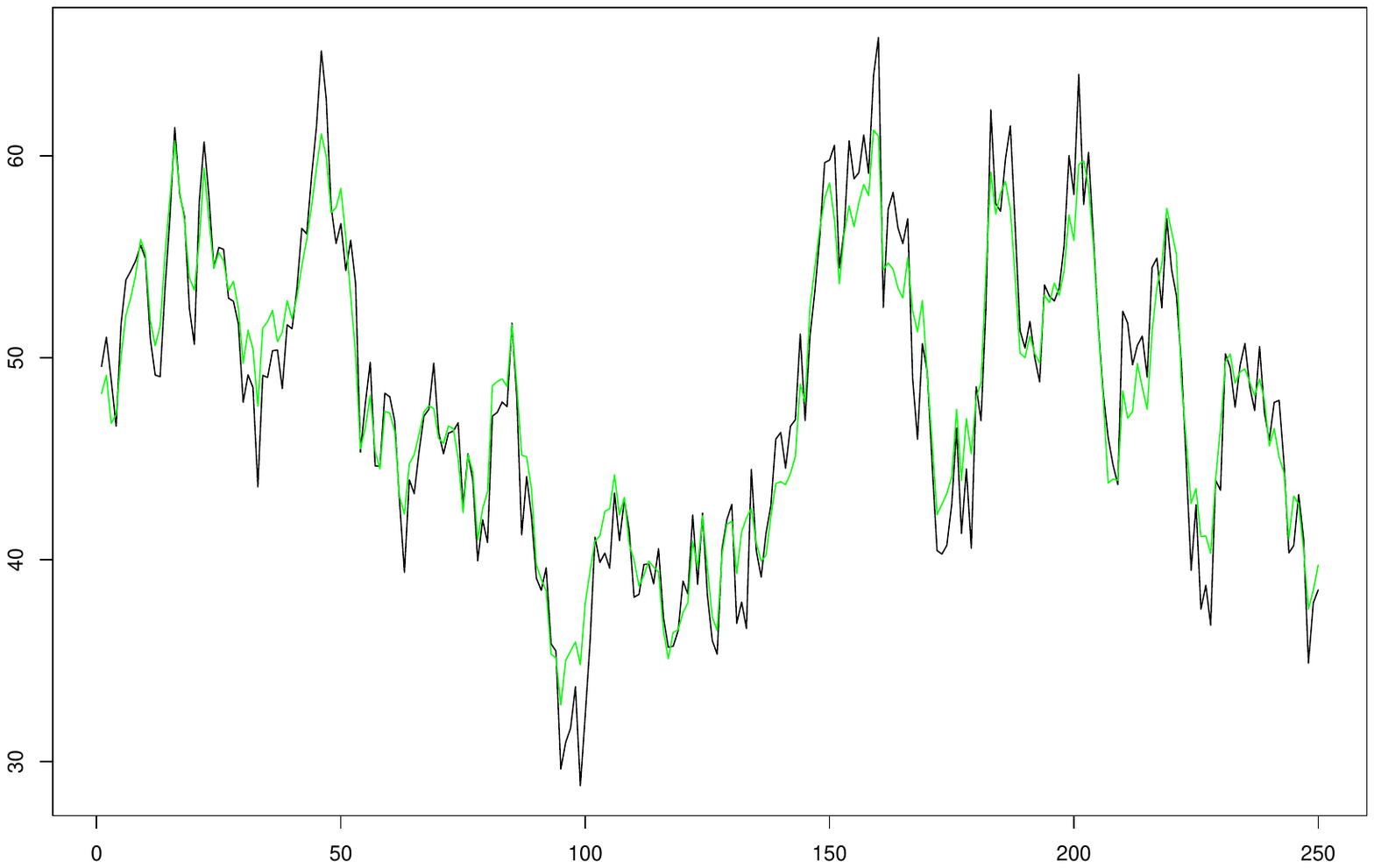}
	\end{subfigure} \kern-2.5em
	\begin{subfigure}{0.5\textwidth}
		\centering
		\includegraphics[width=0.9\textwidth,height=0.25\textheight]{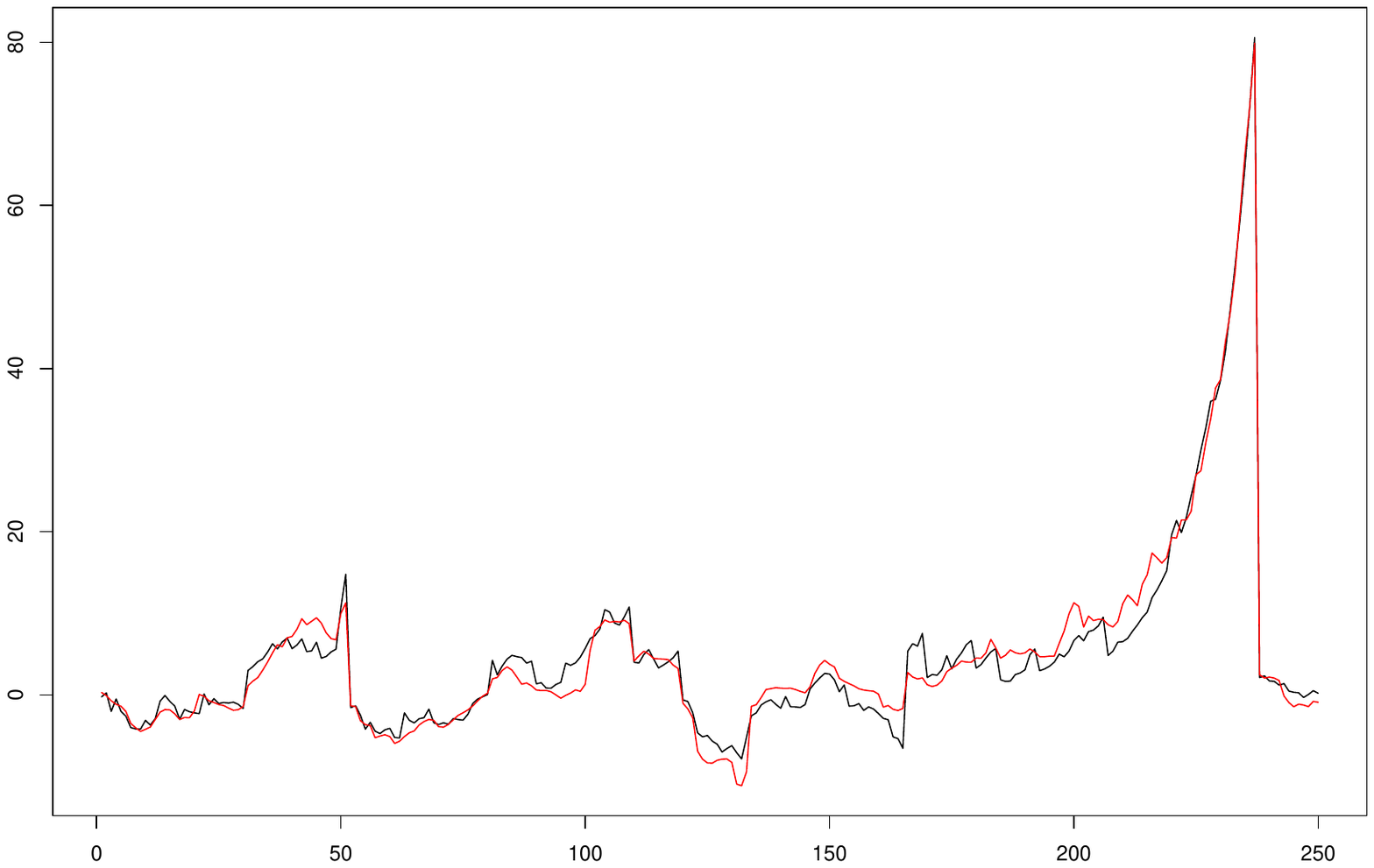}
	\end{subfigure}
	
	\caption{The estimated states of the extended causal non-causal convolution autoregressive model where $\hat{x}_t^F$ is green, $\hat{x}_t^B$ is red, filter is top, smoother is bottom, and $x_t^F$ and $x_t^B$ are both black.}
	
	\label{fig:stainf2}
	
\end{figure}

\section{Application} \label{sec:application}

Many authors have found evidence of a bubble in the NASDAQ index in the 1990s, which is commonly referred to as the dot-com bubble, see, for instance, \cite{PhillipsWuYu2011} and the references therein. In this section, we use the causal non-causal convolution autoregressive model, which is consistent with a rational expectations stock price model as shown in Section \ref{sec:motivation}, to estimate the size of the dot-com bubble in both real time and a posteriori with the non-causal autoregressive model as a benchmark.

The top figure in Figure \ref{fig:data} shows monthly real prices (solid) and dividends (dashed) on the NASDAQ index from December 1973 to December 2003, which we obtain as follows. First, we obtain with-dividend and without-dividend monthly returns on the value-weighted portfolio of all NASDAQ stocks from January 1973 to December 2003 from the Center for Research in Security Prices (CRSP). From without-dividend returns, we obtain nominal prices by initialising the first nominal price at unity. Then, we obtain nominal dividends from year-over-year with-dividend returns, year-over-year without-dividend returns, and nominal prices in order to take the seasonality in the dividends into account. We then obtain the consumer price index from the Federal Reserve Economic Data (FRED) database in order to convert nominal prices and dividends into real prices and dividends. The bottom figure in Figure \ref{fig:data} shows monthly real spreads with $R = 0.007$ on the NASDAQ index where we have set the expected monthly real return ($R$) equal to the difference between the sample mean of monthly nominal return ($1.1\%$) and the sample mean of monthly inflation ($0.4\%$).

\begin{figure}[htbp]
	
	\centering
	
	\begin{subfigure}{1\textwidth}
		\centering
		\includegraphics[width=0.75\textwidth,height=0.25\textheight]{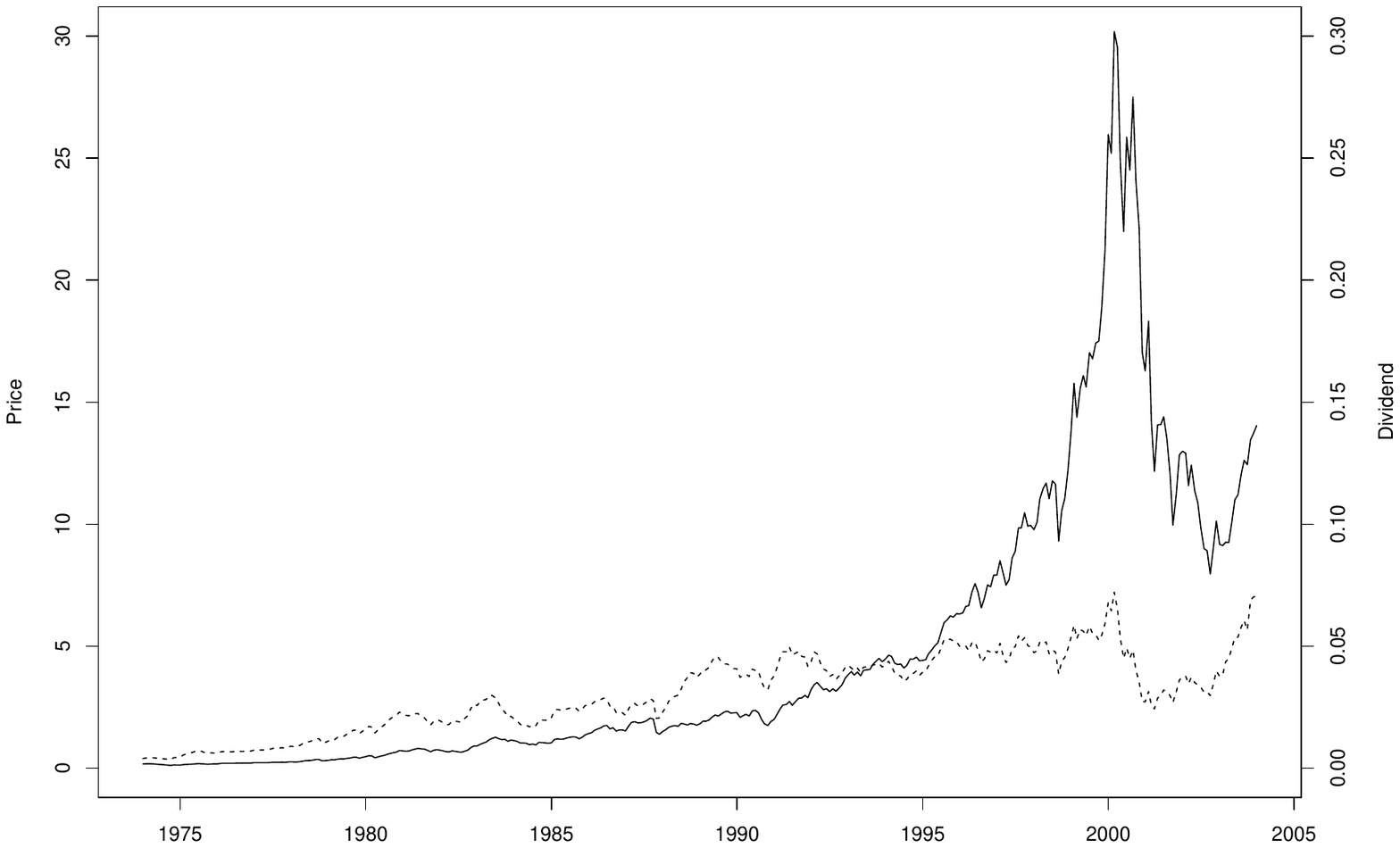}
	\end{subfigure}
	
	\begin{subfigure}{1\textwidth}
		\centering
		\includegraphics[width=0.75\textwidth,height=0.25\textheight]{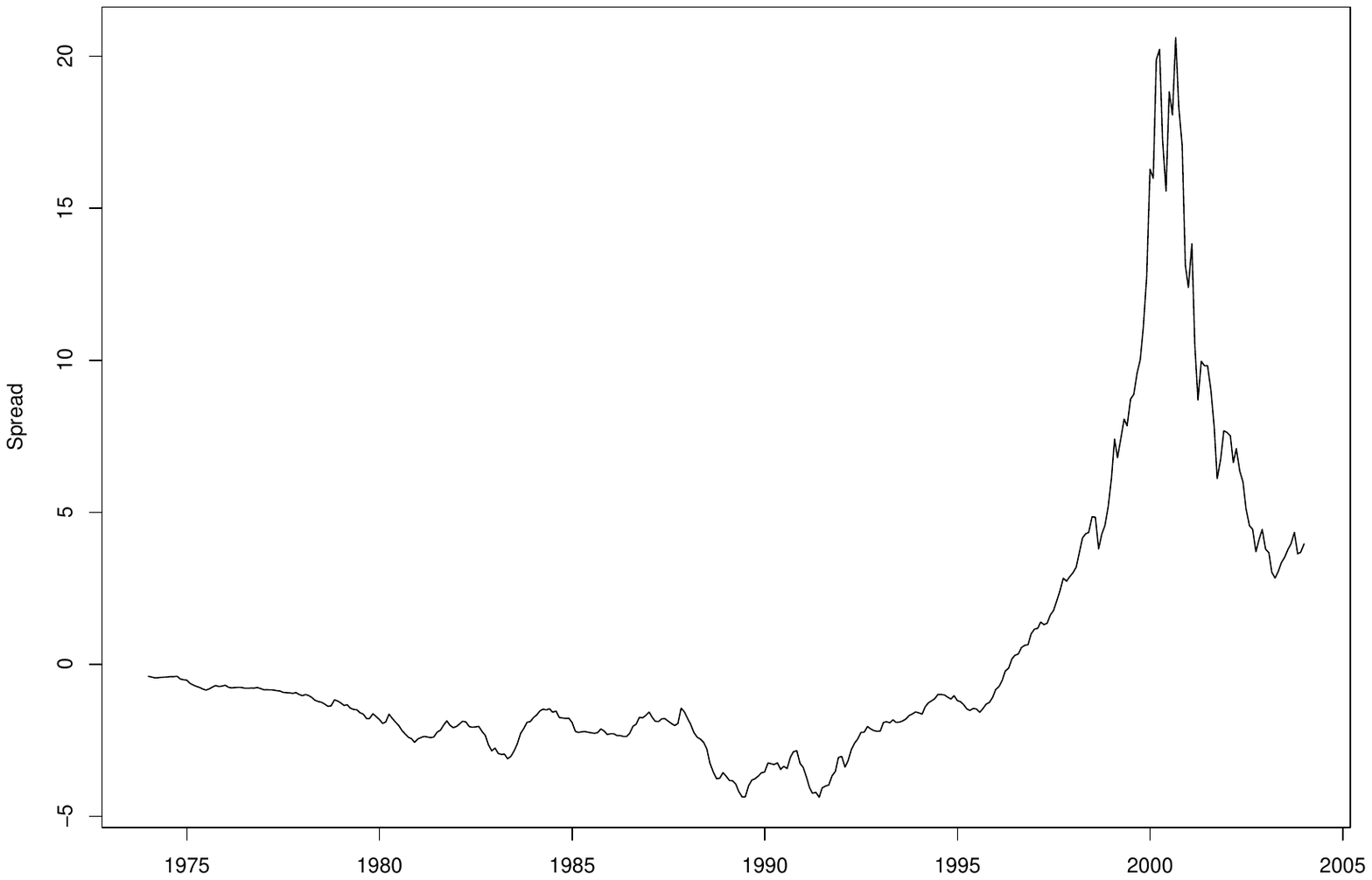}
	\end{subfigure}
	
	\caption{Monthly real prices (top, solid), dividends (top, dashed), and spreads (bottom) on the NASDAQ index from December 1973 to December 2003.}
	
	\label{fig:data}
	
\end{figure}

The first row in Table \ref{tab:parameters} reports the estimated parameters for the causal non-causal convolution autoregressive model with $\rho_B = \left(\frac{1}{1+R}\right)^\frac{1}{1-\alpha}$ and $\beta = 1$, denoted $\mathcal{M}$, obtained using the empirical characteristic function estimator with $p = 1$ and $w(u) = \exp(-u^{\prime}u)$ as in \cite{KnightYu2002}.\footnote{Note that, for all $\alpha \in (0,1)$, Equation \eqref{eq:eqeq} can be solved for $\rho_{B} \in (0,1)$, but that Equation \eqref{eq:eqeq} cannot be solved for $\alpha \in (0,1)$ for all $\rho_{B} \in (0,1)$.} The second row in Table \ref{tab:parameters} reports the estimated parameters for the non-causal autoregressive model with a constant, denoted $\mathcal{M}_b$, obtained in the same way under the additional restrictions that $\rho_F = 0$ and $\sigma_F^2 = 0$. We find that the fundamental value is highly persistent with $\hat{\rho}_F = 0.936$ and that the bubble builds up slowly since $\hat{\rho}_B = 0.877$. If $\rho_F = 0$ and $\sigma_F^2 = 0$, then $\hat{\mu}_F$ drops from $-0.148$ to $-4.701$, which is close to the minimum spread in the sample at $-4.368$, and the bubble builds up even slower since $\hat{\rho}_B = 0.962$.

\begin{table}[htbp]
	
	\centering
	
	\begin{tabular}{ccccccc}
		\toprule
		& $\hat{\mu}_F$ & $\hat{\rho}_F$ & $\hat{\sigma}_F^2$ & $\hat{\rho}_B$ & $\hat{\alpha}$ & $\hat{\sigma}_B$ \\
		\midrule
		$\mathcal{M}$ & $-0.148$ & $0.936$  & $0.278$ & $0.877$ & $0.947$ &  $0.003$ \\
		$\mathcal{M}_b$ & $-4.701$ & $-$  & $-$ & $0.962$ & $0.821$ & $0.012$ \\
		\bottomrule
	\end{tabular}
	
	\caption{The estimated parameters for the causal non-causal convolution autoregressive model ($\mathcal{M}$) and the non-causal autoregressive model with a constant ($\mathcal{M}_b$).}
	
	\label{tab:parameters}
	
\end{table}

The top figure in Figure \ref{fig:states} shows the filtered fundamental value (green) and bubble (red) from the causal non-causal convolution autoregressive model obtained using the filtered expectation with $f(x) = 1_{|x| \leq 10^6} x + 1_{|x| > 10^6} 10^6$ for a causal non-causal convolution autoregressive model where $\Phi((x^F,x^B),B) = \int_{B} f_{\mathcal{N}} (z-x^F-x^B) \lambda(\textup{d}z)$ implemented via the particle filter and smoother by \cite{GordonSalmondSmith1993} and \cite{KlaasBriersDeFreitasDoucetMaskellLang2006} with 2500 particles. The bottom figure in Figure \ref{fig:states} shows the smoothed fundamental value (green) and bubble (red) from the causal non-causal convolution autoregressive model obtained in the same way but using the smoothed expectation. First of all, we find that the causal non-causal convolution autoregressive model is able to detect the bubble in real time. Indeed, at the beginning of the sample, the spread is driven solely by the fundamental value, whereas by the end of the sample, it is solely driven by the bubble. The model can thus be used as a tool to detect bubbles in real time. The overall picture remains the same when estimating the size of the bubble a posteriori, with one notable difference. A posteriori, the increase in the spread at the end of the sample is driven by both the fundamental value and the bubble; mostly, however, by the bubble.

\begin{figure}[htbp]
	
	\centering
	
	\begin{subfigure}{1\textwidth}
		\centering
		\includegraphics[width=0.75\textwidth,height=0.25\textheight]{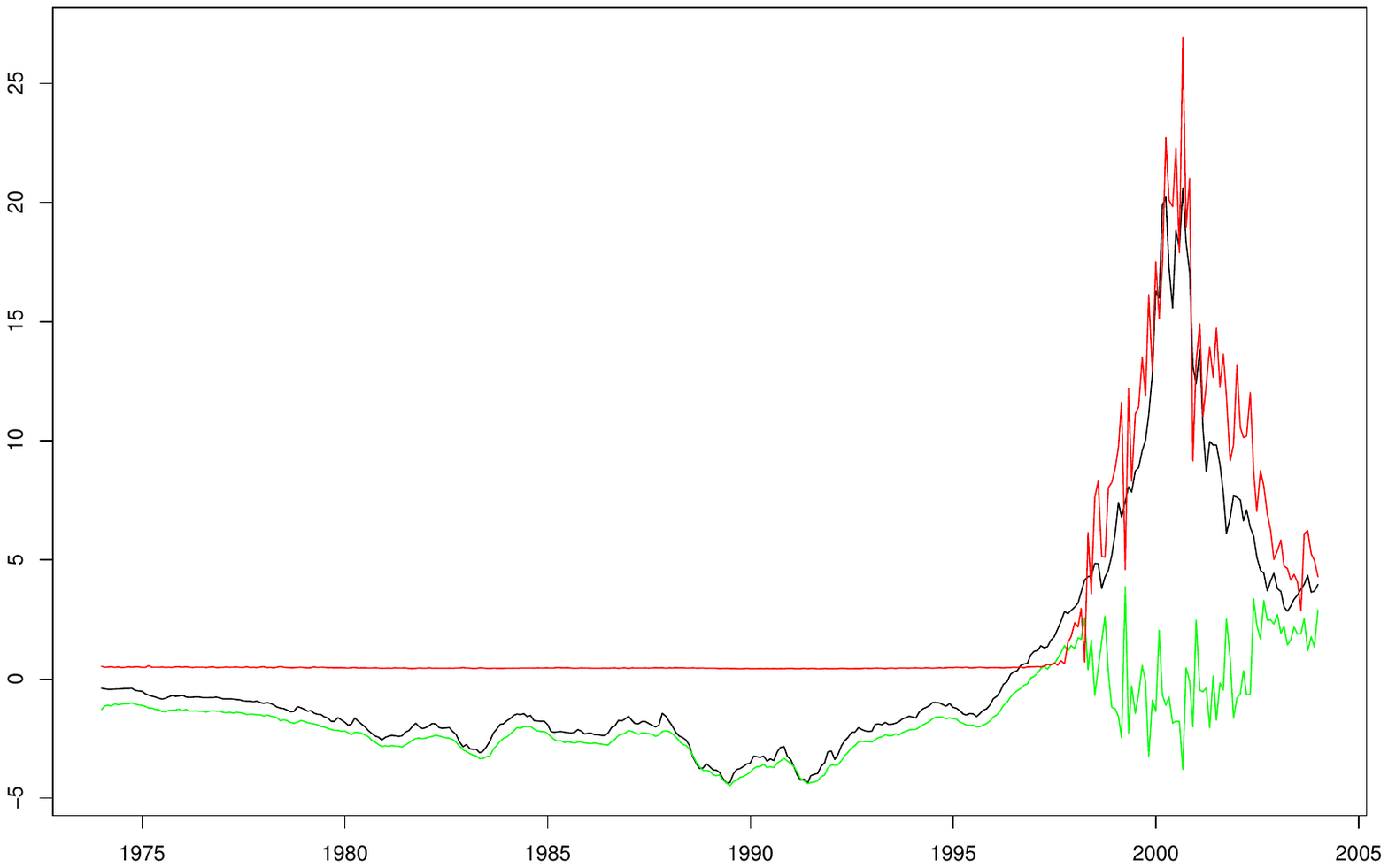}
	\end{subfigure}
	
	\begin{subfigure}{1\textwidth}
		\centering
		\includegraphics[width=0.75\textwidth,height=0.25\textheight]{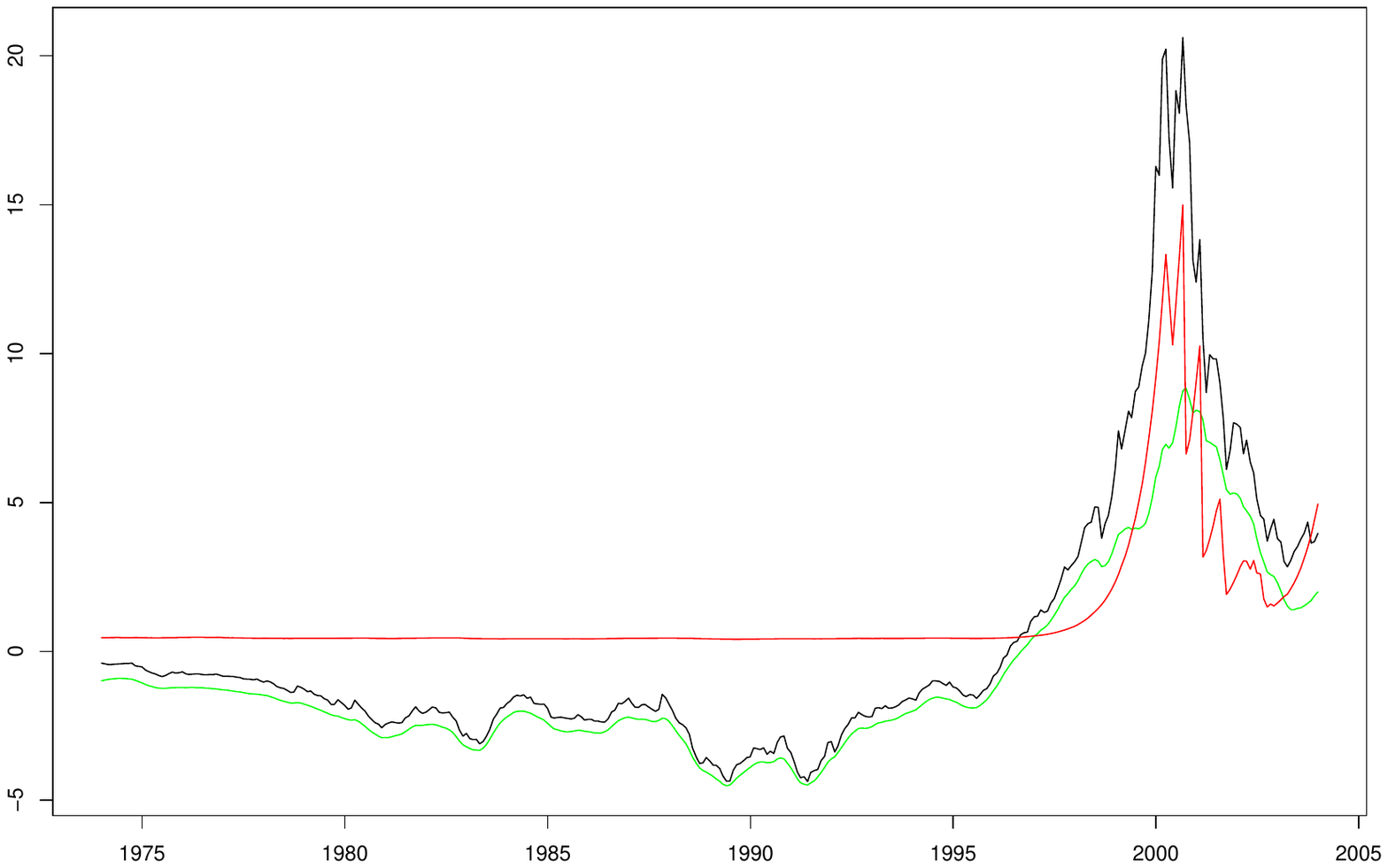}
	\end{subfigure}
	
	\caption{The filtered fundamental value (top, green), filtered bubble (top, red), smoothed fundamental value (bottom, green), smoothed bubble (bottom, red), and spread (black).}
	
	\label{fig:states}
	
\end{figure}

Finally, Figure \ref{fig:comparison} compares the filtered bubble from the causal non-causal convolution autoregressive model (solid) with the one estimated by the non-causal autoregressive model with a constant (dashed); the latter is given by $X_t^B = Y_t - \mu_F$ since $\rho_F = 0$ and $\sigma_F^2 = 0$. Two things are worth mentioning here. First, the bubble estimated by the non-causal autoregressive model with a constant is positive throughout the sample, in contrast to the filtered bubble from the causal non-causal convolution autoregressive model, which is only positive at the end of the sample. At the end of the sample, however, the size of the bubble estimated by the non-causal autoregressive model with a constant is close to identical to the size of the filtered bubble from the causal non-causal convolution autoregressive model.

\begin{figure}[htbp]
	
	\centering

	\includegraphics[width=0.75\textwidth,height=0.25\textheight]{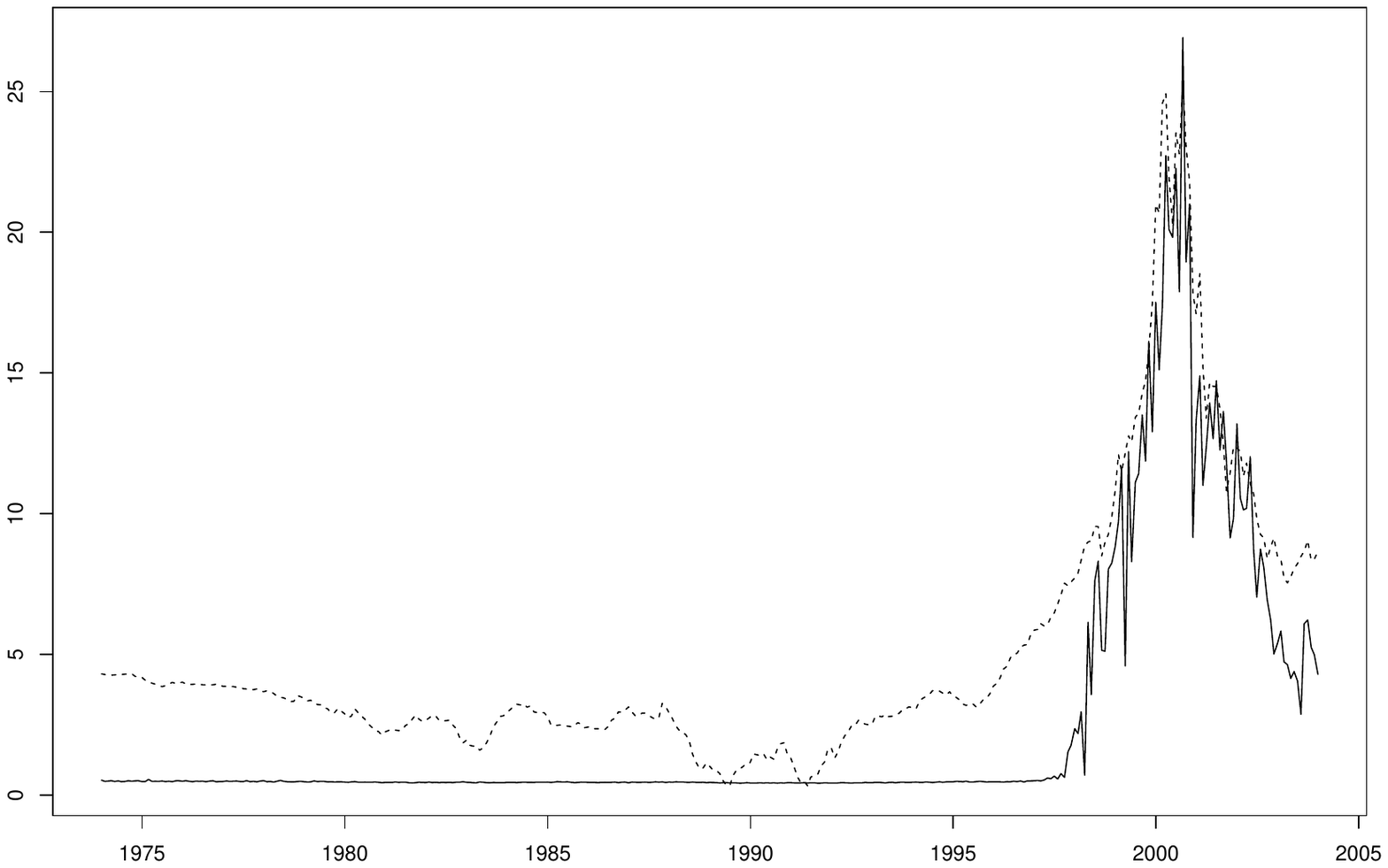}
		
	\caption{The filtered bubble from the causal non-causal convolution autoregressive model (solid) and the bubble estimated by the non-causal autoregressive model with a constant (dashed).}
	
	\label{fig:comparison}
	
\end{figure}

To summarise, the causal non-causal convolution autoregressive model is indeed able to detect the dot-com bubble in both real time and a posteriori. The size of the dot-com bubble estimated in real time is, however, larger than the one estimated a posteriori. The non-causal autoregressive model is also able to detect the dot-com bubble. However, it also detects a bubble before the dot-com bubble, in contrast to the causal non-causal convolution autoregressive model, which contradicts, for instance, the evidence in \cite{PhillipsWuYu2011} and thus shows that the causal non-causal convolution autoregressive model is the more realistic of the two.

\section{Conclusion} \label{sec:conclusion}

Many economic and financial time series are characterised by a local explosive increase followed by a sharp decrease. One example is stock prices where the aforementioned phenomenon can be the result of a bubble that bursts. In recent years, there has been an increasing interest in mixed causal non-causal (vector) autoregressive models to model such time series. Causal non-causal state space models, which we study in this paper, have, however, attracted much less attention. To motivate the use of causal non-causal state space models, we show that the causal non-causal convolution autoregressive model can be consistent with the arguably most fundamental model of stock prices, namely, the rational expectations stock price model, thereby bridging the literature on mixed causal non-causal (vector) autoregressive models with the one on rational expectations stock price models. As in a causal state space model, a central question is how to perform state and parameter inference in the causal non-causal state space model, which we discuss. We also study the causal non-causal convolution autoregressive model in more detail, providing some new results for the model. To illustrate the usefulness of causal non-causal state space models, we use the causal non-causal convolution autoregressive model to estimate the size of the dot-com bubble in both real time and a posteriori with the non-causal autoregressive model as a benchmark and find that the causal non-causal convolution autoregressive model is the more realistic of the two.

\bibliography{references}

\begin{thebibliography}{47}
\providecommand{\natexlab}[1]{#1}
\providecommand{\url}[1]{\texttt{#1}}
\expandafter\ifx\csname urlstyle\endcsname\relax
  \providecommand{\doi}[1]{doi: #1}\else
  \providecommand{\doi}{doi: \begingroup \urlstyle{rm}\Url}\fi

\bibitem[Andreasen and Bro(2024)]{AndreasenBro2024}
M.~M. Andreasen and J.~Bro.
\newblock Identifying a stock price bubble and its macroeconomic implications.
\newblock \emph{SSRN}, 2024.

\bibitem[Bansal and Yaron(2004)]{BansalYaron2004}
R.~Bansal and A.~Yaron.
\newblock Risks for the long run: A potential resolution of asset pricing
  puzzles.
\newblock \emph{The journal of Finance}, 59\penalty0 (4):\penalty0 1481--1509,
  2004.

\bibitem[Brockwell and Davis(1991)]{BrockwellDavis1991}
P.~J. Brockwell and R.~A. Davis.
\newblock \emph{Time series: theory and methods}.
\newblock Springer science \& business media, 1991.

\bibitem[Cambanis and Fakhre-Zakeri(1995)]{CambanisFakhre-Zakeri1995}
S.~Cambanis and I.~Fakhre-Zakeri.
\newblock On prediction of heavy-tailed autoregressive sequences: forward
  versus reversed time.
\newblock \emph{Theory of Probability \& Its Applications}, 39\penalty0
  (2):\penalty0 217--233, 1995.

\bibitem[Campbell and Cochrane(1999)]{CampbellCochrane1999}
J.~Y. Campbell and J.~H. Cochrane.
\newblock By force of habit: A consumption-based explanation of aggregate stock
  market behavior.
\newblock \emph{Journal of political Economy}, 107\penalty0 (2):\penalty0
  205--251, 1999.

\bibitem[Campbell and Shiller(1988)]{CampbellShiller1988}
J.~Y. Campbell and R.~J. Shiller.
\newblock The dividend-price ratio and expectations of future dividends and
  discount factors.
\newblock \emph{The review of financial studies}, 1\penalty0 (3):\penalty0
  195--228, 1988.

\bibitem[Campbell et~al.(1997)Campbell, Lo, and
  MacKinlay]{CampbellLoMacKinlay1997}
J.~Y. Campbell, A.~W. Lo, and A.~C. MacKinlay.
\newblock \emph{The Econometrics of Financial Markets}.
\newblock Princeton University Press, 1997.

\bibitem[Capp{\'e} et~al.(2005)Capp{\'e}, Moulines, and
  Ryd{\'e}n]{CappeMoulinesRyden2005}
O.~Capp{\'e}, E.~Moulines, and T.~Ryd{\'e}n.
\newblock \emph{Inference in Hidden Markov Models}.
\newblock Springer New York, NY, 2005.

\bibitem[Cavaliere et~al.(2020)Cavaliere, Nielsen, and
  Rahbek]{CavaliereNielsenRahbek2020}
G.~Cavaliere, H.~B. Nielsen, and A.~Rahbek.
\newblock Bootstrapping noncausal autoregressions: with applications to
  explosive bubble modeling.
\newblock \emph{Journal of Business \& Economic Statistics}, 38\penalty0
  (1):\penalty0 55--67, 2020.

\bibitem[Chopin et~al.(2020)Chopin, Papaspiliopoulos,
  et~al.]{ChopinPapaspiliopoulos2020}
N.~Chopin, O.~Papaspiliopoulos, et~al.
\newblock \emph{An introduction to sequential Monte Carlo}, volume~4.
\newblock Springer, 2020.

\bibitem[Cochrane(2008)]{Cochrane2008}
J.~H. Cochrane.
\newblock State-space vs. var models for stock returns.
\newblock \emph{Unpublished paper, Chicago GSB}, 2008.

\bibitem[Davis and Song(2020)]{DavisSong2020}
R.~A. Davis and L.~Song.
\newblock Noncausal vector ar processes with application to economic time
  series.
\newblock \emph{Journal of Econometrics}, 216\penalty0 (1):\penalty0 246--267,
  2020.

\bibitem[Fama and French(1988)]{FamaFrench1988}
E.~F. Fama and K.~R. French.
\newblock Dividend yields and expected stock returns.
\newblock \emph{Journal of financial economics}, 22\penalty0 (1):\penalty0
  3--25, 1988.

\bibitem[Ferson et~al.(2003)Ferson, Sarkissian, and
  Simin]{FersonSarkissianSimin2003}
W.~E. Ferson, S.~Sarkissian, and T.~T. Simin.
\newblock Spurious regressions in financial economics?
\newblock \emph{The Journal of Finance}, 58\penalty0 (4):\penalty0 1393--1413,
  2003.

\bibitem[Feuerverger(1990)]{Feuerverger1990}
A.~Feuerverger.
\newblock An efficiency result for the empirical characteristic function in
  stationary time-series models.
\newblock \emph{The Canadian Journal of Statistics}, pages 155--161, 1990.

\bibitem[Feuerverger and Mureika(1977)]{FeuervergerMureika1977}
A.~Feuerverger and R.~A. Mureika.
\newblock The empirical characteristic function and its applications.
\newblock \emph{The annals of Statistics}, pages 88--97, 1977.

\bibitem[Francq and Meintanis(2016)]{FrancqMeintanis2016}
C.~Francq and S.~G. Meintanis.
\newblock Fourier-type estimation of the power garch model with stable-paretian
  innovations.
\newblock \emph{Metrika}, 79\penalty0 (4):\penalty0 389--424, 2016.

\bibitem[Fries(2022)]{Fries2022}
S.~Fries.
\newblock Conditional moments of noncausal alpha-stable processes and the
  prediction of bubble crash odds.
\newblock \emph{Journal of Business \& Economic Statistics}, 40\penalty0
  (4):\penalty0 1596--1616, 2022.

\bibitem[Fries and Zakoian(2019)]{FriesZakoian2019}
S.~Fries and J.-M. Zakoian.
\newblock Mixed causal-noncausal ar processes and the modelling of explosive
  bubbles.
\newblock \emph{Econometric Theory}, 35\penalty0 (6):\penalty0 1234--1270,
  2019.

\bibitem[Gordon et~al.(1993)Gordon, Salmond, and Smith]{GordonSalmondSmith1993}
N.~J. Gordon, D.~J. Salmond, and A.~F. Smith.
\newblock Novel approach to nonlinear/non-gaussian bayesian state estimation.
\newblock In \emph{IEE proceedings F (radar and signal processing)}, volume
  140, pages 107--113. IET, 1993.

\bibitem[Gourieroux and Jasiak(2016)]{GourierouxJasiak2016}
C.~Gourieroux and J.~Jasiak.
\newblock Filtering, prediction and simulation methods for noncausal processes.
\newblock \emph{Journal of Time Series Analysis}, 37\penalty0 (3):\penalty0
  405--430, 2016.

\bibitem[Gourieroux and Jasiak(2017)]{GourierouxJasiak2017}
C.~Gourieroux and J.~Jasiak.
\newblock Noncausal vector autoregressive process: Representation,
  identification and semi-parametric estimation.
\newblock \emph{Journal of Econometrics}, 200\penalty0 (1):\penalty0 118--134,
  2017.

\bibitem[Gouri{\'e}roux and Zako{\"\i}an(2017)]{GourierouxZakoian2017}
C.~Gouri{\'e}roux and J.-M. Zako{\"\i}an.
\newblock Local explosion modelling by non-causal process.
\newblock \emph{Journal of the Royal Statistical Society Series B: Statistical
  Methodology}, 79\penalty0 (3):\penalty0 737--756, 2017.

\bibitem[Gourieroux et~al.(2020)Gourieroux, Jasiak, and
  Monfort]{GourierouxJasiakMonfort2020}
C.~Gourieroux, J.~Jasiak, and A.~Monfort.
\newblock Stationary bubble equilibria in rational expectation models.
\newblock \emph{Journal of Econometrics}, 218\penalty0 (2):\penalty0 714--735,
  2020.

\bibitem[Gouri{\'e}roux et~al.(2021)Gouri{\'e}roux, Jasiak, and
  Tong]{GourierouxJasiakTong2021}
C.~Gouri{\'e}roux, J.~Jasiak, and M.~Tong.
\newblock Convolution-based filtering and forecasting: An application to wti
  crude oil prices.
\newblock \emph{Journal of Forecasting}, 40\penalty0 (7):\penalty0 1230--1244,
  2021.

\bibitem[Heathcote(1977)]{Heathcote1977}
C.~R. Heathcote.
\newblock The integrated squared error estimation of parameters.
\newblock \emph{Biometrika}, 64\penalty0 (2):\penalty0 255--264, 1977.

\bibitem[Hecq et~al.(2020)Hecq, Issler, and Telg]{HecqIsslerTelg2020}
A.~Hecq, J.~V. Issler, and S.~Telg.
\newblock Mixed causal--noncausal autoregressions with exogenous regressors.
\newblock \emph{Journal of Applied Econometrics}, 35\penalty0 (3):\penalty0
  328--343, 2020.

\bibitem[Klaas et~al.(2006)Klaas, Briers, De~Freitas, Doucet, Maskell, and
  Lang]{KlaasBriersDeFreitasDoucetMaskellLang2006}
M.~Klaas, M.~Briers, N.~De~Freitas, A.~Doucet, S.~Maskell, and D.~Lang.
\newblock Fast particle smoothing: If i had a million particles.
\newblock In \emph{Proceedings of the 23rd International Conference on Machine
  Learning}, pages 481--488, 2006.

\bibitem[Knight and Yu(2002)]{KnightYu2002}
J.~L. Knight and J.~Yu.
\newblock Empirical characteristic function in time series estimation.
\newblock \emph{Econometric Theory}, 18\penalty0 (3):\penalty0 691--721, 2002.

\bibitem[Kogon and Williams(1998)]{KogonWilliams1998}
S.~M. Kogon and D.~B. Williams.
\newblock Characteristic function based estimation of stable distribution
  parameters.
\newblock \emph{A practical guide to heavy tails: statistical techniques and
  applications}, pages 311--338, 1998.

\bibitem[Lanne and Luoto(2013)]{LanneLuoto2013}
M.~Lanne and J.~Luoto.
\newblock Autoregression-based estimation of the new keynesian phillips curve.
\newblock \emph{Journal of Economic Dynamics and Control}, 37\penalty0
  (3):\penalty0 561--570, 2013.

\bibitem[Lanne and Luoto(2017)]{LanneLuoto2017}
M.~Lanne and J.~Luoto.
\newblock A new time-varying parameter autoregressive model for us inflation
  expectations.
\newblock \emph{Journal of Money, Credit and Banking}, 49\penalty0
  (5):\penalty0 969--995, 2017.

\bibitem[Lanne and Saikkonen(2011)]{LanneSaikkonen2011}
M.~Lanne and P.~Saikkonen.
\newblock Noncausal autoregressions for economic time series.
\newblock \emph{Journal of Time Series Econometrics}, 3\penalty0 (3), 2011.

\bibitem[Lanne and Saikkonen(2013)]{LanneSaikkonen2013}
M.~Lanne and P.~Saikkonen.
\newblock Noncausal vector autoregression.
\newblock \emph{Econometric Theory}, 29\penalty0 (3):\penalty0 447--481, 2013.

\bibitem[Lanne et~al.(2012)Lanne, Luoto, and
  Saikkonen]{LanneLuotoSaikkonen2012}
M.~Lanne, J.~Luoto, and P.~Saikkonen.
\newblock Optimal forecasting of noncausal autoregressive time series.
\newblock \emph{International Journal of Forecasting}, 28\penalty0
  (3):\penalty0 623--631, 2012.

\bibitem[Lettau and Ludvigson(2005)]{LettauLudvigson2005}
M.~Lettau and S.~C. Ludvigson.
\newblock Expected returns and expected dividend growth.
\newblock \emph{Journal of Financial Economics}, 76\penalty0 (3):\penalty0
  583--626, 2005.

\bibitem[Meintanis and Taufer(2012)]{MeintanisTaufer2012}
S.~G. Meintanis and E.~Taufer.
\newblock Inference procedures for stable-paretian stochastic volatility
  models.
\newblock \emph{Mathematical and Computer Modelling}, 55\penalty0
  (3-4):\penalty0 1199--1212, 2012.

\bibitem[Menzly et~al.(2004)Menzly, Santos, and
  Veronesi]{MenzlySantosVeronesi2004}
L.~Menzly, T.~Santos, and P.~Veronesi.
\newblock Understanding predictability.
\newblock \emph{Journal of Political Economy}, 112\penalty0 (1):\penalty0
  1--47, 2004.

\bibitem[Olver et~al.(2010)Olver, Lozier, Boisvert, and
  Clark]{OlverLozierBoisvertClark2010}
F.~W. Olver, D.~W. Lozier, R.~F. Boisvert, and C.~W. Clark.
\newblock \emph{NIST handbook of mathematical functions}.
\newblock Cambridge university press, 2010.

\bibitem[Parzen(1962)]{Parzen1962}
E.~Parzen.
\newblock On estimation of a probability density function and mode.
\newblock \emph{The annals of mathematical statistics}, 33\penalty0
  (3):\penalty0 1065--1076, 1962.

\bibitem[P{\'a}stor and Stambaugh(2009)]{PastorStambaugh2009}
L.~P{\'a}stor and R.~F. Stambaugh.
\newblock Predictive systems: Living with imperfect predictors.
\newblock \emph{The Journal of Finance}, 64\penalty0 (4):\penalty0 1583--1628,
  2009.

\bibitem[Paulson et~al.(1975)Paulson, Holcomb, and
  Leitch]{PaulsonHolcombEdwardLeitch1975}
A.~S. Paulson, E.~W. Holcomb, and R.~A. Leitch.
\newblock The estimation of the parameters of the stable laws.
\newblock \emph{Biometrika}, 62\penalty0 (1):\penalty0 163--170, 1975.

\bibitem[Phillips et~al.(2011)Phillips, Wu, and Yu]{PhillipsWuYu2011}
P.~C. Phillips, Y.~Wu, and J.~Yu.
\newblock Explosive behavior in the 1990s nasdaq: When did exuberance escalate
  asset values?
\newblock \emph{International economic review}, 52\penalty0 (1):\penalty0
  201--226, 2011.

\bibitem[Press(1972)]{Press1972}
S.~J. Press.
\newblock Estimation in univariate and multivariate stable distributions.
\newblock \emph{Journal of the American statistical association}, 67\penalty0
  (340):\penalty0 842--846, 1972.

\bibitem[Samorodnitsky and Taqqu(1994)]{SamorodnitskyTaqqu1994}
G.~Samorodnitsky and M.~S. Taqqu.
\newblock \emph{Stable Non-Gaussian Random Processes}.
\newblock 1994.

\bibitem[Van~Binsbergen and Koijen(2010)]{VanBinsbergenKoijen2010}
J.~H. Van~Binsbergen and R.~S. Koijen.
\newblock Predictive regressions: A present-value approach.
\newblock \emph{The Journal of Finance}, 65\penalty0 (4):\penalty0 1439--1471,
  2010.

\bibitem[Yu(2004)]{Yu2004}
J.~Yu.
\newblock Empirical characteristic function estimation and its applications.
\newblock \emph{Econometric reviews}, 23\penalty0 (2):\penalty0 93--123, 2004.

\end{thebibliography}

\section*{Appendix}

\appendix

\section{More Details on the Motivating Example} \label{appendix:motivatingexample}

\subsubsection{Equation \eqref{eq:rem2}}

The derivation consists of two parts. First, we show that $F_t^S$ can be written as a linear function of $G_{t} - \phi_0^{G}$. Then, by using that $F_t^S$ can be written as a linear function of $G_{t} - \phi_0^{G}$ and that $G_{t} - \phi_0^{G}$ follows a first-order Gaussian autoregressive model, we show that $F_t^S$ too can be written as a first-order Gaussian autoregressive model. We have that
\begin{align*}
	F_t^S &= \frac{1}{R} \sum_{i=0}^{\infty} \left( \frac{1}{1+R} \right)^i \mathbb{E}_t [\Delta D_{t+1+i}] \\ 
	&= \frac{1}{R} \sum_{i=0}^{\infty} \left( \frac{1}{1+R} \right)^i \mathbb{E}_t [\mathbb{E}_{t+i} [\Delta D_{t+1+i}]] \\
	&= \frac{1}{R} \sum_{i=0}^{\infty} \left( \frac{1}{1+R} \right)^i \mathbb{E}_t [G_{t+i}] \\
	&= \frac{1}{R} \frac{1+R}{R} \phi_0^{G} + \frac{1}{R} \sum_{i=0}^{\infty} \left( \frac{1}{1+R} \right)^i (\mathbb{E}_t [G_{t+i}] - \phi_0^{G}) \\
	&= \frac{1}{R} \frac{1+R}{R} \phi_0^{G} + \frac{1}{R} \sum_{i=0}^{\infty} \left( \frac{1}{1+R} \right)^i (\phi_1^{G})^{i} (G_{t} - \phi_0^{G}) \\
	&= \frac{1}{R} \frac{1+R}{R} \phi_0^{G} + \frac{1}{R} \frac{1+R}{1+R-\phi_1^{G}} (G_{t} - \phi_0^{G}),
\end{align*}
and thus that
\begin{align*}
	F_t^S &= \frac{1}{R} \frac{1+R}{R} \phi_0^{G} + \frac{1}{R} \frac{1+R}{1+R-\phi_1^{G}} (G_{t} - \phi_0^{G}) \\
	&= \frac{1}{R} \frac{1+R}{R} \phi_0^{G} + \frac{1}{R} \frac{1+R}{1+R-\phi_1^{G}} \{\phi_1^{G} (G_{t-1} - \phi_0^{G}) + \varepsilon_t^{G}\} \\
	&= \frac{1}{R} \frac{1+R}{R} \phi_0^{G} + \phi_1^{G} \frac{1}{R} \frac{1+R}{1+R-\phi_1^{G}} (G_{t-1} - \phi_0^{G}) + \frac{1}{R} \frac{1+R}{1+R-\phi_1^{G}} \varepsilon_t^{G} \\
	&= \frac{1}{R} \frac{1+R}{R} \phi_0^{G} + \phi_1^{G} \left\{ - \frac{1}{R} \frac{1+R}{R} \phi_0^{G} + F_{t-1}^S \right\} + \frac{1}{R} \frac{1+R}{1+R-\phi_1^{G}} \varepsilon_t^{G} \\
	&= \frac{1}{R} \frac{1+R}{R} \phi_0^{G} (1-\phi_1^{G}) + \phi_1^{G} F_{t-1}^S + \frac{1}{R} \frac{1+R}{1+R-\phi_1^{G}} \varepsilon_t^{G}.
\end{align*}
This concludes the derivation.

\subsubsection{Remark \ref{remark:transformation}}

Recall that the causal non-causal convolution autoregressive model is given by
\begin{equation*}
	Y_t = X_t^{F} + X_{t}^{B},
\end{equation*}
where
\begin{equation*}
	X_t^{F} = \mu_F + \rho_F X_{t-1}^{F} + \varepsilon_t^{F}, \quad \varepsilon_t^F \overset{i.i.d.}{\sim} \mathcal{N} (0,\sigma_F^2),
\end{equation*}
with $| \rho_F | < 1$ and 
\begin{equation*}
	X_t^{B} = \rho_B X_{t+1}^{B} + \varepsilon_t^{B}, \quad \varepsilon_t^B \overset{i.i.d.}{\sim} \mathcal{S} (\alpha,0,\sigma_B,\beta),
\end{equation*}
with $| \rho_B | < 1$ where $\varepsilon_s^F$ and $\varepsilon_t^B$ are independent for all $s$ and $t$, which is a causal non-causal state space model where the observation kernel do not have a density as shown in Example \ref{ex:un}. Note however that, by Proposition 4.4.2 in \cite{BrockwellDavis1991}, the causal autoregressive model
\begin{equation*}
	X_t^{F} = \mu_F + \rho_F X_{t-1}^{F} + \varepsilon_t^{F}
\end{equation*}
can be rewritten as a non-causal autoregressive model
\begin{equation*}
	X_t^{F} = \mu_F + \rho_F X_{t+1}^{F} + \varepsilon_t^{F},
\end{equation*}
so
\begin{align*}
	Y_t &= X_t^{F} + X_{t}^{B} \\
	&= \{\mu_F + \rho_F X_{t+1}^{F} + \varepsilon_t^{F}\} + X_{t}^{B} \\
	&= \mu_F + \rho_F \{Y_{t+1} - X_{t+1}^{B}\} + \varepsilon_t^{F} + X_{t}^{B} \\
	&= \mu_F + \rho_F Y_{t+1} + \begin{pmatrix} 1 & -\rho_F \end{pmatrix} \begin{pmatrix} X_{t}^{B} \\ X_{t+1}^{B} \end{pmatrix} + \varepsilon_t^{F}.
\end{align*}
The causal non-causal convolution autoregressive model can thus be written as
\begin{equation*}
	Y_t = \mu_F + \rho_F Y_{t+1} + \begin{pmatrix} 1 & -\rho_F \end{pmatrix} \begin{pmatrix} X_{t}^{B} \\ X_{t+1}^{B} \end{pmatrix} + \varepsilon_t^{F},
\end{equation*}
where
\begin{equation*}
	\begin{pmatrix} X_{t}^{B} \\ X_{t+1}^{B} \end{pmatrix} = \begin{pmatrix} \rho_B & 0 \\ 1 & 0 \end{pmatrix} \begin{pmatrix} X_{t+1}^{B} \\ X_{t+2}^{B} \end{pmatrix} + \begin{pmatrix} \varepsilon_t^B \\ 0 \end{pmatrix},
\end{equation*}
which is a non-causal autoregressive state space model where the observation kernel do have a density.

\section{Proofs} \label{app:proofs}

\begin{proof}[Proof of Lemma \ref{lem:cncssm}]
	Let $f \in \mathcal{M}_{b}((\accentset{\rightarrow}{\textup{X}} \times \accentset{\leftarrow}{\textup{X}} \times \textup{Y})^T)$ be given. Then,
	\begin{align*}
		\mathbb{E} [ f(\accentset{\rightarrow}{X}_{1:T},\accentset{\leftarrow}{X}_{1:T},Y_{1:T}) ]
		&= \mathbb{E} [ \mathbb{E} [ f(\accentset{\rightarrow}{X}_{1:T},\accentset{\leftarrow}{X}_{1:T},Y_{1:T}) \mid \accentset{\rightarrow}{X}_{1:T},\accentset{\leftarrow}{X}_{1:T} ] ] \\
		&= \mathbb{E} \left[ \int_{\textup{Y}^T} f(\accentset{\rightarrow}{X}_{1:T},\accentset{\leftarrow}{X}_{1:T},y_{1:T}) \prod_{t=1}^{T} \Phi_t ((\accentset{\rightarrow}{X}_t,\accentset{\leftarrow}{X}_t),\textup{d}y_t)  \right]
	\end{align*}
	by Condition (iv) in Definition \ref{def:cncssm}. Note that
	\begin{equation*}
		\mathbb{P} ( \accentset{\rightarrow}{X}_{1:T} \in \accentset{\rightarrow}{A}_{1:T},\accentset{\leftarrow}{X}_{1:T} \in \accentset{\leftarrow}{A}_{1:T} ) = \mathbb{P} ( \accentset{\rightarrow}{X}_{1:T} \in \accentset{\rightarrow}{A}_{1:T} ) \mathbb{P} ( \accentset{\leftarrow}{X}_{1:T} \in \accentset{\leftarrow}{A}_{1:T} )
	\end{equation*}
	for all $\accentset{\rightarrow}{A}_{1:T} \in \accentset{\rightarrow}{\mathcal{X}}^T$ and $\accentset{\leftarrow}{A}_{1:T} \in \accentset{\leftarrow}{\mathcal{X}}^T$ by Condition (iii) in Definition \ref{def:cncssm} where
	\begin{align*}
		\mathbb{P} ( \accentset{\rightarrow}{X}_{1:T} \in \accentset{\rightarrow}{A}_{1:T} ) 
		&= \mathbb{E} \left[ 1_{\accentset{\rightarrow}{A}_1} (\accentset{\rightarrow}{X}_1) \cdots 1_{\accentset{\rightarrow}{A}_T} (\accentset{\rightarrow}{X}_T) \right] \\
		&= \mathbb{E} \left[ \mathbb{E} \left[ 1_{\accentset{\rightarrow}{A}_1} (\accentset{\rightarrow}{X}_1) \cdots 1_{\accentset{\rightarrow}{A}_T} (\accentset{\rightarrow}{X}_T) \mid \accentset{\rightarrow}{X}_{1:T-1} \right] \right] \\
		&= \mathbb{E} \left[ \int_{\accentset{\rightarrow}{\textup{X}}} 1_{\accentset{\rightarrow}{A}_1} (\accentset{\rightarrow}{X}_1) \cdots 1_{\accentset{\rightarrow}{A}_{T-1}} (\accentset{\rightarrow}{X}_{T-1}) 1_{\accentset{\rightarrow}{A}_T} (\accentset{\rightarrow}{x}_T) \accentset{\rightarrow}{P}_T (\accentset{\rightarrow}{X}_{T-1},\textup{d}\accentset{\rightarrow}{x}_{T}) \right] \\
		& \ \, \vdots \\
		&= \mathbb{E} \left[ \int_{\accentset{\rightarrow}{\textup{X}}} \cdots \int_{\accentset{\rightarrow}{\textup{X}}} 1_{\accentset{\rightarrow}{A}_1} (\accentset{\rightarrow}{X}_1) 1_{\accentset{\rightarrow}{A}_2} (\accentset{\rightarrow}{x}_2) \cdots 1_{\accentset{\rightarrow}{A}_T} (\accentset{\rightarrow}{x}_T) \accentset{\rightarrow}{P}_T (\accentset{\rightarrow}{x}_{T-1},\textup{d}\accentset{\rightarrow}{x}_{T}) \cdots \accentset{\rightarrow}{P}_2 (\accentset{\rightarrow}{X}_{1},\textup{d}\accentset{\rightarrow}{x}_{2}) \right] \\
		&= \int_{\accentset{\rightarrow}{\textup{X}}} \cdots \int_{\accentset{\rightarrow}{\textup{X}}} 1_{\accentset{\rightarrow}{A}_1} (\accentset{\rightarrow}{x}_1) \cdots 1_{\accentset{\rightarrow}{A}_T} (\accentset{\rightarrow}{x}_T) \accentset{\rightarrow}{P}_T (\accentset{\rightarrow}{x}_{T-1},\textup{d}\accentset{\rightarrow}{x}_{T}) \cdots \accentset{\rightarrow}{P}_2 (\accentset{\rightarrow}{x}_{1},\textup{d}\accentset{\rightarrow}{x}_{2}) \accentset{\rightarrow}{\Pi}_1 (\textup{d}\accentset{\rightarrow}{x}_{1})
	\end{align*}
	and
	\begin{align*}
		\mathbb{P} ( \accentset{\leftarrow}{X}_{1:T} \in \accentset{\leftarrow}{A}_{1:T} ) 
		&= \mathbb{E} \left[ 1_{\accentset{\leftarrow}{A}_1} (\accentset{\leftarrow}{X}_1) \cdots 1_{\accentset{\leftarrow}{A}_T} (\accentset{\leftarrow}{X}_T) \right] \\
		&= \mathbb{E} \left[ \mathbb{E} \left[ 1_{\accentset{\leftarrow}{A}_1} (\accentset{\leftarrow}{X}_1) \cdots 1_{\accentset{\leftarrow}{A}_T} (\accentset{\leftarrow}{X}_T) \mid \accentset{\leftarrow}{X}_{2:T} \right] \right] \\
		&= \mathbb{E} \left[ \int_{\accentset{\leftarrow}{\textup{X}}} 1_{\accentset{\leftarrow}{A}_1} (\accentset{\leftarrow}{x}_1) 1_{\accentset{\leftarrow}{A}_2} (\accentset{\leftarrow}{X}_2) \cdots 1_{\accentset{\leftarrow}{A}_T} (\accentset{\leftarrow}{X}_T) \accentset{\leftarrow}{P}_1 (\accentset{\leftarrow}{X}_2,\textup{d}\accentset{\leftarrow}{x}_{1}) \right] \\
		& \ \, \vdots \\
		&= \mathbb{E} \left[ \int_{\accentset{\leftarrow}{\textup{X}}} \cdots \int_{\accentset{\leftarrow}{\textup{X}}} 1_{\accentset{\leftarrow}{A}_1} (\accentset{\leftarrow}{x}_1) \cdots 1_{\accentset{\leftarrow}{A}_{T-1}} (\accentset{\leftarrow}{x}_{T-1}) 1_{\accentset{\leftarrow}{A}_T} (\accentset{\leftarrow}{X}_T) \accentset{\leftarrow}{P}_1 (\accentset{\leftarrow}{x}_{2},\textup{d}\accentset{\leftarrow}{x}_{1}) \cdots \accentset{\leftarrow}{P}_{T-1} (\accentset{\leftarrow}{X}_{T},\textup{d}\accentset{\leftarrow}{x}_{T-1}) \right] \\
		&= \int_{\accentset{\leftarrow}{\textup{X}}} \cdots \int_{\accentset{\leftarrow}{\textup{X}}} 1_{\accentset{\leftarrow}{A}_1} (\accentset{\leftarrow}{x}_1) \cdots 1_{\accentset{\leftarrow}{A}_T} (\accentset{\leftarrow}{x}_T) \accentset{\leftarrow}{P}_1 (\accentset{\leftarrow}{x}_{2},\textup{d}\accentset{\leftarrow}{x}_{1}) \cdots \accentset{\leftarrow}{P}_{T-1} (\accentset{\leftarrow}{x}_{T},\textup{d}\accentset{\leftarrow}{x}_{T-1}) \accentset{\leftarrow}{\Pi}_T (\textup{d}\accentset{\leftarrow}{x}_{T})
	\end{align*}
	by Conditions (i) and (ii) in Definition \ref{def:cncssm}, respectively. Thus,
	\begin{align*}
		\mathbb{E} [ f(\accentset{\rightarrow}{X}_{1:T},\accentset{\leftarrow}{X}_{1:T},Y_{1:T}) ]
		&= \mathbb{E} \left[ \int_{\textup{Y}^T} f(\accentset{\rightarrow}{X}_{1:T},\accentset{\leftarrow}{X}_{1:T},y_{1:T}) \prod_{t=1}^{T} \Phi_t ((\accentset{\rightarrow}{X}_t,\accentset{\leftarrow}{X}_t),\textup{d}y_t) \right] \\
		&= \int_{\accentset{\rightarrow}{\textup{X}}^T} \int_{\accentset{\leftarrow}{\textup{X}}^T} \int_{\textup{Y}^T} f(\accentset{\rightarrow}{x}_{1:T},\accentset{\leftarrow}{x}_{1:T},y_{1:T}) \prod_{t=1}^{T} \Phi_t ((\accentset{\rightarrow}{x}_t,\accentset{\leftarrow}{x}_t),\textup{d}y_t) \\
		& \quad \cdot \prod_{t=T}^{2} \accentset{\rightarrow}{P}_t (\accentset{\rightarrow}{x}_{t-1},\textup{d}\accentset{\rightarrow}{x}_{t}) \accentset{\rightarrow}{\Pi}_1 (\textup{d}\accentset{\rightarrow}{x}_{1}) \prod_{t=1}^{T-1} \accentset{\leftarrow}{P}_t (\accentset{\leftarrow}{x}_{t+1},\textup{d}\accentset{\leftarrow}{x}_{t}) \accentset{\leftarrow}{\Pi}_T (\textup{d}\accentset{\leftarrow}{x}_{T}).
	\end{align*}
	This concludes the proof.	
\end{proof}

\begin{proof}[Proof of Proposition \ref{prop:cncssm}]
	Let $\accentset{\rightarrow}{A}_{1:T} \in \accentset{\rightarrow}{\mathcal{X}}^T$, $\accentset{\leftarrow}{A}_{1:T} \in \accentset{\leftarrow}{\mathcal{X}}^T$, and $B_{1:T} \in \mathcal{Y}^T$ be given and set
	\begin{equation*}
		f (z_{1:T}) = \int_{\accentset{\rightarrow}{\textup{X}}^T} \int_{\textup{Y}^T} \prod_{t=1}^{T} 1_{\accentset{\rightarrow}{A}_t} (\accentset{\rightarrow}{x}_t) \prod_{t=1}^{T} 1_{B_t} (y_t) \prod_{t=1}^{T} \Phi_t ((\accentset{\rightarrow}{x}_t,z_t),\textup{d}y_t) \prod_{t=T}^{2} \accentset{\rightarrow}{P}_t (\accentset{\rightarrow}{x}_{t-1},\textup{d}\accentset{\rightarrow}{x}_{t}) \accentset{\rightarrow}{\Pi}_1 (\textup{d}\accentset{\rightarrow}{x}_{1}).
	\end{equation*}
	Then,
	\begin{align*}
		&\mathbb{P} (\accentset{\rightarrow}{X}_{1:T} \in \accentset{\rightarrow}{A}_{1:T},\accentset{\leftarrow}{X}_{1:T} \in \accentset{\leftarrow}{A}_{1:T},Y_{1:T} \in B_{1:T}) \\
		&= \int_{\accentset{\leftarrow}{\textup{X}}^T} \prod_{t=1}^{T} 1_{\accentset{\leftarrow}{A}_t} (z_t) f (z_{1:T}) \accentset{\leftarrow}{P}_1 (z_{2},\textup{d}z_{1}) \cdots \accentset{\leftarrow}{P}_{T-1} (z_{T},\textup{d}z_{T-1}) \accentset{\leftarrow}{\Pi}_T (\textup{d}z_{T}) \\
		&= \int_{\accentset{\leftarrow}{\textup{X}}^T} \prod_{t=1}^{T} 1_{\accentset{\leftarrow}{A}_t} (z_t) f (z_{1:T}) \tilde{P}_{T} (z_{T-1},\textup{d}z_{T}) \accentset{\leftarrow}{P}_1 (z_{2},\textup{d}z_{1}) \cdots \accentset{\leftarrow}{P}_{T-2} (z_{T-1},\textup{d}z_{T-2}) \accentset{\leftarrow}{\Pi}_{T-1} (\textup{d}z_{T-1}) \\
		&= \int_{\accentset{\leftarrow}{\textup{X}}^T} \prod_{t=1}^{T} 1_{\accentset{\leftarrow}{A}_t} (z_t) f (z_{1:T}) \tilde{P}_{T} (z_{T-1},\textup{d}z_{T}) \tilde{P}_{T-1} (z_{T-2},\textup{d}z_{T-1}) \accentset{\leftarrow}{P}_1 (z_{2},\textup{d}z_{1}) \cdots \accentset{\leftarrow}{P}_{T-3} (z_{T-2},\textup{d}z_{T-3}) \accentset{\leftarrow}{\Pi}_{T-2} (\textup{d}z_{T-2}) \\
		& \ \, \vdots \\
		&= \int_{\accentset{\leftarrow}{\textup{X}}^T} \prod_{t=1}^{T} 1_{\accentset{\leftarrow}{A}_t} (z_t) f (z_{1:T}) \tilde{P}_{T} (z_{T-1},\textup{d}z_{T}) \cdots \tilde{P}_2 (z_{1},\textup{d}z_{2}) \accentset{\leftarrow}{\Pi}_1 (\textup{d}z_{1}) \\
		&= \mathbb{P} (\bar{X}_{1:T} \in (\accentset{\rightarrow}{A} \times \accentset{\leftarrow}{A})_{1:T},\bar{Y}_{1:T} \in B_{1:T}),
	\end{align*}
	where the first equality follows from Lemma \ref{lem:cncssm}, the intermediate ones follows from Equation \eqref{eq:RK1}, and the last equality follows from Lemma \ref{lem:cncssm} as well.
\end{proof}

\begin{proof}[Proof of Proposition \ref{prop:filtering}]
	First, by Lemma \ref{lem:cncssm} together with Fubini's theorem,
	\begin{align*}
		\mathbb{E} [ f(X_1,Y_1,...,X_t,Y_t) ] &= \int_{\textup{Y}} \cdots \int_{\textup{Y}} \int_{\textup{X}} \cdots \int_{\textup{X}} f(x_1,y_1,...,x_{t},y_{t}) \phi_1(x_1,y_1) \cdots \phi_t(x_t,y_t) \\
		& \quad \cdot P_1(x_2,\textup{d}x_1) \cdots P_{t-1}(x_t,\textup{d}x_{t-1}) \Pi_{t} (\textup{d}x_t) \mu_Y (\textup{d}y_{1}) \cdots \mu_Y (\textup{d}y_{t})
	\end{align*}
	for all $f \in \mathcal{M}_{b}((\textup{X} \times \textup{Y})^t)$. Therefore, by Lemma \ref{lem:app1},
	\begin{align*}
		&\mathbb{P} (X_1 \in A_1,...,X_t \in A_t \mid Y_1,...,Y_t) \\
		&= \frac{\int_{\textup{X}} \cdots \int_{\textup{X}} 1_{A_1} (x_1) \cdots 1_{A_t} (x_t) \phi_1(x_1,Y_1) \cdots \phi_t(x_t,Y_t) P_1(x_2,\textup{d}x_1) \cdots P_{t-1}(x_t,\textup{d}x_{t-1}) \Pi_{t} (\textup{d}x_t)}{\int_{\textup{X}} \cdots \int_{\textup{X}} \phi_1(x_1,Y_1) \cdots \phi_t(x_t,Y_t) P_1(x_2,\textup{d}x_1) \cdots P_{t-1}(x_t,\textup{d}x_{t-1}) \Pi_{t} (\textup{d}x_t)}
	\end{align*}
	for all $A_1,...,A_t \in \mathcal{X}$. The conclusion follows by choosing $A_1 = \cdots = A_{t-1} = \textup{X}$.
\end{proof}

\begin{proof}[Proof of Proposition \ref{prop:smoothing}]
	The proof follows the same lines as the proof of Proposition \ref{prop:filtering}. First, by Lemma \ref{lem:cncssm} together with Fubini's theorem,
	\begin{align*}
		\mathbb{E} [ f(X_1,Y_1,...,X_s,Y_s) ] &= \int_{\textup{Y}} \cdots \int_{\textup{Y}} \int_{\textup{X}} \cdots \int_{\textup{X}} f(x_1,y_1,...,x_{s},y_{s}) \phi_1(x_1,y_1) \cdots \phi_s(x_s,y_s) \\
		& \quad \cdot P_1(x_2,\textup{d}x_1) \cdots P_{s-1}(x_s,\textup{d}x_{s-1}) \Pi_{s} (\textup{d}x_s) \mu_Y (\textup{d}y_{1}) \cdots \mu_Y (\textup{d}y_{s})
	\end{align*}
	for all $f \in \mathcal{M}_{b}((\textup{X} \times \textup{Y})^s)$. Therefore, by Lemma \ref{lem:app1},
	\begin{align*}
		&\mathbb{P} (X_1 \in A_1,...,X_s \in A_s \mid Y_1,...,Y_s) \\
		&= \frac{\int_{\textup{X}} \cdots \int_{\textup{X}} 1_{A_1} (x_1) \cdots 1_{A_s} (x_s) \phi_1(x_1,Y_1) \cdots \phi_s(x_s,Y_s) P_1(x_2,\textup{d}x_1) \cdots P_{s-1}(x_s,\textup{d}x_{s-1}) \Pi_{s} (\textup{d}x_s)}{\int_{\textup{X}} \cdots \int_{\textup{X}} \phi_1(x_1,Y_1) \cdots \phi_s(x_s,Y_s) P_1(x_2,\textup{d}x_1) \cdots P_{s-1}(x_s,\textup{d}x_{s-1}) \Pi_{s} (\textup{d}x_s)}
	\end{align*}
	for all $A_1,...,A_s \in \mathcal{X}$. The conclusion follows by choosing $A_1 = \cdots = A_{t-1} = A_{t+1} = \cdots = A_{s} = \textup{X}$ and using Lemma \ref{lem:app2}.
\end{proof}

\noindent To prove Proposition \ref{prop:Proposition3}, we need the following two results. Recall that the characteristic function of $X \sim \mathcal{S} (\alpha,\mu,\sigma,\beta)$ is given by
\begin{equation*}
	\varphi_X (t) = \exp \left( \textup{i} \mu t - \sigma^{\alpha} \left| t \right|^{\alpha} \left( 1 - \textup{i} \beta \textup{sign} (t) \Phi \right) \right), \quad t \in \mathbb{R},
\end{equation*}
with $\Phi = \tan \left( \frac{\pi \alpha}{2} \right)$ if $\alpha \neq 1$ and $\Phi = - \frac{2}{\pi} \log \left( \left| t \right| \right)$ if $\alpha = 1$. In the following, $\int_{-\infty}^{\infty} f(x) \textup{d}x$ and $\int_{\mathbb{R}} f(x) \lambda(\textup{d}x)$ denote the Riemann and Lebesgue integral of $f$, respectively.

\begin{lemmaappendix} \label{lem:Proposition3A}
	If $X \sim \mathcal{S} (\alpha,\mu,\sigma,\beta)$, then $\int_{\mathbb{R}} |t|^n |\varphi_X (t)| \lambda(\textup{d}t) < \infty$ for all $n \in \mathbb{N}_0$.
\end{lemmaappendix}
\begin{proof}[Proof of Lemma \ref{lem:Proposition3A}]
	Let $n \in \mathbb{N}_0$ be given. Then,
	\begin{align*}
		\int_{\mathbb{R}} |t|^n |\varphi_X (t)| \lambda(\textup{d}t) &= \int_{\mathbb{R}} |t|^n \exp \left( - \sigma^{\alpha} \left| t \right|^{\alpha} \right) \lambda(\textup{d}t) \\
		&= \int_{\mathbb{R}} \lim_{N \rightarrow \infty} 1_{[-N,N]} (t) |t|^n \exp \left( - \sigma^{\alpha} \left| t \right|^{\alpha} \right) \lambda(\textup{d}t) \\
		&= \lim_{N \rightarrow \infty} \int_{\mathbb{R}} 1_{[-N,N]} (t) |t|^n \exp \left( - \sigma^{\alpha} \left| t \right|^{\alpha} \right) \lambda(\textup{d}t) \\
		&= \int_{-\infty}^{\infty} |t|^n \exp \left( - \sigma^{\alpha} \left| t \right|^{\alpha} \right) \textup{d}t \\
		&= 2 \int_{0}^{\infty} t^n \exp \left( - \sigma^{\alpha} t^{\alpha} \right) \textup{d}t =: A,
		\intertext{where the third equality follows from the monotone convergence theorem. Let $u = \sigma^{\alpha} t^{\alpha}$. Then, by integration by substitution,}
		A &= 2 \frac{1}{\alpha} \frac{1}{\sigma^{n+1}} \int_{0}^{\infty} u^{\frac{n+1}{\alpha}-1} \exp \left( - u \right) \textup{d}u \\
		&= 2 \frac{1}{\alpha} \frac{1}{\sigma^{n+1}} \Gamma \left( \frac{n+1}{\alpha} \right) < \infty,
	\end{align*}
	where $\Gamma(z) = \int_{0}^{\infty} t^{z-1} \exp(-t) \textup{d}t, z \in \mathbb{C} \ \textup{such that} \ \textup{Re} (z) > 0$ is the gamma function. This concludes the proof.
\end{proof}

\begin{lemmaappendix} \label{lem:Proposition3B}
	If $X \sim \mathcal{S} (\alpha,0,\sigma,\beta)$ where $\alpha \neq 1$ or $\beta = 0$, then $\int_{\mathbb{R}} |\varphi_X^{\prime} (t)| \lambda(\textup{d}t) < \infty$.
\end{lemmaappendix}
\begin{proof}[Proof of Lemma \ref{lem:Proposition3B}]
	Note that the derivative of the characteristic function of $X \sim \mathcal{S} (\alpha,0,\sigma,\beta)$ where $\alpha \neq 1$ or $\beta = 0$ is given by
	 \begin{equation*}
	 	\varphi_X^{\prime} (t) = - \alpha \sigma^{\alpha} \left( 1 - \textup{i} \beta \textup{sign} (t) \tan \left( \frac{\pi \alpha}{2} \right) \right) \left| t \right|^{\alpha-1} \textup{sign} (t) \varphi_X (t), \quad t \in \mathbb{R} \backslash \{0\},
	 \end{equation*}
 	such that
 	\begin{equation*}
 		\left| \varphi_X^{\prime} (t) \right| \leq \alpha \sigma^{\alpha} \left( 1 + \beta \tan \left( \frac{\pi \alpha}{2} \right)  \right) \left| t \right|^{\alpha-1} \left| \varphi_X (t) \right|, \quad t \in \mathbb{R} \backslash \{0\}.
 	\end{equation*}
 	Thus,
 	\begin{align*}
 		\int_{\mathbb{R}} |\varphi_X^{\prime} (t)| \lambda(\textup{d}t) &\leq \alpha \sigma^{\alpha} \left( 1 + \beta \tan \left( \frac{\pi \alpha}{2} \right)  \right) \int_{\mathbb{R}} \left| t \right|^{\alpha-1} \exp \left( - \sigma^{\alpha} \left| t \right|^{\alpha} \right) \lambda(\textup{d}t) \\
 		&= \alpha \sigma^{\alpha} \left( 1 + \beta \tan \left( \frac{\pi \alpha}{2} \right)  \right) \int_{\mathbb{R}} \lim_{N \rightarrow \infty} 1_{[-N,N]} (t) \left| t \right|^{\alpha-1} \exp \left( - \sigma^{\alpha} \left| t \right|^{\alpha} \right) \lambda(\textup{d}t) \\
 		&= \alpha \sigma^{\alpha} \left( 1 + \beta \tan \left( \frac{\pi \alpha}{2} \right)  \right) \lim_{N \rightarrow \infty} \int_{\mathbb{R}} 1_{[-N,N]} (t) \left| t \right|^{\alpha-1} \exp \left( - \sigma^{\alpha} \left| t \right|^{\alpha} \right) \lambda(\textup{d}t) \\
 		&= \alpha \sigma^{\alpha} \left( 1 + \beta \tan \left( \frac{\pi \alpha}{2} \right)  \right) \int_{-\infty}^{\infty} \left| t \right|^{\alpha-1} \exp \left( - \sigma^{\alpha} \left| t \right|^{\alpha} \right) \textup{d}t \\
 		&= 2 \alpha \sigma^{\alpha} \left( 1 + \beta \tan \left( \frac{\pi \alpha}{2} \right) \right) \int_{0}^{\infty} t^{\alpha-1} \exp \left( - \sigma^{\alpha} t^{\alpha} \right) \textup{d}t =: A,
 		\intertext{where the third equality follows from the monotone convergence theorem. As in the proof of Lemma \ref{lem:Proposition3A}, let $u = \sigma^{\alpha} t^{\alpha}$. Then, by integration by substitution,}
 		A &= 2 \alpha \sigma^{\alpha} \left( 1 + \beta \tan \left( \frac{\pi \alpha}{2} \right) \right) \frac{1}{\alpha} \frac{1}{\sigma^{\alpha}} \int_{0}^{\infty} \exp \left( - u \right) \textup{d}u \\
 		& = 2 \left( 1 + \beta \tan \left( \frac{\pi \alpha}{2} \right) \right) < \infty.
 	\end{align*}
 	This concludes the proof.
\end{proof}

\begin{proof}[Proof of Proposition \ref{prop:Proposition3}]
	The proof is similar to the proof of Theorem 3 in \cite{CambanisFakhre-Zakeri1995}. We have that
	\begin{equation}
		\begin{aligned}
			\mathbb{E} \left[ X_{t}^{B} \mid X_{t-1}^{B} = c \right] &= \int_{\mathbb{R}} x f_{X_{t}^{B} \mid X_{t-1}^{B}} (x \mid c) \lambda(\textup{d}x) \\
			&= \frac{1}{f_{X^{B}} (c)} \int_{\mathbb{R}} x f_{X_{t}^{B} , X_{t-1}^{B}} (x , c) \lambda(\textup{d}x) \\
			&= \frac{1}{f_{X^{B}} (c)} \int_{\mathbb{R}} x f_{X_{t}^{B} , \varepsilon_{t-1}^{B}} (x , c-\rho_Bx) \lambda(\textup{d}x) \\
			&= \frac{1}{f_{X^{B}} (c)} \int_{\mathbb{R}} x f_{X^{B}} (x) f_{\varepsilon^{B}} (c-\rho_Bx) \lambda(\textup{d}x).
		\end{aligned} \label{eq:P3A}		 
	\end{equation}
	Note that, by Fubini's theorem,
	\begin{align*}
		\int_{\mathbb{R}} \exp(\textup{i}tx) f_{X^{B}}(x) f_{\varepsilon^{B}}(c-\rho_Bx) \lambda(\textup{d}x) &= \int_{\mathbb{R}} \exp(\textup{i}tx) f_{X^{B}}(x) \frac{1}{2\pi} \int_{\mathbb{R}} \exp(-\textup{i}u(c-\rho_{B}x)) \varphi_{\varepsilon^{B}} (u) \lambda(\textup{d}u) \lambda(\textup{d}x) \\
		&= \frac{1}{2\pi} \int_{\mathbb{R}} \exp(-\textup{i}uc) \int_{\mathbb{R}} \exp(\textup{i}(t+u\rho_{B})x) f_{X^{B}}(x) \lambda(\textup{d}x) \varphi_{\varepsilon^{B}} (u) \lambda(\textup{d}u) \\
		&= \frac{1}{2\pi} \int_{\mathbb{R}} \exp(-\textup{i}uc) \varphi_{X^{B}} (t+u\rho_{B}) \varphi_{\varepsilon^{B}} (u) \lambda(\textup{d}u)
	\end{align*}
	for all $t \in \mathbb{R}$ since, by Tonelli's theorem and Lemma \ref{lem:Proposition3A}, $\frac{1}{2\pi} \int_{\mathbb{R}^2} | \exp(\textup{i}tx) f_{X^{B}}(x) \exp(-\textup{i}u(c-\rho_{B}x)) \varphi_{\varepsilon^{B}} (u)| \lambda(\textup{d}x,\textup{d}u) = \frac{1}{2\pi} \int_{\mathbb{R}} |\varphi_{\varepsilon^{B}} (u)| \lambda(\textup{d}u)< \infty$, and thus that
	\begin{equation}
		\lim_{t \rightarrow 0} \int_{\mathbb{R}} \frac{\exp(\textup{i}tx)-1}{t} f_{X^{B}}(x) f_{\varepsilon^{B}}(c-\rho_Bx) \lambda(\textup{d}x) = \frac{1}{2\pi} \lim_{t \rightarrow 0} \int_{\mathbb{R}} \exp(-\textup{i}uc) \frac{\varphi_{X^{B}} (t+u\rho_{B})-\varphi_{X^{B}} (u\rho_{B})}{t} \varphi_{\varepsilon^{B}} (u) \lambda(\textup{d}u). \label{eq:P3B}
	\end{equation}
	The left-hand side (LHS) of Equation \eqref{eq:P3B} is equal to
	\begin{equation*}
		\textup{LHS} = \lim_{n \rightarrow \infty} \int_{\mathbb{R}} \frac{\exp\left(\textup{i}\frac{1}{n}x\right)-1}{\frac{1}{n}} f_{X^{B}}(x) f_{\varepsilon^{B}}(c-\rho_Bx) \lambda(\textup{d}x).
	\end{equation*}
	Note that 
	\begin{equation*}
		\lim_{n \rightarrow \infty} \frac{\exp\left(\textup{i}\frac{1}{n}x\right)-1}{\frac{1}{n}} f_{X^{B}}(x) f_{\varepsilon^{B}}(c-\rho_Bx) = \textup{i} x f_{X^{B}}(x) f_{\varepsilon^{B}}(c-\rho_Bx)
	\end{equation*}
	for all $x \in \mathbb{R}$ since $\lim_{t \rightarrow 0} \frac{\exp(\textup{i}tx)-1}{t} = \lim_{t \rightarrow 0} \frac{\cos(tx)-1}{t} + \textup{i} \lim_{t \rightarrow 0} \frac{\sin(tx)}{t} = - \lim_{t \rightarrow 0} \sin(tx) x + \textup{i} \lim_{t \rightarrow 0} \cos(tx) x = \textup{i} x$ where the second equality follows from L'Hôpital's rule. Moreover, note that there exists a $C \in \mathbb{R}_{+}$ such that 
	\begin{equation*}
		\left| \frac{\exp\left(\textup{i}\frac{1}{n}x\right)-1}{\frac{1}{n}} f_{X^{B}}(x) f_{\varepsilon^{B}}(c-\rho_Bx) \right| \leq (C+|x|) f_{X^{B}}(x) f_{\varepsilon^{B}}(c-\rho_Bx)
	\end{equation*}
	for all $x \in \mathbb{R}$ and $n \in \mathbb{N}$ since $\left| \frac{\exp\left(\textup{i}\frac{1}{n}x\right)-1}{\frac{1}{n}} \right| \leq 2n$ and $\lim_{n \rightarrow \infty} \frac{\exp\left(\textup{i}\frac{1}{n}x\right)-1}{\frac{1}{n}} = \textup{i}x$, and that
	\begin{equation*}
		\int_{\mathbb{R}} (C+|x|) f_{X^{B}}(x) f_{\varepsilon^{B}}(c-\rho_Bx) \lambda(\textup{d}x) = C f_{X^{B}}(c) + f_{X^{B}}(c) \mathbb{E} \left[ \left| X_{t}^{B} \right| \mid X_{t-1}^{B} = c \right] < \infty.
	\end{equation*}
	Thus, by the dominated convergence theorem,
	\begin{equation*}
		\textup{LHS} = \textup{i} \int_{\mathbb{R}} x f_{X^{B}}(x) f_{\varepsilon^{B}}(c-\rho_Bx) \lambda(\textup{d}x).
	\end{equation*}
	Moreover, the right-hand side (RHS) of Equation \eqref{eq:P3B} is equal to
	\begin{equation*}
		\textup{RHS} = \frac{1}{2\pi} \lim_{n \rightarrow \infty} \int_{\mathbb{R}} \exp(-\textup{i}uc) \frac{\varphi_{X^{B}} \left(\frac{1}{n}+u\rho_{B}\right)-\varphi_{X^{B}} (u\rho_{B})}{\frac{1}{n}} \varphi_{\varepsilon^{B}} (u) \lambda(\textup{d}u).
	\end{equation*}
	Note here that
	\begin{equation*}
		\lim_{n \rightarrow \infty} \exp(-\textup{i}uc) \frac{\varphi_{X^{B}} \left(\frac{1}{n}+u\rho_{B}\right)-\varphi_{X^{B}} (u\rho_{B})}{\frac{1}{n}} \varphi_{\varepsilon^{B}} (u) = \exp(-\textup{i}uc) \varphi_{X^{B}}^{\prime} (u\rho_{B}) \varphi_{\varepsilon^{B}} (u)
	\end{equation*}
	for all $u \in \mathbb{R} \backslash \{0\}$. Moreover, note that there exists a $C \in \mathbb{R}_{+}$ such that
	\begin{equation*}
		\left| \exp(-\textup{i}uc) \frac{\varphi_{X^{B}} \left(\frac{1}{n}+u\rho_{B}\right)-\varphi_{X^{B}} (u\rho_{B})}{\frac{1}{n}} \varphi_{\varepsilon^{B}} (u) \right| \leq C \left| \varphi_{\varepsilon^{B}} (u) \right| + \left| \varphi_{X^{B}}^{\prime} (u\rho_{B}) \right|
	\end{equation*}
	for all $u \in \mathbb{R} \backslash \{0\}$ and $n \in \mathbb{N}$ since $\left| \frac{\varphi_{X^{B}} \left(\frac{1}{n}+u\rho_{B}\right)-\varphi_{X^{B}} (u\rho_{B})}{\frac{1}{n}} \right| \leq 2n$ and $\lim_{n \rightarrow \infty} \frac{\varphi_{X^{B}} \left(\frac{1}{n}+u\rho_{B}\right)-\varphi_{X^{B}} (u\rho_{B})}{\frac{1}{n}} = \varphi_{X^{B}}^{\prime} (u\rho_{B})$, and that
	\begin{equation*}
		\int_{\mathbb{R}} \left( C \left| \varphi_{\varepsilon^{B}} (u) \right| + \left| \varphi_{X^{B}}^{\prime} (u\rho_{B}) \right| \right) \lambda(\textup{d}u) = C \int_{\mathbb{R}} \left| \varphi_{\varepsilon^{B}} (u) \right| \lambda (\textup{d}u) + \int_{\mathbb{R}} \left| \varphi_{X^{B}}^{\prime} (u\rho_{B}) \right| \lambda(\textup{d}u) < \infty
	\end{equation*}
	since $\int_{\mathbb{R}} \left| \varphi_{\varepsilon^{B}} (u) \right| \lambda (\textup{d}u) < \infty$ and $\int_{\mathbb{R}} \left| \varphi_{X^{B}}^{\prime} (u\rho_{B}) \right| \lambda(\textup{d}u) < \infty$ by Lemmata \ref{lem:Proposition3A} and \ref{lem:Proposition3B}, respectively. Thus, by the dominated convergence theorem,
	\begin{equation}
		\textup{RHS} = \frac{1}{2\pi} \int_{\mathbb{R}} \exp(-\textup{i}uc) \varphi_{X^{B}}^{\prime} (u\rho_{B}) \varphi_{\varepsilon^{B}} (u) \lambda(\textup{d}u). \label{eq:P3C}
	\end{equation}
	Note that if $\rho_{B} \neq 0$ and $\beta = 0$ or $\rho_B \in (0,1)$ and $\alpha \neq 1$, then
	\begin{equation*}
		\varphi_{X^{B}} (u) = \exp \left( - \sigma_{X^B}^{\alpha} \left| u \right|^{\alpha} \left( 1 - \textup{i} \beta_{X^B} \textup{sign} (u) \tan \left( \frac{\pi \alpha}{2} \right) \right) \right)
	\end{equation*}
	for all $u \in \mathbb{R}$ and
	\begin{equation*}
		\varphi_{X^{B}}^{\prime} (u) = \varphi_{X^{B}} (u) g(u) 
	\end{equation*}
	for all $u \in \mathbb{R} \backslash \{0\}$ where
	\begin{equation*}
		g(u) = - \alpha \sigma_{X^B}^{\alpha} \left| u \right|^{\alpha-1} \textup{sign} (u) + \alpha \sigma_{X^B}^{\alpha} \textup{i} \beta_{X^B} \left| u \right|^{\alpha-1} \tan \left( \frac{\pi \alpha}{2} \right).
	\end{equation*}
	It thus follows that
	\begin{align*}
		\varphi_{X^{B}}^{\prime} (\rho_{B}u) \varphi_{\varepsilon^{B}} (u) &= \varphi_{X^{B}} (\rho_{B}u) g(\rho_{B}u) \varphi_{\varepsilon^{B}} (u) \\
		&= \varphi_{X^{B}} (u) g(\rho_{B}u)
	\end{align*}
	for all $u \in \mathbb{R} \backslash \{0\}$ where
	\begin{align*}
		g(\rho_{B}u) &= - \alpha \sigma_{X^B}^{\alpha} \left| \rho_{B}u \right|^{\alpha-1} \textup{sign} (\rho_{B}u) + \alpha \sigma_{X^B}^{\alpha} \textup{i} \beta_{X^B} \left| \rho_{B}u \right|^{\alpha-1} \tan \left( \frac{\pi \alpha}{2} \right) \\
		&= \textup{sign} (\rho_{B})  \left| \rho_{B} \right|^{\alpha-1} \left( - \alpha \sigma_{X^B}^{\alpha} \left| u \right|^{\alpha-1} \textup{sign} (u) + \alpha \sigma_{X^B}^{\alpha} \textup{i} \beta_{X^B} \left| u \right|^{\alpha-1} \tan \left( \frac{\pi \alpha}{2} \right) \right) \\
		&= \textup{sign} (\rho_{B})  \left| \rho_{B} \right|^{\alpha-1} g(u)
	\end{align*}
	since $\rho_{B} \neq 0$ and $\beta = 0$ or $\rho_B \in (0,1)$ and $\alpha \neq 1$, so
	\begin{align*}
		\varphi_{X^{B}}^{\prime} (\rho_{B}u) \varphi_{\varepsilon^{B}} (u) &= \textup{sign} (\rho_{B}) \left| \rho_{B} \right|^{\alpha-1} \varphi_{X^{B}} (u) g(u) \\
		&= \textup{sign} (\rho_{B}) \left| \rho_{B} \right|^{\alpha-1} \varphi_{X^{B}}^{\prime} (u)
	\end{align*}
	for all $u \in \mathbb{R} \backslash \{0\}$. Thus, continuing from Equation \eqref{eq:P3C}, 
	\begin{equation*}
		\textup{RHS} = \textup{sign} (\rho_{B}) \left| \rho_{B} \right|^{\alpha-1} \frac{1}{2\pi} \int_{\mathbb{R}} \exp(-\textup{i}uc) \varphi_{X^{B}}^{\prime} (u) \lambda(\textup{d}u).
	\end{equation*}
	Finally,
	\begin{align*}
		&\int_{\mathbb{R}} \exp(-\textup{i}uc) \varphi_{X^{B}}^{\prime} (u) \lambda(\textup{d}u) \\
		&= \int_{\mathbb{R}} \lim_{N \rightarrow \infty} 1_{[-N,-1/N]} (u) \exp(-\textup{i}uc) \varphi_{X^{B}}^{\prime} (u) \lambda(\textup{d}u) + \int_{\mathbb{R}} \lim_{N \rightarrow \infty} 1_{[1/N,N]} (u) \exp(-\textup{i}uc) \varphi_{X^{B}}^{\prime} (u) \lambda(\textup{d}u) \\
		&= \lim_{N \rightarrow \infty} \int_{\mathbb{R}} 1_{[-N,-1/N]} (u) \exp(-\textup{i}uc) \varphi_{X^{B}}^{\prime} (u) \lambda(\textup{d}u) + \lim_{N \rightarrow \infty} \int_{\mathbb{R}} 1_{[1/N,N]} (u) \exp(-\textup{i}uc) \varphi_{X^{B}}^{\prime} (u) \lambda(\textup{d}u) \\
		&= \lim_{N \rightarrow \infty} \int_{-N}^{-1/N} \exp(-\textup{i}uc) \varphi_{X^{B}}^{\prime} (u) \textup{d}u + \lim_{N \rightarrow \infty} \int_{1/N}^{N} \exp(-\textup{i}uc) \varphi_{X^{B}}^{\prime} (u) \textup{d}u =: A,
	\end{align*}
	where the second equality follows from the dominated convergence theorem since $\int_{\mathbb{R}} | \varphi_{X^{B}}^{\prime} (u) | \lambda(\textup{d}u) < \infty$ by Lemma \ref{lem:Proposition3B}. Then, by integration by parts,
	\begin{align*}
		A &= \lim_{N \rightarrow \infty} \left[ \exp(-\textup{i}uc) \varphi_{X^{B}} (u) \right]_{-N}^{-1/N} + \textup{i} c \lim_{N \rightarrow \infty} \int_{-N}^{-1/N} \exp(-\textup{i}uc) \varphi_{X^{B}} (u) \textup{d}u \\
		& \quad + \lim_{N \rightarrow \infty} \left[ \exp(-\textup{i}uc) \varphi_{X^{B}} (u) \right]_{1/N}^{N} + \textup{i} c \lim_{N \rightarrow \infty} \int_{1/N}^{N} \exp(-\textup{i}uc) \varphi_{X^{B}} (u) \textup{d}u \\
		&= \lim_{N \rightarrow \infty} \left[ \exp(-\textup{i}uc) \varphi_{X^{B}} (u) \right]_{-N}^{N} + \textup{i} c \lim_{N \rightarrow \infty} \int_{\mathbb{R}} 1_{[-N,-1/N]} (u) \exp(-\textup{i}uc) \varphi_{X^{B}} (u) \lambda(\textup{d}u) \\
		& \quad + \textup{i} c \lim_{N \rightarrow \infty} \int_{\mathbb{R}} 1_{[1/N,N]} (u) \exp(-\textup{i}uc) \varphi_{X^{B}} (u) \lambda(\textup{d}u) \\
		&= \textup{i} c \int_{\mathbb{R}} \lim_{N \rightarrow \infty} 1_{[-N,-1/N]} (u) \exp(-\textup{i}uc) \varphi_{X^{B}} (u) \lambda(\textup{d}u) + \textup{i} c \int_{\mathbb{R}} \lim_{N \rightarrow \infty} 1_{[1/N,N]} (u) \exp(-\textup{i}uc) \varphi_{X^{B}} (u) \lambda(\textup{d}u) \\
		&= \textup{i} c \int_{\mathbb{R}} \exp(-\textup{i}uc) \varphi_{X^{B}} (u) \lambda(\textup{d}u),
	\end{align*}
	where the third equality follows from the dominated convergence theorem since $\int_{\mathbb{R}} | \varphi_{X^{B}} (u) | \lambda(\textup{d}u) < \infty$ by Lemma \ref{lem:Proposition3A} and the fact that $\lim_{N \rightarrow \infty} \varphi_{X^{B}} (N) = 0 = \lim_{N \rightarrow \infty} \varphi_{X^{B}} (-N)$ since $\int_{\mathbb{R}} | \varphi_{X^{B}}^{\prime} (u) | \lambda(\textup{d}u) < \infty$ by Lemma \ref{lem:Proposition3B}. The right-hand side of Equation \eqref{eq:P3B} is thus equal to
	\begin{equation*}
		\textup{RHS} = \textup{sign} (\rho_{B}) \left| \rho_{B} \right|^{\alpha-1} \textup{i} c \frac{1}{2\pi} \int_{\mathbb{R}} \exp(-\textup{i}uc) \varphi_{X^{B}} (u) \lambda(\textup{d}u).
	\end{equation*}
	Equation \eqref{eq:P3B} thus becomes
	\begin{equation}
		\int_{\mathbb{R}} x f_{X^{B}}(x) f_{\varepsilon^{B}}(c-\rho_Bx) \lambda(\textup{d}x) = \textup{sign} (\rho_{B}) \left| \rho_{B} \right|^{\alpha-1} c f_{X^{B}}(c). \label{eq:P3D}
	\end{equation}
	Hence, continuing from Equation \eqref{eq:P3A} and using Equation \eqref{eq:P3D}, we have that
	\begin{equation*}
		\mathbb{E} \left[ X_{t}^{B} \mid X_{t-1}^{B} = c \right] = \textup{sign} (\rho_{B}) \left| \rho_{B} \right|^{\alpha-1} c,
	\end{equation*}
	which concludes the proof.
\end{proof}

\begin{proof}[Proof of Proposition \ref{prop:BLP}]
	First, recall that a linear predictor $L(X_t^B \mid X_{t-1}^B) = a X_{t-1}^B$ is called a best linear predictor if it minimises $\sigma(X_t^B-L(X_t^B \mid X_{t-1}^B))$. Note that if $X \sim \mathcal{S} (\alpha,\mu_X,\sigma(X),\beta_X)$ and $Y \sim \mathcal{S} (\alpha,\mu_Y,\sigma(Y),\beta_Y)$ are independent, then 
	\begin{equation*}
		X + Y \sim \mathcal{S} \left(\alpha,\mu_X+\mu_Y,(\sigma^{\alpha}(X)+\sigma^{\alpha}(Y))^{1/\alpha},\frac{\beta_X \sigma^{\alpha}(X) + \beta_Y \sigma^{\alpha}(Y)}{\sigma^{\alpha}(X)+\sigma^{\alpha}(Y)}\right),
	\end{equation*}
	and that if $X \sim \mathcal{S} (\alpha,\mu_X,\sigma(X),\beta_X)$, then 
	\begin{equation*}
	\begin{aligned}
		cX &\sim \mathcal{S} (\alpha,c \mu_X,|c|\sigma(X),\textup{sign}(c)\beta_X) && \textup{if} \quad \alpha \neq 1 \\
		cX &\sim \mathcal{S} \left(\alpha,c \mu_X - \frac{2}{\pi}c\log(|c|)\sigma(X)\beta_X,|c|\sigma(X),\textup{sign}(c)\beta_X\right) && \textup{if} \quad \alpha = 1
	\end{aligned}
	\end{equation*}
	for all $c \in \mathbb{R} \backslash \{0\}$, see, for instance, \cite{SamorodnitskyTaqqu1994}. Thus,
	\begin{align*}
		\sigma^{\alpha} (X_t^B - L(X_t^B \mid X_{t-1}^B)) 
		&= \sigma^{\alpha} (X_t^B - a X_{t-1}^B) \\
		&= \sigma^{\alpha} (X_t^B - a \rho_B X_{t}^B - a \varepsilon_{t-1}^{B}) \\
		&= \sigma^{\alpha} ((1 - a \rho_B) X_{t}^B) + \sigma^{\alpha}(- a \varepsilon_{t-1}^{B}) \\
		&= |1 - a \rho_B|^{\alpha} \sigma^{\alpha} (X_{t}^B) + |a|^{\alpha} \sigma^{\alpha} (\varepsilon_{t-1}^{B}) \\
		&= \left( \frac{|\rho_B|^{\alpha}}{1-|\rho_B|^{\alpha}} \left|a - \frac{1}{\rho_B}\right|^{\alpha} + |a|^{\alpha} \right) \sigma_B^{\alpha}.
	\end{align*}
	A linear predictor $L(X_t^B \mid X_{t-1}^B) = a X_{t-1}^B$ is thus called a best linear predictor if $a$ minimises
	\begin{equation*}
		f(a) = \frac{|\rho_B|^{\alpha}}{1-|\rho_B|^{\alpha}} \left|a - \frac{1}{\rho_B}\right|^{\alpha} + |a|^{\alpha}.
	\end{equation*}
	Assume that $\rho_B \in (0,1)$. Then, $f(a)$ is strictly decreasing on $(-\infty,0]$ and strictly increasing on $[1/\rho_B,\infty)$, so $f(a)$ is minimised on $[0,1/\rho_B]$. Let thus $f^{+} : [0,1/\rho_B] \rightarrow \mathbb{R}$ be given by
	\begin{equation*}
		f^{+} (a) = \frac{|\rho_B|^{\alpha}}{1-|\rho_B|^{\alpha}} \left(\frac{1}{\rho_B}-a\right)^{\alpha} + a^{\alpha}.
	\end{equation*}
	The first-order derivative of $f^{+}$ is given by
	\begin{equation*}
		\frac{\textup{d}f^{+} (a)}{\textup{d}a} = -\alpha \frac{|\rho_B|^{\alpha}}{1-|\rho_B|^{\alpha}} \left(\frac{1}{\rho_B}-a\right)^{\alpha-1} + \alpha a^{\alpha-1},
	\end{equation*}
	and the second-order derivative of $f^{+}$ is given by
	\begin{equation*}
		\frac{\textup{d}^2f^{+} (a)}{\textup{d}a^2} = \alpha (\alpha-1) \frac{|\rho_B|^{\alpha}}{1-|\rho_B|^{\alpha}} \left(\frac{1}{\rho_B}-a\right)^{\alpha-2} + \alpha (\alpha-1) a^{\alpha-2}.
	\end{equation*}
	Thus, if $1 < \alpha \leq 2$, then $f^{+}(a)$ is strictly convex on $(0,1/\rho_B)$, so $f(a)$ is minimised at the value that solves $\frac{\textup{d}f^{+} (a)}{\textup{d}a} = 0$, which is
	\begin{equation*}
		a^{*} = \frac{\textup{sign}(\rho_B)|\rho_B|^{1/(\alpha-1)}}{|\rho_B|^{\alpha/(\alpha-1)}+(1-|\rho_B|^{\alpha})^{1/(\alpha-1)}}.
	\end{equation*}
	If $\alpha = 1$, then
	\begin{equation*}
		f^{+} (a) = \frac{|\rho_B|}{1-|\rho_B|} \frac{1}{\rho_B} + \left( 1 - \frac{|\rho_B|}{1-|\rho_B|} \right) a,
	\end{equation*}
	so $f(a)$ is minimised at $a^{*} = 0$ if $0 < |\rho_B| < 1/2$, $a^{*} \in [0,2]$ if $\rho_B = 1/2$, and $a^{*} = 1/\rho_B$ if $1/2 < |\rho_B| < 1$. If $0 < \alpha < 1$, then $f^{+}(a)$ is strictly concave on $(0,1/\rho_B)$, so $f(a)$ is minimised at $0$, $1/\rho_B$, or both, which is $a^{*} = 0$ if $0 < |\rho_B|^{\alpha} < 1/2$, $a^{*} \in \{0,1/\rho_B\}$ if $|\rho_B|^{\alpha} = 1/2$, and $a^{*} = 1/\rho_B$ if $1/2 < |\rho_B|^{\alpha} < 1$. Assume now that $\rho_B \in (-1,0)$. Then, $f(a)$ is strictly decreasing on $(-\infty,1/\rho_B]$ and strictly increasing on $[0,\infty)$, so $f(a)$ is minimised on $[1/\rho_B,0]$. Let thus $f^{-} : [1/\rho_B,0] \rightarrow \mathbb{R}$ be given by
	\begin{equation*}
		f^{-} (a) = \frac{|\rho_B|^{\alpha}}{1-|\rho_B|^{\alpha}} \left(a-\frac{1}{\rho_B}\right)^{\alpha} + (-a)^{\alpha}.
	\end{equation*}
	The first-order derivative of $f^{-}$ is given by
	\begin{equation*}
		\frac{\textup{d}f^{-} (a)}{\textup{d}a} = \alpha \frac{|\rho_B|^{\alpha}}{1-|\rho_B|^{\alpha}} \left(a-\frac{1}{\rho_B}\right)^{\alpha-1} - \alpha (-a)^{\alpha-1},
	\end{equation*}
	and the second-order derivative of $f^{-}$ is given by
	\begin{equation*}
		\frac{\textup{d}^2f^{-} (a)}{\textup{d}a^2} = \alpha (\alpha-1) \frac{|\rho_B|^{\alpha}}{1-|\rho_B|^{\alpha}} \left(a-\frac{1}{\rho_B}\right)^{\alpha-2} + \alpha (\alpha-1) (-a)^{\alpha-2}.
	\end{equation*}
	Thus, if $1 < \alpha \leq 2$, then $f^{-}(a)$ is strictly convex on $(1/\rho_B,0)$, so $f(a)$ is minimised at the value that solves $\frac{\textup{d}f^{-} (a)}{\textup{d}a} = 0$, which is
	\begin{equation*}
		a^{*} = \frac{\textup{sign}(\rho_B)|\rho_B|^{1/(\alpha-1)}}{|\rho_B|^{\alpha/(\alpha-1)}+(1-|\rho_B|^{\alpha})^{1/(\alpha-1)}}.
	\end{equation*}
	If $\alpha = 1$, then
	\begin{equation*}
		f^{-} (a) = - \frac{|\rho_B|}{1-|\rho_B|} \frac{1}{\rho_B} + \left( \frac{|\rho_B|}{1-|\rho_B|} - 1\right) a,
	\end{equation*}
	so $f(a)$ is minimised at $a^{*} = 1/\rho_B$ if $1/2 < |\rho_B| < 1$, $a^{*} \in [-2,0]$ if $\rho_B = -1/2$, and $a^{*} = 0$ if $0 < |\rho_B| < 1/2$. If $0 < \alpha < 1$, then $f^{-}(a)$ is strictly concave on $(1/\rho_B,0)$, so $f(a)$ is minimised at $1/\rho_B$, $0$, or both, which is $a^{*} = 1/\rho_B$ if $1/2 < |\rho_B|^{\alpha} < 1$, $a^{*} \in \{1/\rho_B,0\}$ if $|\rho_B|^{\alpha} = 1/2$, and $a^{*} = 0$ if $0 < |\rho_B|^{\alpha} < 1/2$.	
\end{proof}

\begin{proof}[Proof of Proposition \ref{prop:cf}]
	The characteristic function of $Y_{t-p:t}$ is given by 
	\begin{align*}
		\varphi_Y^p (u) &= \mathbb{E} \left[ \exp \left( \textup{i} \sum_{i=0}^{p} u_i Y_{t-i} \right) \right] \\
		&= \mathbb{E} \left[ \exp \left( \textup{i} \sum_{i=0}^{p} u_i X_{t-i}^F \right) \exp \left( \textup{i} \sum_{i=0}^{p} u_i X_{t-i}^B \right) \right].
	\end{align*}
	Note that
	\begin{equation*}
		X_{t-i}^{F} = \frac{\mu_F}{1-\rho_F} + \sum_{j=0}^{\infty} \rho_F^j \varepsilon_{t-i-j}^F
	\end{equation*}
	 and
	\begin{equation*}
		X_{t-i}^{B} = \sum_{j=0}^{\infty} \rho_B^j \varepsilon_{t-i+j}^B.
	\end{equation*}
	Thus, since $\varepsilon_s^F$ and $\varepsilon_t^B$ are independent for all $s \in \mathbb{Z}$ and $t \in \mathbb{Z}$,
	\begin{align*}
		\varphi_Y^p (u) &= \mathbb{E} \left[ \exp \left( \textup{i} \sum_{i=0}^{p} u_i X_{t-i}^F \right) \right] \mathbb{E} \left[ \exp \left( \textup{i} \sum_{i=0}^{p} u_i X_{t-i}^B \right) \right] \\
		&= \varphi_{X^F}^p (u) \varphi_{X^B}^p (u),
	\end{align*}
	where
	\begin{align*}
		\varphi_{X^F}^p (u) &= \mathbb{E} \left[ \exp \left( \textup{i} \sum_{i=0}^{p} u_i X_{t-i}^F \right) \right] \\
		&= \mathbb{E} \left[ \exp \left( \textup{i} \sum_{i=0}^{p} u_i \left( \frac{\mu_F}{1-\rho_F} + \sum_{j=0}^{\infty} \rho_F^j \varepsilon_{t-i-j}^F \right) \right) \right] \\
		&= \exp \left( \textup{i} \frac{\mu_F}{1-\rho_F} \sum_{i=0}^{p} u_i \right) \mathbb{E} \left[ \exp \left( \textup{i} \sum_{i=0}^{p} u_i \sum_{j=0}^{\infty} \rho_F^j \varepsilon_{t-i-j}^F \right) \right] \\
		&= \exp \left( \textup{i} \frac{\mu_F}{1-\rho_F} \sum_{i=0}^{p} u_i \right) \mathbb{E} \left[ \exp \left( \textup{i} \sum_{j=0}^{p-1} \sum_{i=0}^{j} u_i \rho_F^{j-i} \varepsilon_{t-j}^F + \textup{i} \sum_{j=0}^{\infty} \sum_{i=0}^{p} u_i \rho_F^{p-i} \rho_F^j \varepsilon_{t-p-j}^F \right) \right] \\
		&= \exp \left( \textup{i} \frac{\mu_F}{1-\rho_F} \sum_{i=0}^{p} u_i \right) \prod_{j=0}^{p-1} \mathbb{E} \left[ \exp \left( \textup{i} \sum_{i=0}^{j} u_i \rho_F^{j-i} \varepsilon_{t-j}^F \right) \right] \prod_{j=0}^{\infty} \mathbb{E} \left[ \exp \left( \textup{i} \sum_{i=0}^{p} u_i \rho_F^{p-i} \rho_F^j \varepsilon_{t-p-j}^F \right) \right] \\
		&= \exp \left( \textup{i} \frac{\mu_F}{1-\rho_F} \sum_{i=0}^{p} u_i \right) \prod_{j=0}^{p-1} \exp \left( - \frac{\sigma_F^2 \left( \sum_{i=0}^{j} u_i \rho_F^{j-i} \right)^2}{2} \right) \prod_{j=0}^{\infty} \exp \left( - \frac{\sigma_F^2 \left( \sum_{i=0}^{p} u_i \rho_F^{p-i} \rho_F^j \right)^2}{2} \right) \\
		&= \exp \left( \textup{i} \frac{\mu_F}{1-\rho_F} \sum_{i=0}^{p} u_i - \frac{\sigma_F^2 \sum_{j=0}^{p-1} \left( \sum_{i=0}^{j} u_i \rho_F^{j-i} \right)^2}{2} - \frac{\frac{\sigma_F^2}{1-\rho_F^2} \left( \sum_{i=0}^{p} u_i \rho_F^{p-i} \right)^2}{2} \right)
	\end{align*}
	and, if $\alpha \neq 1$,
	\begin{align*}
		\varphi_{X^B}^p (u) &= \mathbb{E} \left[ \exp \left( \textup{i} \sum_{i=0}^{p} u_i X_{t-i}^B \right) \right] \\
		&= \mathbb{E} \left[ \exp \left( \textup{i} \sum_{i=0}^{p} u_i \sum_{j=0}^{\infty} \rho_B^j \varepsilon_{t-i+j}^B \right) \right] \\
		&= \mathbb{E} \left[ \exp \left( \textup{i} \sum_{j=0}^{p-1} \sum_{i=0}^{j} u_{p-i} \rho_B^{j-i} \varepsilon_{t-p+j}^B + \textup{i} \sum_{j=0}^{\infty} \sum_{i=0}^{p} u_{p-i} \rho_B^{p-i} \rho_B^j \varepsilon_{t+j}^B \right) \right] \\
		&= \prod_{j=0}^{p-1} \mathbb{E} \left[ \exp \left( \textup{i} \sum_{i=0}^{j} u_{p-i} \rho_B^{j-i} \varepsilon_{t-p+j}^B \right) \right] \prod_{j=0}^{\infty} \mathbb{E} \left[ \exp \left( \textup{i} \sum_{i=0}^{p} u_{p-i} \rho_B^{p-i} \rho_B^j \varepsilon_{t+j}^B \right) \right] \\
		&= \prod_{j=0}^{p-1} \exp \left( - \sigma_B^{\alpha} \left| \sum_{i=0}^{j} u_{p-i} \rho_B^{j-i} \right|^{\alpha} \left( 1 - \textup{i} \beta \textup{sign} \left( \sum_{i=0}^{j} u_{p-i} \rho_B^{j-i} \right) \tan \left( \frac{\pi \alpha}{2} \right) \right) \right) \\
		& \quad \cdot \prod_{j=0}^{\infty} \exp \left( - \sigma_B^{\alpha} \left| \sum_{i=0}^{p} u_{p-i} \rho_B^{p-i} \rho_B^j \right|^{\alpha} \left( 1 - \textup{i} \beta \textup{sign} \left( \sum_{i=0}^{p} u_{p-i} \rho_B^{p-i} \rho_B^j \right) \tan \left( \frac{\pi \alpha}{2} \right) \right) \right) \\
		&= \exp \left\{ - \sigma_B^{\alpha} \sum_{j=0}^{p-1} \left| \sum_{i=0}^{j} u_{p-i} \rho_B^{j-i} \right|^{\alpha} \left( 1 - \textup{i} \beta \textup{sign} \left( \sum_{i=0}^{j} u_{p-i} \rho_B^{j-i} \right) \tan \left( \frac{\pi \alpha}{2} \right) \right) \right. \\
		& \quad \left. - \left( \frac{\sigma_B}{(1-|\rho_{B}|^{\alpha})^{1/\alpha}} \right)^{\alpha} \left| \sum_{i=0}^{p} u_{p-i} \rho_B^{p-i} \right|^{\alpha} \left( 1 - \textup{i} \frac{1-|\rho_{B}|^{\alpha}}{1-\textup{sign}(\rho_{B}) |\rho_{B}|^{\alpha}} \beta \textup{sign} \left( \sum_{i=0}^{p} u_{p-i} \rho_B^{p-i} \right) \tan \left( \frac{\pi \alpha}{2} \right) \right) \right\}
	\end{align*}
	and, if $\alpha = 1$,
	\begin{align*}
		\varphi_{X^B}^p (u) &= \mathbb{E} \left[ \exp \left( \textup{i} \sum_{i=0}^{p} u_i X_{t-i}^B \right) \right] \\
		&= \mathbb{E} \left[ \exp \left( \textup{i} \sum_{i=0}^{p} u_i \sum_{j=0}^{\infty} \rho_B^j \varepsilon_{t-i+j}^B \right) \right] \\
		&= \mathbb{E} \left[ \exp \left( \textup{i} \sum_{j=0}^{p-1} \sum_{i=0}^{j} u_{p-i} \rho_B^{j-i} \varepsilon_{t-p+j}^B + \textup{i} \sum_{j=0}^{\infty} \sum_{i=0}^{p} u_{p-i} \rho_B^{p-i} \rho_B^j \varepsilon_{t+j}^B \right) \right] \\
		&= \prod_{j=0}^{p-1} \mathbb{E} \left[ \exp \left( \textup{i} \sum_{i=0}^{j} u_{p-i} \rho_B^{j-i} \varepsilon_{t-p+j}^B \right) \right] \prod_{j=0}^{\infty} \mathbb{E} \left[ \exp \left( \textup{i} \sum_{i=0}^{p} u_{p-i} \rho_B^{p-i} \rho_B^j \varepsilon_{t+j}^B \right) \right] \\
		&= \prod_{j=0}^{p-1} \exp \left( - \sigma_B \left| \sum_{i=0}^{j} u_{p-i} \rho_B^{j-i} \right| \left( 1 + \textup{i} \beta \textup{sign} \left( \sum_{i=0}^{j} u_{p-i} \rho_B^{j-i} \right) \frac{2}{\pi} \log \left( \left| \sum_{i=0}^{j} u_{p-i} \rho_B^{j-i} \right| \right) \right) \right) \\
		& \quad \cdot \prod_{j=0}^{\infty} \exp \left( - \sigma_B \left| \sum_{i=0}^{p} u_{p-i} \rho_B^{p-i} \rho_B^j \right| \left( 1 + \textup{i} \beta \textup{sign} \left( \sum_{i=0}^{p} u_{p-i} \rho_B^{p-i} \rho_B^j \right) \frac{2}{\pi} \log \left( \left|  \sum_{i=0}^{p} u_{p-i} \rho_B^{p-i} \rho_B^j \right| \right) \right) \right) \\
		&= \exp \left\{ - \textup{i} \beta \sigma_B \frac{2}{\pi} \frac{\rho_B \log |\rho_B|}{(1-\rho_{B})^2} \sum_{i=0}^{p} u_{p-i} \rho_B^{p-i} \right. \\
		& \quad \left. - \sigma_B \sum_{j=0}^{p-1} \left| \sum_{i=0}^{j} u_{p-i} \rho_B^{j-i} \right| \left( 1 + \textup{i} \beta \textup{sign} \left( \sum_{i=0}^{j} u_{p-i} \rho_B^{j-i} \right) \frac{2}{\pi} \log \left( \left| \sum_{i=0}^{j} u_{p-i} \rho_B^{j-i} \right| \right) \right) \right. \\
		& \quad \left. - \frac{\sigma_B}{1-|\rho_{B}|} \left| \sum_{i=0}^{p} u_{p-i} \rho_B^{p-i} \right| \left( 1 + \textup{i} \frac{1-|\rho_{B}|}{1-\textup{sign}(\rho_{B}) |\rho_{B}|} \beta \textup{sign} \left( \sum_{i=0}^{p} u_{p-i} \rho_B^{p-i} \right) \frac{2}{\pi} \log \left( \left| \sum_{i=0}^{p} u_{p-i} \rho_B^{p-i} \right| \right) \right) \right\}.
	\end{align*} 
\end{proof}

\begin{proof}[Proof of Proposition \ref{prop:Voigt}]
	The characteristic function follows from Corollary \ref{coro:cf}. We now prove that the probability density function of $Z = X + Y$ where $X \sim \mathcal{N} (\mu_X,\sigma_X^2)$ and $Y \sim \mathcal{C} (0,\sigma_Y)$ are independent is given by
	\begin{equation*}
		f_Z (z) = \frac{1}{\sqrt{2\pi} \sigma_X} \textup{Re} \left( w \left( \frac{z-\mu_X+\textup{i}\sigma_Y}{\sqrt{2}\sigma_X} \right) \right), \quad z \in \mathbb{R}.
	\end{equation*}
	The probability density function of $Z$ is given by
	\begin{align*}
		f_Z (z) &= \int_{-\infty}^{\infty} f_X (t) f_Y (z-t) \textup{d}t \\
		&= \frac{1}{\sqrt{2\pi}\sigma_X} \frac{1}{\pi\sigma_Y} \int_{-\infty}^{\infty} \exp \left( - \frac{1}{2} \left( \frac{t-\mu_X}{\sigma_X} \right)^2 \right) \frac{1}{1+\left(\frac{z-t}{\sigma_Y} \right)^2 } \textup{d}t \\
		&= \frac{1}{\sqrt{2\pi}\sigma_X} \frac{\frac{\sigma_Y}{\sqrt{2}\sigma_X}}{\pi} \frac{1}{\sqrt{2} \sigma_X} \int_{-\infty}^{\infty} \exp \left( - \left( \frac{t-\mu_X}{\sqrt{2} \sigma_X} \right)^2 \right) \frac{1}{\left( \frac{\sigma_Y}{\sqrt{2} \sigma_X} \right)^2 +\left(\frac{z-\mu_X}{\sqrt{2} \sigma_X} - \frac{t-\mu_X}{\sqrt{2} \sigma_X} \right)^2 } \textup{d}t \\
		&= \frac{1}{\sqrt{2\pi} \sigma_X} \frac{a}{\pi} \frac{1}{\sqrt{2}\sigma_X} \int_{-\infty}^{\infty} \exp \left( - \left(\frac{t-\mu_X}{\sqrt{2}\sigma_X}\right)^2 \right) \frac{1}{a^2 + \left(b-\frac{t-\mu_X}{\sqrt{2}\sigma_X}\right)^2} \textup{d}t,
	\end{align*}
	where $a = \frac{\sigma_Y}{\sqrt{2}\sigma_X}$ and $b = \frac{z-\mu_X}{\sqrt{2}\sigma_X}$. Let $u = \frac{t-\mu_X}{\sqrt{2}\sigma_X}$. Then, by integration by substitution, 
	\begin{equation*}
		f_Z (z) = \frac{1}{\sqrt{2\pi} \sigma_X} \frac{a}{\pi} \int_{-\infty}^{\infty} \exp \left( - u^2 \right) \frac{1}{a^2 + \left(b-u\right)^2} \textup{d}u,
	\end{equation*}
	where
	\begin{equation*}
		\frac{a}{\pi} \int_{-\infty}^{\infty} \exp \left( - u^2 \right) \frac{1}{a^2 + \left(b-u\right)^2} \textup{d}u = \textup{Re} \left( w \left( b+\textup{i}a \right) \right),
	\end{equation*}
	see Section 7.19 in \cite{OlverLozierBoisvertClark2010}. Thus,
	\begin{equation*}
		f_Z (z) = \frac{1}{\sqrt{2\pi} \sigma_X} \textup{Re} \left( w \left( \frac{z-\mu_X+\textup{i}\sigma_Y}{\sqrt{2}\sigma_X} \right) \right).
	\end{equation*}
\end{proof}

\section{Technical Lemmata} \label{app:lemmata}

\begin{lemmaappendix} \label{lem:app1}
	Let $(\Omega,\mathcal{F},\mathbb{P})$ be a probability space, let $(\textup{X},\mathcal{X})$ and $(\textup{Y},\mathcal{Y})$ be measurable spaces, and let $X$ and $Y$ be $\textup{X}$- and $\textup{Y}$-valued random variables on $(\Omega,\mathcal{F},\mathbb{P})$, respectively. Assume that there exist a measurable function $\gamma : \textup{X} \times \textup{Y} \rightarrow (0,\infty)$, a probability measure $\mu_X : \mathcal{X} \rightarrow [0,1]$, and a $\sigma$-finite measure $\mu_Y : \mathcal{Y} \rightarrow [0,\infty]$ such that
	\begin{equation*}
		\mathbb{E} [ f(X,Y) ] = \int_{\textup{Y}} \int_{\textup{X}} f(x,y) \gamma(x,y) \mu_X (\textup{d}x) \mu_Y (\textup{d}y)
	\end{equation*}
	for all bounded, measurable functions $f : \textup{X} \times \textup{Y} \rightarrow \mathbb{R}$. Then, the conditional distribution of $X$ given $Y$ is given by
	\begin{equation*}
		\mathbb{P} (X \in A \mid Y = y) = \frac{\int_{\textup{X}} 1_{A} (x) \gamma(x,y) \mu_X (\textup{d}x)}{\int_{\textup{X}} \gamma(x,y) \mu_X (\textup{d}x)}, \quad A \in \mathcal{X}, y \in \textup{Y}.
	\end{equation*}
\end{lemmaappendix}
\begin{proof}
	The conclusion follows if (i) for all $y \in \textup{Y}$, $A \mapsto \mathbb{P} (X \in A \mid Y = y)$ is a probability measure, (ii) for all $A \in \mathcal{X}$, $y \mapsto \mathbb{P} (X \in A \mid Y = y)$ is measurable, and (iii) for all $A \in \mathcal{X}$ and $B \in \mathcal{Y}$, $\mathbb{E} [\mathbb{P} (X \in A \mid Y) 1_{B} (Y)] = \mathbb{E} [1_{A} (X) 1_{B} (Y)]$.
	
	Condition (i): Let $y \in \textup{Y}$ be given. Then, $A \mapsto \mathbb{P} (X \in A \mid Y = y)$ is a probability measure since $\mathbb{P} (X \in \emptyset \mid Y = y) = 0$, if $(A_n)_{n \in \mathbb{N}}$ is a sequence of disjoint sets from $\mathcal{X}$, then $\mathbb{P} ( X \in \cup_{n \in \mathbb{N}} A_n \mid Y = y ) = \sum_{n=1}^{\infty} \mathbb{P} \left( X \in A_n \mid Y = y \right)$ since $1_{\cup_{n \in \mathbb{N}} A_n} (x) = \sum_{n=1}^{\infty} 1_{A_n} (x)$, and $\mathbb{P} (X \in \textup{X} \mid Y = y) = 1$. Condition (ii): Let now $A \in \mathcal{X}$ be given. Then, $y \mapsto \mathbb{P} (X \in A \mid Y = y)$ is measurable since $(x,y) \mapsto 1_{A \times \textup{Y}} (x,y)$ and $(x,y) \mapsto \gamma(x,y)$ are measurable so, by Tonelli's theorem, $y \mapsto \int_{\textup{X}} 1_{A \times \textup{Y}} (x,y) \gamma(x,y) \mu_X (\textup{d}x)$ and $y \mapsto \int_{\textup{X}} \gamma(x,y) \mu_X (\textup{d}x)$ are measurable. Condition (iii): Let $A \in \mathcal{X}$ and $B \in \mathcal{Y}$ be given. Then,
	\begin{align*}
		\mathbb{E} [\mathbb{P} (X \in A \mid Y) 1_{B} (Y)] &= \mathbb{E} \left[\frac{\int_{\textup{X}} 1_{A} (x) \gamma(x,Y) \mu_X (\textup{d}x)}{\int_{\textup{X}} \gamma(x,Y) \mu_X (\textup{d}x)} 1_{B} (Y)\right] \\
		&= \int_{\textup{Y}} \int_{\textup{X}} \frac{\int_{\textup{X}} 1_{A} (x) \gamma(x,y) \mu_X (\textup{d}x)}{\int_{\textup{X}} \gamma(x,y) \mu_X (\textup{d}x)} 1_{B} (y) \gamma(z,y) \mu_X (\textup{d}z) \mu_Y (\textup{d}y) \\
		&= \int_{\textup{Y}} \int_{\textup{X}} 1_{A} (x) 1_{B} (y) \gamma(x,y) \mu_X (\textup{d}x) \mu_Y (\textup{d}y) \\
		&= \mathbb{E} [1_{A} (X) 1_{B} (Y)].
	\end{align*}
	
	This concludes the proof.
\end{proof}

Before we state the next lemma, we introduce some standard notation, see, for instance, Chapter 2 in \cite{CappeMoulinesRyden2005}. Let $(\textup{X},\mathcal{X})$ be a measurable space, $P$ a kernel from $(\textup{X},\mathcal{X})$ to $(\textup{X},\mathcal{X})$, $\Pi$ a measure on $(\textup{X},\mathcal{X})$, and $f$ a bounded, measurable function from $\textup{X}$ to $\mathbb{R}$. Then, we define $P f : \textup{X} \rightarrow \mathbb{R}$ by
\begin{equation*}
	P f (x) := \int_{\textup{X}} f(z) P(x,\textup{d}z), \quad x \in \textup{X},
\end{equation*}
and $\Pi P : \mathcal{X} \rightarrow [0,\infty]$ by
\begin{equation*}
	\Pi P (A) := \int_{\textup{X}} P(z,A) \Pi(\textup{d}z), \quad A \in \mathcal{X}.
\end{equation*}
Note that
\begin{equation}
	\int_{\textup{X}} P f (z) \Pi (\textup{d}z) = \int_{\textup{X}} f (z) \Pi P (\textup{d}z), \label{eq:ididnotknowwhattocallthisone}
\end{equation}
see Chapter 2 in \cite{CappeMoulinesRyden2005} once again.
\begin{lemmaappendix} \label{lem:app2}
	Let $(\textup{X},\mathcal{X})$ be a measurable space, let $(P_t)_{t=1}^{T-1}$ be a sequence of kernels from $(\textup{X},\mathcal{X})$ to $(\textup{X},\mathcal{X})$, let $\Pi_T$ be a measure on $(\textup{X},\mathcal{X})$, and let $f$ be a bounded, measurable function from $\textup{X}$ to $\mathbb{R}$. Then,
	\begin{equation*}
		\int_{\textup{X}} \cdots \int_{\textup{X}} f(x_t) P_t(x_{t+1},\textup{d}x_t) \cdots P_{T-1}(x_{T},\textup{d}x_{T-1}) \Pi_T(\textup{d}x_T) = \int_{\textup{X}} f(x_t) \Pi_T P_{T-1} \cdots P_t (\textup{d}x_t),
	\end{equation*}
	where
	\begin{equation*}
		\Pi_T P_{T-1} \cdots P_t (A) = \int_{\textup{X}} \cdots \int_{\textup{X}} 1_A(x_t) P_t(x_{t+1},\textup{d}x_t) \cdots P_{T-1}(x_{T},\textup{d}x_{T-1}) \Pi_T(\textup{d}x_T), \quad A \in \mathcal{X},
	\end{equation*}
	which can be computed recursively as
	\begin{equation*}
		\Pi_T P_{T-1} \cdots P_t (A) = \int_{\textup{X}} \int_{\textup{X}} 1_A (x_t) P_t(x_{t+1},\textup{d}x_t) \Pi_T P_{T-1} \cdots P_{t+1} (\textup{d}x_{t+1}).
	\end{equation*}
\end{lemmaappendix}
\begin{proof}
	We have that
	\begin{align*}
		& \int_{\textup{X}} \cdots \int_{\textup{X}} f(x_t) P_t(x_{t+1},\textup{d}x_t) \cdots P_{T-1}(x_{T},\textup{d}x_{T-1}) \Pi_T(\textup{d}x_T) \\
		&= \int_{\textup{X}} \cdots \int_{\textup{X}} P_t f(x_{t+1}) P_{t+1}(x_{t+2},\textup{d}x_{t+1}) \cdots P_{T-1}(x_{T},\textup{d}x_{T-1}) \Pi_T(\textup{d}x_T) \\
		& \ \, \vdots \\
		&= \int_{\textup{X}} P_t \cdots P_{T-1} f(x_T) \Pi_T(\textup{d}x_T) \\
		&= \int_{\textup{X}} P_t \cdots P_{T-2} f(x_{T-1}) \Pi_T P_{T-1} (\textup{d}x_{T-1}) \\
		& \ \, \vdots \\
		&= \int_{\textup{X}} f(x_t) \Pi_T P_{T-1} \cdots P_t (\textup{d}x_t),
	\end{align*}
	where we have used Equation \eqref{eq:ididnotknowwhattocallthisone} repeatedly, and thus that
	\begin{equation*}
		\Pi_T P_{T-1} \cdots P_t (A) = \int_{\textup{X}} \cdots \int_{\textup{X}} 1_A(x_t) P_t(x_{t+1},\textup{d}x_t) \cdots P_{T-1}(x_{T},\textup{d}x_{T-1}) \Pi_T(\textup{d}x_T).
	\end{equation*}
\end{proof}

\section{More Details on the Other Examples} \label{appendix:otherexamples}

\subsection{Example \ref{ex:un}}

We show that $((X_t^F,X_t^B,Y_t))_{t=1}^{T}$ satisfies the conditions in Definition \ref{def:cncssm}. Condition (i): We have that
\begin{align*}
	\mathbb{P} (X_{t}^F \in A^F \mid X_{t-1}^F,...,X_{1}^F) &= \mathbb{E} [1_{A^F} (X_{t}^F) \mid X_{t-1}^F,...,X_{1}^F] \\
	&= \mathbb{E} [1_{A^F} (\mu_F+\rho_F X_{t-1}^{F}+\varepsilon_{t}^F) \mid X_{t-1}^F,...,X_{1}^F] \\
	&= \int_{\mathbb{R}} 1_{A^F} (\mu_F+\rho_F X_{t-1}^{F}+z) \frac{1}{\sigma_F} f_{\mathcal{N}} \left( \frac{z}{\sigma_F} \right) \lambda (\textup{d}z) \\
	&= \int_{\mathbb{R}} 1_{A^F} (z) \frac{1}{\sigma_F} f_{\mathcal{N}} \left( \frac{z-\mu_F-\rho_F X_{t-1}^{F}}{\sigma_F} \right) \lambda (\textup{d}z) \\
	&= P^F(X_{t-1}^{F},A^F)
\end{align*}
and that
\begin{equation*}
	\mathbb{P} (X_{1}^F \in A^F) = \int_{\mathbb{R}} 1_{A^F} (z) \frac{1}{\frac{\sigma_F}{\sqrt{1-\rho_F^2}}} f_{\mathcal{N}} \left( \frac{z-\frac{\mu_F}{1-\rho_F}}{\frac{\sigma_F}{\sqrt{1-\rho_F^2}}} \right) \lambda(\textup{d}z) = \Pi^F (A^F).
\end{equation*}
Condition (ii): We have that
\begin{align*}
	\mathbb{P} (X_{t}^B \in A^B \mid X_{t+1}^B,...,X_{T}^B) &= \mathbb{E} [1_{A^B} (X_{t}^B) \mid X_{t+1}^B,...,X_{T}^B] \\
	&= \mathbb{E} [1_{A^B} (\rho_B X_{t+1}^{B}+\varepsilon_{t}^B) \mid X_{t+1}^B,...,X_{T}^B] \\
	&= \int_{\mathbb{R}} 1_{A^B} (\rho_B X_{t+1}^{B}+z) \frac{1}{\sigma_B} f_{\mathcal{C}} \left( \frac{z}{\sigma_B} \right) \lambda (\textup{d}z) \\
	&= \int_{\mathbb{R}} 1_{A^B} (z) \frac{1}{\sigma_B} f_{\mathcal{C}} \left( \frac{z-\rho_B X_{t+1}^{B}}{\sigma_B} \right) \lambda (\textup{d}z) \\
	&= P^B(X_{t+1}^{B},A^B)
\end{align*}
and that
\begin{equation*}
	\mathbb{P} (X_{T}^B \in A^B) = \int_{\mathbb{R}} 1_{A^B} (z) \frac{1}{\frac{\sigma_B}{1-|\rho_B|}} f_{\mathcal{C}} \left( \frac{z}{\frac{\sigma_B}{1-|\rho_B|}} \right) \lambda(\textup{d}z) = \Pi^B (A^B).
\end{equation*}
Condition (iii): Note that
\begin{align*}
	& \mathbb{P} (X_{1}^F \in A_1^F,...,X_{T}^F \in A_T^F,X_{1}^B \in A_1^B,...,X_{T}^B \in A_T^B) \\
	&= \mathbb{E} \left[ \prod_{t=1}^{T} 1_{A_t^F} (X_{t}^F) \prod_{t=1}^{T} 1_{A_t^B} (X_{t}^B) \right] \\
	&= \mathbb{E} \left[ \prod_{t=1}^{T} 1_{A_t^F} \left( \frac{1-\rho_F^{t-1}}{1-\rho_F} \mu_F + \rho_F^{t-1} X_{1}^F + \sum_{i=0}^{t-2} \rho_F^{i} \varepsilon_{t-i}^F \right) \prod_{t=1}^{T} 1_{A_t^B} \left( \rho_B^{T-t} X_{T}^B + \sum_{i=0}^{T-t-1} \rho_B^{i} \varepsilon_{t+i}^{B} \right) \right] \\
	&= \mathbb{E} \left[ \prod_{t=1}^{T} 1_{A_t^F} \left( \frac{1-\rho_F^{t-1}}{1-\rho_F} \mu_F + \rho_F^{t-1} X_{1}^F + \sum_{i=0}^{t-2} \rho_F^{i} \varepsilon_{t-i}^F \right) \right] \mathbb{E} \left[ \prod_{t=1}^{T} 1_{A_t^B} \left( \rho_B^{T-t} X_{T}^B + \sum_{i=0}^{T-t-1} \rho_B^{i} \varepsilon_{t+i}^{B} \right) \right] \\
	&= \mathbb{E} \left[ \prod_{t=1}^{T} 1_{A_t^F} (X_{t}^F) \right] \mathbb{E} \left[ \prod_{t=1}^{T} 1_{A_t^B} (X_{t}^B) \right] \\
	&= \mathbb{P} (X_{1}^F \in A_1^F,...,X_{T}^F \in A_T^F) \mathbb{P} (X_{1}^B \in A_1^B,...,X_{T}^B \in A_T^B).
\end{align*}
Condition (iv): Finally, we have that
\begin{align*}
	\mathbb{P} (Y_1 \in B_1,...,Y_T \in B_T \mid X_{1}^F,X_{1}^B,...,X_{T}^F,X_{T}^B) &= \mathbb{E} \left[ \prod_{t=1}^{T} 1_{B_t} (Y_t) \mid X_{1}^F,X_{1}^B,...,X_{T}^F,X_{T}^B \right] \\
	&= \mathbb{E} \left[ \prod_{t=1}^{T} 1_{B_t} (X_{t}^F+X_{t}^B) \mid X_{1}^F,X_{1}^B,...,X_{T}^F,X_{T}^B \right] \\
	&= \prod_{t=1}^{T} 1_{B_t} (X_{t}^F+X_{t}^B) \\
	&= \prod_{t=1}^{T} \Phi ((X_{t}^F,X_{t}^B),B_t).
\end{align*}

\subsection{Example \ref{ex:deux}}

We show that $((\bar{X}_t^F,\bar{X}_t^B,Y_t))_{t=1}^{T}$ satisfies the conditions in Definition \ref{def:cncssm} as above. Condition (i): We have that
\begin{align*}
	\mathbb{P} (\bar{X}_{t}^F \in A^F \mid \bar{X}_{t-1}^F,...,\bar{X}_{1}^F) &= \mathbb{E} [1_{A_1^F} (\bar{X}_{1,t}^F) 1_{A_2^F} (\bar{X}_{2,t}^F) \mid \bar{X}_{t-1}^F,...,\bar{X}_{1}^F] \\
	&= \mathbb{E} [1_{A_1^F} (\mu_F+\rho_{1F} \bar{X}_{1,t-1}^{F}+\rho_{2F} \bar{X}_{2,t-1}^{F}+\varepsilon_{t}^F) 1_{A_2^F} (\bar{X}_{1,t-1}^{F}) \mid \bar{X}_{t-1}^F,...,\bar{X}_{1}^F] \\
	&= \int_{\mathbb{R}} 1_{A_1^F} (\mu_F+\rho_{1F} \bar{X}_{1,t-1}^{F}+\rho_{2F} \bar{X}_{2,t-1}^{F}+z) \frac{1}{\sigma_F} f_{\mathcal{N}} \left( \frac{z}{\sigma_F} \right) \lambda (\textup{d}z) 1_{A_2^F} (\bar{X}_{1,t-1}^{F}) \\
	&= \int_{\mathbb{R}} 1_{A_1^F} (z) \frac{1}{\sigma_F} f_{\mathcal{N}} \left( \frac{z-\mu_F-\rho_{1F} \bar{X}_{1,t-1}^{F}-\rho_{2F} \bar{X}_{2,t-1}^{F}}{\sigma_F} \right) \lambda (\textup{d}z) 1_{A_2^F} (\bar{X}_{1,t-1}^{F}) \\
	&= P^F(\bar{X}_{t-1}^{F},A^F)
\end{align*}
and that
\begin{equation*}
	\mathbb{P} (\bar{X}_{1}^F \in A^F) = \int_{A^F} \det(V)^{-1/2} f_{\mathcal{N}_2} \left( V^{-1/2} (z-M) \right) \lambda_2(\textup{d}z) = \Pi^F (A^F).
\end{equation*}
Condition (ii): We have that
\begin{align*}
	\mathbb{P} (\bar{X}_{t}^B \in A^B \mid \bar{X}_{t+1}^B,...,\bar{X}_{T}^B) &= \mathbb{E} [1_{A_1^B} (\bar{X}_{1,t}^B) 1_{A_2^B} (\bar{X}_{2,t}^B) \mid \bar{X}_{t+1}^B,...,\bar{X}_{T}^B] \\
	&= \mathbb{E} [1_{A_1^B} (\rho_{1B} \bar{X}_{1,t+1}^{B}+\rho_{2B} \bar{X}_{2,t+1}^{B}+\varepsilon_{t}^B) 1_{A_2^B} (\bar{X}_{1,t+1}^B) \mid \bar{X}_{t+1}^B,...,\bar{X}_{T}^B] \\
	&= \int_{\mathbb{R}} 1_{A_1^B} (\rho_{1B} \bar{X}_{1,t+1}^{B}+\rho_{2B} \bar{X}_{2,t+1}^{B}+z) \frac{1}{\sigma_B} f_{\mathcal{C}} \left( \frac{z}{\sigma_B} \right) \lambda (\textup{d}z) 1_{A_2^B} (\bar{X}_{1,t+1}^B) \\
	&= \int_{\mathbb{R}} 1_{A_1^B} (z) \frac{1}{\sigma_B} f_{\mathcal{C}} \left( \frac{z-\rho_{1B} \bar{X}_{1,t+1}^{B}-\rho_{2B} \bar{X}_{2,t+1}^{B}}{\sigma_B} \right) \lambda (\textup{d}z) 1_{A_2^B} (\bar{X}_{1,t+1}^B) \\
	&= P^B(\bar{X}_{t+1}^{B},A^B)
\end{align*}
and that
\begin{equation*}
	\mathbb{P} (\bar{X}_{T}^B \in A^B) = \mathbb{E} [1_{A_1^B} (\bar{X}_{1,T}^B) 1_{A_2^B} (\bar{X}_{2,T}^B)] = \mathbb{E} [1_{A_1^B} (\bar{x}_{1,T}^B) 1_{A_2^B} (\bar{x}_{2,T}^B)] = 1_{A_1^B} (\bar{x}_{1,T}^B) 1_{A_2^B} (\bar{x}_{2,T}^B) = \Pi^B (A^B).
\end{equation*}
Condition (iii): This can be shown following the same lines as above. Condition (iv): Finally, we have that
\begin{align*}
	\mathbb{P} (Y_1 \in B_1,...,Y_T \in B_T \mid \bar{X}_{1}^F,\bar{X}_{1}^B,...,\bar{X}_{T}^F,\bar{X}_{T}^B) &= \mathbb{E} \left[ \prod_{t=1}^{T} 1_{B_t} (Y_t) \mid \bar{X}_{1}^F,\bar{X}_{1}^B,...,\bar{X}_{T}^F,\bar{X}_{T}^B \right] \\
	&= \mathbb{E} \left[ \prod_{t=1}^{T} 1_{B_t} (\bar{X}_{1,t}^F+\bar{X}_{1,t}^B+\varepsilon_t^Y) \mid \bar{X}_{1}^F,\bar{X}_{1}^B,...,\bar{X}_{T}^F,\bar{X}_{T}^B \right] \\
	&= \prod_{t=1}^{T} \int_{\mathbb{R}} 1_{B_t} (\bar{X}_{1,t}^F+\bar{X}_{1,t}^B+z) \frac{1}{\sigma_Y} f_{\mathcal{N}} \left( \frac{z}{\sigma_Y} \right) \lambda (\textup{d}z) \\
	&= \prod_{t=1}^{T} \int_{\mathbb{R}} 1_{B_t} (z) \frac{1}{\sigma_Y} f_{\mathcal{N}} \left( \frac{z-\bar{X}_{1,t}^F-\bar{X}_{1,t}^B}{\sigma_Y} \right) \lambda (\textup{d}z) \\
	&= \prod_{t=1}^{T} \Phi ((\bar{X}_{t}^F,\bar{X}_{t}^B),B_t).
\end{align*}

\end{document}